\documentclass[a4paper, 11pt]{article}

\usepackage[left=2.5cm,right=2.5cm,top=3cm,bottom=3cm]{geometry} 
\usepackage[utf8]{inputenc}
\usepackage{amsmath,amssymb,amsthm,amscd,amstext,mathtools,mathrsfs}

\usepackage{enumitem}
\usepackage[sort&compress,numbers,square]{natbib}
\usepackage{tikz}
\usepackage{tikz-cd}
\usepackage{hyperref}
\usepackage{orcidlink}
\hypersetup{
	unicode,
	colorlinks,
	urlcolor=blue, 
	citecolor=blue,
	linkcolor=black, 
}

\newtheorem{theorem}{Theorem}
\numberwithin{theorem}{section}
\newtheorem{corollary}[theorem]{Corollary}
\newtheorem{lemma}[theorem]{Lemma}
\newtheorem{proposition}[theorem]{Proposition}

\theoremstyle{definition}
\newtheorem{definition}[theorem]{Definition}
\newtheorem{example}[theorem]{Example}
\newtheorem{remark}[theorem]{Remark}

\newcommand{\supp}{\mathrm{supp}}
\newcommand{\rat}{\mathbb{Q}}

\newcommand{\vf}{\varphi}
\newcommand{\U}{\mathscr{U}}
\newcommand{\h}{\mathscr{X}}
\newcommand{\hy}{\mathscr{Y}}
\newcommand{\sC}{\mathscr{C}}
\newcommand{\E}{\mathcal{E}}
\newcommand{\ha}{\overset{\leftharpoonup}}
\newcommand{\de}{\mathrm{d}}
\newcommand{\rn}{\mathrm{nrd}}
\newcommand{\ts}{\mathtt{s}}

\newcommand{\pqal}{\mathbb{H}_p} 
\newcommand{\mpqal}{\mathbf{H}_p} 
\newcommand{\pqais}{\theta_p} 
\newcommand{\qal}{\mathbb{H}} 
\newcommand{\mqal}{\mathbf{H}}  
\newcommand{\pqrot}{\kappa_p}   
\newcommand{\qrot}{\kappa}
\newcommand{\so}{\mathrm{SO}}
\newcommand{\un}{\mathrm{U}}
\newcommand{\Os}{\mathcal{O}}
\newcommand{\Oinv}{\Omega_{\mathrm{inv}}}

\newcommand{\I}{\mathbf{i}}
\newcommand{\J}{\mathbf{j}}  
\newcommand{\K}{\mathbf{k}}

\allowdisplaybreaks 

\title{Invariant measures on $p$-adic Lie groups: the $p$-adic quaternion algebra and the Haar integral on the $p$-adic rotation groups}
\author{Paolo Aniello\,\orcidlink{0000-0003-4298-8275},$^{1,2}$\thanks{Email: paolo.aniello@na.infn.it\hspace{3cm} $~^\mathsection$Email: vincenzo.parisi@unicam.it
} \and Sonia L'Innocente\,\orcidlink{0000-0002-9224-7451},$^{3}$\thanks{Email: sonia.linnocente@unicam.it \hspace{2.31cm} $~^\text{\textparagraph}$Email: ilaria.svampa@unicam.it 
} \and Stefano Mancini\,\orcidlink{0000-0002-3797-3987},$^{3,4}$\thanks{Email: stefano.mancini@unicam.it \hspace{2.39cm} $~^\text{\textbardbl}$Email: andreas.winter@uab.cat
} \and Vincenzo Parisi\,\orcidlink{0000-0001-9563-6471},$^{3,4\mathsection}$ \and Ilaria Svampa\,\orcidlink{0000-0002-1389-0319},$^{3,4,5^\text{\textparagraph}}$  \and Andreas Winter\,\orcidlink{0000-0001-6344-4870},$^{5,6,7\text{\textbardbl}}$}

\date{\normalsize
	$^1$\textit{Dipartimento di Fisica  ``Ettore Pancini'', Universit\`a di Napoli ``Federico II'',\\
Complesso Universitario di Monte S.~Angelo, via Cintia, I-80126 Napoli, Italy}\\%
	\vspace{3mm}
	$^2$\textit{Istituto Nazionale di Fisica Nucleare, Sezione di Napoli,\\
Complesso Universitario di Monte S.~Angelo, via Cintia, I-80126 Napoli, Italy}\\%
	\vspace{3mm}
	$^3$\textit{School of Science and Technology, University of Camerino,\\
via Madonna delle Carceri, 9, Camerino, I-62032, Italy}\\%
	\vspace{3mm}
	$^4$\textit{Istituto Nazionale di Fisica Nucleare, Sezione di Perugia,\\
via A.~Pascoli, I-06123 Perugia, Italy}\\%
	\vspace{3mm}
	$^5$\textit{Departament de F\'isica: Grup d'Informaci\'o Qu\`antica,\\
              Universitat Aut\`onoma de Barcelona, ES-08193 Bellaterra (Barcelona), Spain}\\%
              \vspace{3mm}
    $^6$\textit{ICREA---Instituci\'o Catalana de la Recerca i Estudis Avan\c{c}ats, \\
              Pg. Llu\'is Companys, 23, ES-08010 Barcelona, Spain}\\
    \vspace{3mm}
    $^7$ \textit{Institute for Advanced Study, Technische Universit\"at M\"unchen,\\ 
Lichtenbergstra{\ss}e 2a, D-85748 Garching, Germany}
}

\begin{document}
	\maketitle

	\begin{abstract}
We provide a general expression of the Haar measure --- that is, the essentially unique translation-invariant measure ---
on a $p$-adic Lie group. We then argue that this measure can be regarded as the measure naturally induced by the invariant volume
form on the group, as it happens for a standard Lie group over the reals. As an important application, we next consider the
problem of determining the Haar measure on the $p$-adic special orthogonal groups in dimension two, three and four
(for every prime number $p$). In particular, the Haar measure on $\so(2,\rat_p)$  is obtained by a direct application
of our general formula. As for $\so(3,\rat_p)$ and $\so(4,\rat_p)$, instead, we show that Haar integrals on these two groups can conveniently 
be lifted to Haar integrals on certain $p$-adic Lie groups 
from which the special orthogonal groups are obtained as 
quotients. This construction involves a suitable quaternion algebra over the field $\rat_p$ and is reminiscent of the quaternionic
realization of the real rotation groups. Our results should pave the way to the development of harmonic analysis on the $p$-adic
special orthogonal groups, with potential applications in $p$-adic quantum mechanics and in the recently proposed $p$-adic quantum information theory.

		\vspace{5mm}
		\noindent\textbf{Keywords:}{
		\textit{Locally compact group; Haar measure; $p$-adic Lie group; quaternion algebra}.}

	\end{abstract}

	\tableofcontents
	
\section{Introduction}
During the last decades of the XX century, a new branch of mathematical physics, the so-called
\emph{$p$-adic mathematical physics}, has been developed as an effort to find a non-Archimedean
approach to space-time and string dynamics at the Planck scale~\cite{brekke1993,volovich1987,volovich87,arefeva1991,khrennikov1991pp}.
Since then, various $p$-adic quantum mechanical models have been considered and
studied~\cite{vladimirov1994,vladimirov1989,ruelle1989,meurice1989,albeverio97,vladimirov89p,vladimirov90,zelenov89,khrennikov2013,
albeverio1996,khrennikov1991p,khrennikov1990,volovich2010,dragovich2017},
and several applications to quantum field theory and string theory have
been proposed~\cite{parisi1988,Khrennikov1990p,albeverio2009,khrennikov1993p,albeverio1997p}. Although the original focus of these
theories was on the foundational aspects, further investigation has revealed new surprising applications, especially in the context
of statistical and condensed matter physics. For instance, it has been observed that the natural ultrametric hierarchical structure of $p$-adic
numbers makes them suitable for the description of the dynamics of chaotic and disordered systems. From this observation,
M\'ezard, Parisi and their collaborators have shown, in the first half of the 1980s, that the ground state of the spin glasses
exhibits a natural (non-Archimedean) ultrametric structure~\cite{mezard1984,rammal1986,PS00}.

In more recent years, further intriguing applications of $p$-adic numbers  have emerged, well beyond their original
mathematical and physical context. Indeed, $p$-adic numbers have proved to be a valuable tool in solving issues related to
algebraic dynamical systems, image analysis, compression of information, image recognition, cryptography and computer science
(see~\cite{anashin2009}, and references therein). Even more recently, there has been an increasing interest in the potential
applications of the field of $p$-adic numbers to quantum information theory, as well~\cite{zelenov2020,aniello2022,our2}.
This interest stems from the unique properties of $p$-adic numbers, that may provide new solutions to challenging problems
in quantum information science. E.g., it has been observed that the $p$-adic numbers can be profitably used in the construction
of mutually unbiased bases (MUBs), for any Hilbert space dimension~\cite{zelenov2022}.

As a first step in establishing the foundations of a $p$-adic theory of quantum information, it has been argued~\cite{our2} that a suitable model
of a $p$-adic qubit can be obtained by resorting to two-dimensional irreducible projective representations of the group of rotations
on the configuration space $\mathbb{Q}_p^3$ (for an alternative `purely $p$-adic' approach to the qubit also see~\cite{aniello2023}).
The special orthogonal groups over the $p$-adic fields can be defined through quadratic forms over $\rat_p$. Unlike the real case, however,
definite (i.e., representing the zero trivially) quadratic forms over $\rat_p$ exist only in dimension two, three and four~\cite{serre2012course}.
The resulting symmetry groups $\so(2,\rat_p)$, $\so(3,\rat_p)$ and $\so(4,\rat_p)$ are the only compact $p$-adic special orthogonal groups. In particular, $\so(3,\rat_p)$ can be
thought of as the group of rotations on $\mathbb{Q}_p^3$, and its geometrical features have been explored in~\cite{our1st}. The compactness of
the aforementioned groups entails that all their irreducible unitary representations occur (and can be studied) as subrepresentations of the
regular representation, according to the celebrated \emph{Peter-Weyl theorem}~\cite{folland2016course}.

Now, the study of the regular representation of compact groups --- in particular, the formulation and the application of Schur's
orthogonality relations --- as well as several other fundamental issues of abstract  harmonic analysis, involve the \emph{Haar measure} on
such groups, namely, the essentially unique (say, left) invariant Radon measure, or, regarding such a measure as a \emph{functional}~\cite{folland2016course},
the \emph{Haar integral}. More generally, the irreducible --- in general, \emph{projective} --- representations of compact groups are
\emph{square integrable} (see~\cite{AnielloSDP,AnielloPR,AnielloSP,AnielloSIR}, and references therein),
and thus satisfy suitable orthogonality relations, where, once again, the Haar measure is involved.
Still another class of problems where this measure plays a central role, is related to the `phase-space' formulation of quantum mechanics~\cite{FollandHAPS}.
Here, the phase space appears \emph{in quotes} for a two-fold reason: first, because a \emph{$p$-adic model} of
phase space is what we have in mind; second, because the usual group of translations on phase space (with its genuinely projective representations)
is replaced with a locally compact group --- e.g., with a compact $p$-adic Lie group --- admitting square integrable
representations. Such representations allow one to define \emph{generalized Wigner transforms} mapping quantum-mechanical operators
into complex functions on the relevant group~\cite{AnielloSDP,AnielloPR,AnielloSP,AnielloSIR}.

In the present work --- as a first point of an ideal program devoted to the study of harmonic analysis on
the compact $p$-adic special orthogonal groups, and to its applications to quantum information science ---
we face the problem of describing the Haar integral on the $p$-adic Lie groups $\so(2,\rat_p)$, $\so(3,\rat_p)$ and $\so(4,\rat_p)$,
for every prime $p$. Whereas the Haar measure on standard special orthogonal groups over the reals has been extensively studied
using different approaches, the corresponding $p$-adic problem seems to be (to the best of our knowledge) still unexplored.

Our strategy to deal with such a problem articulates in two main steps:

\begin{enumerate}

\item On the one hand, we derive a general formula for the invariant measure on a generic $p$-adic Lie group.
Our construction relies on the existence of a suitable atlas of \emph{mutually disjoint} charts on such a group, which allows one
to express its Haar measure in the local $\rat_p^n$-coordinates by exploiting the change-of-variables formula
for integrals on $\rat_p^n$. Precisely, we first obtain a \emph{quasi}-invariant measure on the group. At this point,
we observe that every quasi-invariant measure on a (in general, locally compact) group immediately yields a Haar
measure. This method is tailored on the peculiar properties of a $p$-adic Lie group, but we next show that
our result can be interpreted within the \emph{invariant volume form} approach to the Haar integral usually
adopted for standard Lie groups over the reals.

\item We then observe that, on the other hand, a \emph{direct} `brute force' application of the previous general approach may not be
very \emph{practical} or \emph{convenient}, depending on the (more or less) manageable parametrization of the group one
has to work with. In concrete applications, it is often more convenient to exploit a realization of the group one is interested in
as a suitable \emph{quotient} group $X=G/H$, where $G$ is some suitable $p$-adic Lie group and $H$ is a closed subgroup of $G$.
Now, if the approach outlined in the previous step provides us with a convenient expression of the Haar integral on $G$,
then we can simply \emph{lift} the Haar integral living `downstairs' on $X$ to an integral `upstairs' on $G$. This nice lifting strategy
for computing the Haar integrals relies on the so-called \emph{Weil-Mackey-Bruhat} formula~\cite{folland2016course,reiter2000,doran88}.

\end{enumerate}

A direct application of the general formula mentioned in the first point above easily yields the Haar measure on $\so(2,\rat_p)$.
In dimensions three and four, instead, we find it convenient to introduce a quaternion algebra over $\rat_p$ first.
We can then realize $\so(3,\rat_p)$ and $\so(4,\rat_p)$ as suitable quotient groups and next apply the lifting strategy
of the Haar integrals outlined in the second point above. This approach is reminiscent of the quaternionic realization
of the standard rotation group $\so(3)$ the reader may be familiar with.

The structure of the paper is as follows. In Section~\ref{sec2}, we collect the basic notions and tools which will be used throughout the remaining sections of the paper. Specifically, in Subsection~\ref{sec2n} we recall some basic facts concerning the Haar measure on locally compact groups and the lifts of Haar integrals on quotient groups.  In Subsection~\ref{subsec.2.1}, we discuss $p$-adic manifolds and introduce the notion of a $p$-adic Lie group, before delving into the specific class of $p$-adic special orthogonal groups, in Subsection~\ref{subsec.3.2n}. Section~\ref{cosHaarmes} deals with a general construction of the Haar measure on a $p$-adic Lie group, eventually showing that it naturally coincides with the measure associated with the (maximal-rank) invariant differential form defined on the group. Section~\ref{sec5} is devoted to the applications of the (previously constructed) theory to the $p$-adic special orthogonal groups in dimension two, three and four. Specifically, in Subsection~\ref{subsec.5.1}, we derive the Haar measure on $\so(2,\rat_p)$. In Subsection~\ref{sec:quatalp}, we explicitly construct, for any prime number $p$, the $p$-adic quaternion algebra and, in Subsection~\ref{sec:relquatrotp}, we highlight its relation with the elements of $\so(3,\rat_p)$ and $\so(4,\rat_p)$. Then, in Subsections~\ref{subsec:haarso3p} and~\ref{subsec:haarso4p} we construct the Haar integrals on $\so(3,\rat_p)$ and $\so(4,\rat_p)$ by exploiting the suitable `lifting strategy’ and, hence, by realizing them as Haar integrals on specific subgroups of the $p$-adic quaternion algebra. Finally, in Section~\ref{sec6},
conclusions are drawn, with a quick glance at future prospects.

\section{Basic notions and tools}\label{sec2}
In this section, we collect some  basic results and tools which will be relevant for all our later derivations. We begin by recalling the notion of Haar measure on a locally compact group. Then, we introduce the $p$-adic Lie groups and we consider, in particular, the class of $p$-adic special orthogonal groups. Finally, we provide a brief outline of integration theory on a $p$-adic manifold.


\subsection{The Haar measure on a locally
compact group and the lifts of Haar integrals}\label{sec2n}
Let $G$ be a locally compact (Hausdorff
) topological group; in short, a LC group. By a  \emph{left} (resp.\ \emph{right}) \emph{Haar measure} $\mu$ on $G$ we mean a non-zero \emph{Radon measure} for which the following condition holds: 
\begin{equation}\label{eq:condleftrightinv}
\mu(g\E)=\mu(\E)\quad (\text{resp.}\, \mu(\E g)=\mu(\E)),
\end{equation}
for every Borel set $\E\subset G$, and $g\in G$~\cite{folland2016course, hewitt1979}. We refer to~\eqref{eq:condleftrightinv} as to the \emph{left-invariance} (resp.\ \emph{right-invariance}) property of the measure. 

It is worth recalling a remarkable characterization of the left (resp.\ right) Haar measure provided by a suitable left- (right-)invariance condition for a class of functionals on $\mathrm{C}_c(G)$ --- the algebra of compactly-supported continuous complex-valued functions on $G$~\cite{folland2016course,folland99}.
\begin{remark}
We are adopting the convention that the \emph{support}, $\supp(f)$, of a continuous function $f$ is the closure of the open set $\{g\in G\mid f(g)\in\mathbb{C}\setminus\{0\}\}$.
\end{remark}
Let $\mu$ be a fixed Radon measure on a LC group $G$. The map defined as
\begin{equation}\label{funconG}
\mathrm{C}_c(G)\ni f\mapsto I(f)\coloneqq\int_G f(g)\de\mu(g)\in \mathbb{C}
\end{equation}
is easily seen to be a positive linear functional on $\mathrm{C}_c(G)$. On the other hand, the celebrated Riesz Representation Theorem (cf.\ Theorem $7.2$ in~\cite{folland99}) assures that for \emph{every} positive linear functional on $\mathrm{C}_c(G)$, there is a \emph{unique} Radon measure $\mu$ on $G$ such that $I$ is represented as in~\eqref{funconG}. Exploiting this correspondence, a Radon measure $\mu$ is a left Haar measure iff the associated functional is left-invariant, i.e., iff the condition
\begin{equation}\label{lefinv}
\int_{G}(L_{h}f)(g)\de\mu(g)=\int_{G}f(g)\de\mu(g)
\end{equation}
holds for every $f\in C_{\mathrm{c}}(G)$. Here, the map $L_h$, for $h\in G$, of \emph{left translation} on $\mathrm{C}_c(G)$ is defined as $(L_h f)(g)\coloneqq f(h^{-1}g)$. By defining the \emph{right translation} via $(R_h f)(g) = f(gh)$, we capture analogously right-invariance of the measure. In what follows, whenever $\mu$ is a Haar measure on $G$, we will refer to the integral in the r.h.s.\ of~\eqref{funconG} as to the \emph{Haar integral} associated with $\mu$.

It is a well known result (see, e.g., Theorem $2.10$ and $2.20$ in~\cite{folland2016course}) that any LC group admits an \emph{essentially uniquely defined} Haar measure. In particular, if $\mu$ and $\nu$ are left Haar measures on $G$, then there exists  $c\in \mathbb{R}^+_{\ast}$ such that $\mu=c\nu$. If $G$ is a LC group, its left and right Haar measures are related via the so-called \emph{modular function} $\Delta\colon G\rightarrow \mathbb{R}^+_{\ast}$~\cite{folland2016course}. In the case where $\Delta\equiv 1$ (as it happens for abelian and compact groups), $G$ is called \emph{unimodular}, meaning that left and right Haar measures coincide.

\begin{remark}\label{rem.finmes}
A locally compact group $G$ has \emph{finite} left (and right) Haar measure $\mu$ if and only if it is compact~\cite{folland2016course, gaal}; in this case, it is possible (and customary) to normalize the Haar measure in such a way that $\mu(G)=1$. 
\end{remark}

\begin{example}[Haar measure on $\rat_p$]\label{exa.2.2}
The (additive) group of the field of $p$-adic numbers $\rat_p$ ($p\in\mathbb{N}$ prime) is a LC group once endowed with its standard \emph{ultrametric} topology (namely, the topology induced by the non-trivially valued, \emph{non-Archimedean absolute value} $|\,\cdot\,|_p$ on $\mathbb{Q}_p$). Therefore, it admits a left Haar measure $\lambda$. Since $(\rat_p,+)$ is abelian (hence, unimodular), $\lambda$ is right-invariant as well, i.e., 
\begin{equation}\label{measureonqp}
\lambda(\E+x)=\lambda(\E)=\lambda(x+\E)
\end{equation}
holds for every Borel subset $\E$ in $\mathcal{B}_{\mathbb{Q}_p}$, and any $x\in\rat_p$. Since the subring $\mathbb{Z}_p$ of \emph{$p$-adic integers} is a compact subset of $\rat_p$, we can normalize $\lambda$ by setting
\begin{equation}\label{norm.cond}
\lambda(\mathbb{Z}_p)=1.
\end{equation}
It is now not difficult to explicitly construct the measure $\lambda$. Indeed, let $\overline{B}(r,x_0):=\{x\in\mathbb{Q}_p\mid |x-x_0|_p\leq r\}$ be a ball centred in $x_0\in\mathbb{Q}_p$ of radius $r\in\mathbb{Z}_{>0}$. Since $\overline{B}(1,0)=\mathbb{Z}_p$, owing to the invariance condition~\eqref{measureonqp} and the normalization~\eqref{norm.cond}, we get
$\lambda\big(\overline{B}(1,x)\big)=1$ for every $x\in \mathbb{Q}_p$. Moreover, the topological features of $\rat_p$ --- i.e., any ball of radius $p^k$, $k>0$, is a \emph{disjoint} union of $p^k$ balls of radius $1$ --- also entail that $\lambda\left(\overline{B}(p^k,x)\right)=p^k$ for every $k\in\mathbb{Z}$, $x\in\mathbb{Q}_p$. Hence, we get to the conclusion that the measure of every Borel set $\E$ of $\mathbb{Q}_p$ is given by
\begin{equation}\label{eq:haarQp}
\lambda(\E)=\inf\left\{\sum_{j\geq1}p^{m_j}\mid \E\subset\bigcup_{j\geq1}\overline{B}(p^{m_j},x_j)\right\},
\end{equation}
analogous to the formula for the Lebesgue measure on the real line. 
\end{example}

\begin{example}\label{exa.2.3}
The group $\rat_p^n=\rat_p\times\ldots\times\rat_p \,\,\mbox{($n$-times)}$, endowed with the product topology, has a natural structure of (additive) LC group; hence, it admits a left (and right) Haar measure. To find it explicitly, it is enough to observe that, being $\rat_p$ a second countable LC group, there is no distinction between the standard product of measures and the Radon product (see §2.2 in~\cite{folland2016course}). Therefore, the Haar measure on $\rat_p^n$ is provided by the $n$-times product of the Haar measure on $\rat_p$, i.e., 
\begin{equation}
\lambda^n=\lambda\times\ldots\times\lambda\;\;\mbox{($n$-times)},\quad\mbox{$\lambda$ Haar measure on $\rat_p$}.
\end{equation}
With a slight abuse of notation, we will denote by $\lambda$ the Haar measure on $\rat_p^n$ for every $n\in\mathbb{N}$, as the dimension $n$ will be clear from the context.
\end{example}

Let $G$ be a LC group, and let $X$ be a LC Hausdorff space. We call $X$ a (\emph{transitive}) \emph{$G$-space} whenever it is equipped with a (transitive) continuous left action  $(\hspace{0.2mm}\cdot\hspace{0.2mm})\,[\hspace{0.5mm}\cdot\hspace{0.5mm}]\colon G\times X\rightarrow X$ of $G$. 
If $G$ is a locally compact second countable Hausdorff (in short, LCSC) group, and $H$ a closed \emph{normal} subgroup of $G$ (e.g., the centre of $G$), let  $X\equiv G/H$ denote the quotient (LCSC) group. Furthermore, let $q\colon G\rightarrow X$ be the \emph{quotient map}  (i.e., the projection homomorphism) which is an \emph{open continuous} map. We can then define a natural continuous action $(\hspace{0.2mm}\cdot\hspace{0.2mm})\,[\hspace{0.5mm}\cdot\hspace{0.5mm}]\colon G\times X\rightarrow X$ of $G$ on $X$, i.e., 
\begin{equation}\label{eq.7s}
g[x]\coloneqq q(g)x,\qquad g\in G,\,\,x\in X. 
\end{equation}
This action is transitive and, hence, turns $X\equiv G/H$ into a transitive $G$-space. In the literature, one refers to such a $G$-space as to a \emph{homogeneous space}~\cite{folland2016course,varadarajan,reiter2000,hewitt1979, gaal,doran88}.

Let now $\mu_G$, $\mu_H$, $\mu_X$ denote the (left) Haar measures on $G$, $H$, $X\equiv G/H$ respectively, and let $\Delta_G$, $\Delta_H$ be the modular functions on $G$ and $H$. It is a standard fact that (since $X$ admits a $X$-invariant, hence $G$-invariant, measure $\mu_X$; see Theorem~$2.51$ of~\cite{folland2016course})
\begin{equation}
\Delta_G(h)=\Delta_H(h), \qquad \forall h\in H,    
\end{equation}
i.e., $\Delta_H=\Delta_{G}|_H$. Therefore, if $G$ is unimodular, then $H$ shares the same property.

Let $(X,\mathcal{B}_X)$, $(Y,\mathcal{B}_Y)$ be \emph{(Borel) measurable spaces}. We recall that a map $\varphi\colon X\rightarrow Y$ is called a \emph{Borel map} if, for every Borel set $\mathcal{E}\in \mathcal{B}_Y$,  $\varphi^{-1}(\mathcal{E})\in \mathcal{B}_X$; it is called a \emph{Borel isomorphism} if it is one-one, onto, and $f^{-1}$ is a Borel map. If $X\equiv G/H$ is a quotient group, we also denote by $\ts\colon X\rightarrow G$ a \emph{Borel (cross) section} of $X$ into $G$, i.e., a Borel map satisfying the condition $q(\ts(x))=x$, for every $x\in X$.
\begin{proposition}[Lemma~$6$ of~\cite{AnielloPR}]
For every Borel section $\ts\colon X\rightarrow G$, the mapping
\begin{equation}
    \gamma_\ts\colon X\times H\ni(x,h)\mapsto \ts(x)h\in G
\end{equation}
is a Borel isomorphism ($X\times H$ being endowed with the product topology).
\end{proposition}
For every $f\in \mathrm{C}_c(G)$, we put
\begin{align}\label{eq.4s2}
(Pf)(x)&\coloneqq \int_H\de\mu_H(h)(f\circ\gamma_\ts)(x,h)\nonumber\\
&=\int_H \de\mu_H(h)f(\ts(x)h), \qquad x\in X.
\end{align}
\begin{remark}
It is worth observing that the function $H\ni h\mapsto f(gh)\in\mathbb{C}$, for any $g\in G$, is in $\mathrm{C}_c(H)$ (in particular, $gh\in\supp(f)\implies h\in g^{-1}\supp(f)\bigcap H$), where $g^{-1}\supp(f)\bigcap H$ is a compact subset of $G$ and, hence, of $H$). Therefore, the integral on the r.h.s.\ of~\eqref{eq.4s2} is well-defined.
\end{remark}
\begin{remark}
Note that, by the left-invariance of $\mu_H$, the integral $\int_H \de\mu_H(h)f(gh)$ is constant w.r.t.\ $g$ varying in $q^{-1}(\{x\})$, for every $x\in X$. Hence, $(Pf)(x)\in\mathbb{C}$ does not depend on the choice of the cross section $\ts$.
\end{remark}
\begin{theorem}\label{theo.2.3s2}
For every $f\in \mathrm{C}_c(G)$, the function 
\begin{equation}
X\ni x\mapsto (Pf)(x)\in\mathbb{C}
\end{equation}
belongs to $\mathrm{C}_c(X)$, and the mapping $\mathrm{C}_c(G)\ni f\mapsto Pf\in \mathrm{C}_c(X)$ is surjective. Moreover, for every $f\in \mathrm{C}_c(G)$, we have that 
\begin{align}\label{eq.6s2}
\int_G \de\mu_G(g)f(g)&=\int_{X\times H}\de\mu_X\times\mu_H(x,h)f(\ts(x)h)\nonumber\\
&=\int_X\de\mu_X(x)\int_H\de\mu_H(h)f(\ts(x)h)\nonumber\\
&=\int_X\de\mu_X(x)(Pf)(x), \qquad\textup{(Weil-Mackey-Bruhat formula)},
\end{align}
where the Haar measures $\mu_G$, $\mu_H$, $\mu_X$ are supposed to be suitably normalized and $\ts\colon X\rightarrow G$ is any Borel cross section.
\end{theorem}
(Note: Since $X,H$ are LCSC groups, in the first line of~\eqref{eq.6s2} it is not necessary to make a distinction between the standard product of measures and the Radon product~\cite{folland2016course}.)
\begin{proof}
See Section~$2.6$ of~\cite{folland2016course}; in particular, Proposition~$2.50$ and Theorem~$2.51$.
\end{proof}
For every $\phi\in \mathrm{C}_c(X)$ and $\psi\in \mathrm{C}_c(G)$, we set
\begin{equation}
(\mathscr{L}_\psi\phi)(g)\coloneqq \psi(g)\phi(q(g)),\qquad g\in G.
\end{equation}
It is easy to see that $\mathscr{L}_\psi\phi\in \mathrm{C}_c(G)$; in particular, we have that
\begin{equation}
\supp(\mathscr{L}_\psi\phi)\subset\supp(\psi)\bigcap q^{-1}(\supp(\phi))
\end{equation}
is a compact subset of $G$. 
\begin{lemma}\label{lem1.s2}
For every compact subset $K$ of $X$, there exists a function $\psi\in \mathrm{C}_c^+(G)$ such that 
\begin{equation}
(P\psi)(x)=1,\quad\forall x\in K.
\end{equation}
Here and in the following, we set $\mathrm{C}_c^+(G)\coloneqq\{f\in\mathrm{C}_c(G)\mid f\geq 0,\,\,f\not\equiv 0\}$.
\end{lemma}
\begin{proof}
Use Lemma~$2.49$ of~\cite{folland2016course}.    
\end{proof}
By Lemma~\ref{lem1.s2}, for every nonempty compact subset $K$ of $X$, we can define the following (nonempty) subset of $\mathrm{C}_c^+(G)$
\begin{equation}\label{eq.15s2}
\Psi_K\coloneqq \{\psi\in \mathrm{C}_c^+(G)\mid (P\psi)(x)=1,\;\forall x\in K\}.   
\end{equation}
By convention, we put $\Psi_{\emptyset}=\{\psi\equiv 0\}$. 
\begin{definition}\label{wmblift}
Given any $\phi\in \mathrm{C}_c(X)$, for every $\psi\in\Psi_{\supp(\phi)}$, we call the function $\mathscr{L}_\psi\phi\in \mathrm{C}_c(G)$ a \emph{Weil-Mackey-Bruhat (WMB) lift} --- specifically, the $\psi$-\emph{lift} --- of $\phi$.
\end{definition}
The notion of Weil-Mackey-Bruhat lift in Definition~\ref{wmblift} is strictly related to the WMB formula~\eqref{eq.6s2}. Indeed, given a $\psi$-lift of a function $\phi\in\mathrm{C}_c(X)$, exploiting the WMB formula, it is not difficult to prove the following results:
\begin{lemma}\label{lem2.6s2}
For every $\phi\in \mathrm{C}_c(X)$, and every $\psi\in \Psi_{\supp(\phi)}$, we have that
\begin{equation}\label{eq.10s2}
P(\mathscr{L}_\psi\phi)=\phi.    
\end{equation}
\end{lemma}
\begin{proof}
In fact, by Lemma~\ref{lem1.s2}, we have:
\begin{align}
\big(P(\mathscr{L}_\psi\phi)\big)(x)&=\int_H\de\mu_H(h)\psi(\ts(x)h)\phi(q(\ts(x))h)\nonumber\\
&=(P\psi)(x)\phi(x)=\phi(x),\qquad\forall x\in X,
\end{align}
where $\ts \colon X\rightarrow G$ is any Borel cross section ($q(\ts(x)h)=x$).
\end{proof}
We are now able to express any Haar integral on $X$ as a Haar integral on $G$:
\begin{theorem}\label{th2.7s2}
Let $\phi$ be a function in $\mathrm{C}_c(X)$. Then, for every WMB lift $\mathscr{L}_\psi\phi\in \mathrm{C}_c(G)$ of $\phi$ \normalfont{($\psi\in\Psi_{\supp(\phi)}$)}, we have that 
\begin{equation}\label{eq.12s2}
\int_X \de\mu_X(x)\phi(x)=\int_G\de\mu_G(g)(\mathscr{L}_\psi\phi)(g),
\end{equation}
where a suitable (mutual) normalization of $\mu_X$ and $\mu_G$ is assumed. 
\end{theorem}
\begin{proof}
In fact, by the second assertion of Theorem~\ref{theo.2.3s2}, 
\begin{align}
\int_G\de\mu_G(g)(\mathscr{L}_\psi\phi)(g)&=\int_X\de\mu_X(x)\big(P(\mathscr{L}_\psi\phi)\big)(x)\nonumber\\
&=\int_X\de\mu_X(x)\phi(x),
\end{align}
where, for the second equality, we have used Lemma~\ref{lem2.6s2}.
\end{proof}
We will call the Haar integral on the r.h.s.\ of~\eqref{eq.12s2} a \emph{lift} of the Haar integral on the l.h.s.\ of the same formula.  

\bigskip

In our specific applications, $X=G/H$ will be a \emph{compact} group. In this case, some of the previously discussed results admit a remarkable generalization. To start with, let us notice that, when $X$ is compact, $\mathrm{C}_c(X)$ coincides with the set $\mathrm{C}(X)$ of all continuous functions on $X$. Let us put $\Psi\equiv\Psi_X$. From Theorem~\ref{th2.7s2}, we can immediately prove the following:
\begin{corollary}
Let $X=G/H$ be compact. Then, for every $\psi\in \Psi$, we have that
\begin{equation}
\int_X \de\mu_X(x)\phi(x)=\int_G\de\mu_G(g)(\mathscr{L}_\psi\phi)(g),\qquad\forall\phi\in \mathrm{C}(X),
\end{equation}
where a suitable (mutual) normalization of $\mu_X$ and $\mu_G$ is assumed.  
\end{corollary}
\begin{remark}
Fixed any $\psi\in\Psi\equiv\Psi_X$, the map $\mathscr{L}_\psi\colon \mathrm{C}(X)\rightarrow \mathrm{C}_c(G)$ is a right inverse of $P\colon \mathrm{C}_c(G)\rightarrow \mathrm{C}(X)$, i.e., it satisfies relation~\eqref{eq.10s2} for \emph{all} $\phi\in \mathrm{C}(X)$.    
\end{remark}
\begin{remark}\label{rem.2.14s}
Without any assumption of compactness of $X$, the map $P\colon \mathrm{C}_c(G)\rightarrow \mathrm{C}_c(X)$ can be extended to a (surjective) map $\widehat{P}\colon \mathrm{L}^1(G)\rightarrow \mathrm{L}^1(X)$, defined by 
\begin{equation}\label{eq.21s2}
\big(\widehat{P}f\big)(x)\coloneqq \int_H\de\mu_H(h)f(\ts(x)h),\qquad x\in X,\,f\in\mathrm{L}^1(G);
\end{equation}
see Lemma~$7$ of~\cite{AnielloPR} and Theorem~$3.4.6$ of~\cite{reiter2000} ($\mathrm{L}^1(G)\equiv \mathrm{L}^1(G,\mu_G)$ denotes the set of complex-valued functions on $G$ whose absolute value is integrable w.r.t.\ $\mu_G$). Moreover, the \emph{extended WMB formula} holds:
\begin{equation}
\int_G\de\mu_G(g)f(g)=\int_X\de\mu_X(x)\big(\widehat{P}f\big)(x),\qquad\forall f\in\mathrm{L}^1(G),
\end{equation}
for a suitable (mutual) normalization of the Haar measures $\mu_X$, $\mu_G$.
\end{remark}
The forthcoming Theorem will provide us with a suitable generalization of the results in Lemma~\ref{lem2.6s2} and in Theorem~\ref{th2.7s2}, tailored to the case where $X=G/H$ is a compact group. To this end, we find useful to preliminary recall the notion of  \emph{pushforward measure}. 
\begin{definition}
Let $(X,\mathcal{B}_X)$ and $(Y,\mathcal{B}_Y)$ be (Borel) measurable spaces. Let $\mu$ be a Borel measure on $X$.
If $\varphi\colon X\rightarrow Y$ is a Borel map of $X$ into $Y$, the \emph{pushforward measure} $\varphi_\ast\mu$ of 
$\mu$ through $\varphi$ is the measure on $(Y,\mathcal{B}_Y)$ defined by
\begin{equation}
\varphi_\ast\mu(\E)\coloneqq \mu\circ\varphi^{-1}(\E),
\end{equation} 
for every Borel set $\E$ in $\mathcal{B}_Y$.
\end{definition}
\begin{remark}\label{rem.3.5}
If $(X,\mathcal{B}_X)$ and $(Y,\mathcal{B}_Y)$ are Borel measurable spaces, and if $f\colon Y\rightarrow \mathbb{R}$ is a Borel function on $Y$, the following (abstract) change-of-variables formula (C.O.V.F., in short) holds~\cite{knapp05}:
\begin{equation}\label{eq.30}
\int_X (f\circ\varphi)\de\mu=\int_Y f\de(\varphi_\ast\mu).
\end{equation}
Moreover, from~\eqref{eq.30}, it is not difficult to prove the  following relation~\cite{knapp05}:
\begin{equation}\label{eq.2}
\varphi_\ast(g\de\mu)=g\circ\varphi^{-1}\de(\varphi_\ast\mu),
\end{equation}
for every Borel function $g\colon X\rightarrow\mathbb{R}$. We shall constantly resort to this formula in our description of integration theory on $\rat_p$-manifolds.
\end{remark}
We are now ready to prove the following result
\begin{theorem}\label{th.2.17}
Let $X=G/H$ be compact. Then, for every $\psi\in \Psi$, the map $\mathscr{L}_\psi\colon \mathrm{C}(X)\rightarrow \mathrm{C}_c(G)$ admits an extension --- a so-called \emph{extended WMB lift} --- 
\begin{equation}
\widehat{\mathscr{L}}_\psi\colon \mathrm{L}^1(X)\rightarrow \mathrm{L}^1(G),
\end{equation}
defined by
\begin{equation}\label{eq.27s}
\big(\widehat{\mathscr{L}}_\psi\phi\big)(g)\coloneqq \psi(g)(\phi\circ q)(g),
\end{equation}
that is a right inverse of $\widehat{P}$:
\begin{equation}\label{eq.19s2}
\widehat{P}\big(\widehat{\mathscr{L}}_\psi\phi\big)=\phi,\quad\forall\phi\in\mathrm{L}^1(X).
\end{equation}
Moreover, for every $\phi\in L^1(X)$, we have that
\begin{equation}\label{eq.20s2}
\int_X\de\mu_X(x)\phi(x)=\int_G\de\mu_G(g)\big(\widehat{\mathscr{L}}_\psi\phi\big)(g),
\end{equation}
for a suitable (mutual) normalization of $\mu_X$, $\mu_G$.
\end{theorem}
\begin{proof}
Let us first prove that, for every $\phi\in\mathrm{L}^1(X)$, the (Borel) function $\widehat{\mathscr{L}}_\psi\phi$ belongs to $\mathrm{L}^1(G)$. In fact, by Lemma~$7$ of~\cite{AnielloPR}, for any Borel section $\ts\colon X\rightarrow G$, the pushforward measure $(\gamma_\ts)_\ast(\mu_X\times\mu_H)$ coincides (up to normalization) with $\mu_G$. Hence, we have that
\begin{align}
\int_G \de\mu_G(g)\big|\big(\widehat{\mathscr{L}}_\psi\phi\big)(g)\big|&=\int_G\de\mu_G(g)\psi(g)|\phi(q(g))|\nonumber\\
&=\int_G \de\big((\gamma_\ts)_\ast(\mu_X\times\mu_H)\big)(g)\psi(g)|\phi(q(g))|\nonumber\\
&=\int_{X\times H}\de\mu_X\times\mu_H(x,h)\psi(\ts(x)h)|\phi(x)|\nonumber\\
&=\int_X \de\mu_X(x)\int_H\de\mu_H(h)\psi(\ts(x)h)|\phi(x)|,
\end{align}
where the last equality is obtained by Tonelli's theorem. Therefore, we find that
\begin{equation}\label{eq.22s2}
\int_G\de\mu_G(g)\big|\big(\widehat{\mathscr{L}}_\psi\phi\big)(g)\big|=\int_X\de\mu_X(x)|\phi(x)|=\|\phi\|_{\mathrm{L}^1(X)},
\end{equation}
and $\widehat{\mathscr{L}}_\psi\phi\in\mathrm{L}^1(G)$. At this point, one easily proves~\eqref{eq.19s2} and~\eqref{eq.20s2}.
\end{proof}
\begin{remark}
Relation~\eqref{eq.22s2} shows that $\widehat{\mathscr{L}}_\psi\colon\mathrm{L}^1(X)\rightarrow\mathrm{L}^1(G)$ is a (linear) isometry.
\end{remark}
To conclude this section, we state the following remarkable consequence of Theorem~\ref{theo.2.3s2}
\begin{theorem}
Let us suppose that $H$ is compact. Then, for a suitable normalization of $\mu_G$ and $\mu_X$, $q_\ast(\mu_G)=\mu_X$.    
\end{theorem}
\begin{proof}
Since $H$ is compact, $q_\ast(\mu_G)$ is a Radon measure on $X$ (for every compact $E\subset X$, $q^{-1}(E)=KH$, with $K\subset G$ compact by Lemma~$2.48$ of~\cite{folland2016course}; hence, $KH$ is compact too). Then, for every $\phi\in \mathrm{C}_c(X)$, by the WMB formula we have that 
\begin{align}
\int_X \de q_\ast(\mu_G)(x)\phi(x)&=\int_G\de\mu_G(g)\phi(q(g))\nonumber\\
&=\int_X\de\mu_X(x)\int_H\de\mu_H(h)\phi(x)=\int_X\de\mu_X(x)\phi(x),
\end{align}
where we have assumed that $\mu_H(H)=1$ and $\mu_G$,$\mu_H$ are suitably normalized. Hence, $q_\ast(\mu_G)=\mu_X$.
\end{proof}


\subsection{\texorpdfstring{$p$}{Lg}-Adic Lie groups}\label{subsec.2.1}
In this subsection, we discuss the main features of $p$-adic manifolds and $p$-adic Lie groups~\cite{schneider2011p,glockner,serre2009}.

As in the standard real setting, the starting point is to introduce a suitable notion of chart on a Hausdorff space.
\begin{definition}
Let $\h$ be a Hausdorff space. A \emph{chart} on $\h$ is a triple $(\U, \varphi,  \rat_p^n)$, where $\U\subset\h$ is an open subset and $\varphi\colon \U\rightarrow \rat_p^n$ is a map such that $\varphi\colon \U\rightarrow \varphi(\U)$ is a homeomorphism. We refer to $\U$ as the \emph{domain} of the chart, and to $n\in\mathbb{N}$ as its dimension. 
\end{definition}
If $x\in\U\subset\h$, we say that $(\U,\varphi,\rat_p^n)$ is a \emph{chart around $x$}. In the following, we will set $\ha\vf\colon\vf(\U)\rightarrow\U$ to be the inverse map of $\varphi$ on its range. 
\begin{definition}
If $U$ is an open subset of $\rat_p^n$, a  function $f\colon U\rightarrow \rat_p$ is said to be a   \emph{$\rat_p$-analytic function}, if it is expressed by a convergent power series in a neighborhood of every $x$ in $U$. A map $f=(f_1,\ldots,f_m)\colon U\rightarrow \rat_p^m$ is said to be a \emph{$\rat_p$-analytic map}, if every $f_i$, $i=1,\ldots,m$, is a $\rat_p$-analytic function.
\end{definition}
\begin{definition}
Two charts $(\U_1,\varphi_1,\rat_p^{n_1})$ and $(\U_2,\varphi_2,\rat_p^{n_2})$  on $\h$ are  \emph{compatible}, if both  $\varphi_2\circ\ha\varphi_1\colon \varphi_1(\U_1 \cap \U_2)\rightarrow \varphi_2(\U_1\cap \U_2)$ and $\varphi_1\circ\ha\varphi_2\colon \varphi_2(\U_1 \cap \U_2)\rightarrow \varphi_1(\U_1\cap \U_2)$ are $\rat_p$-analytic maps.
\end{definition}
If $(\U_1,\varphi_1,\rat_p^{n_1})$ and $(\U_2,\varphi_2,\rat_p^{n_2})$ are compatible charts on $\h$ such that $\U_1\cap \U_2\neq\emptyset$, then one can prove that, necessarily, it is $n_1=n_2$~\cite{schneider2011p}. 
\begin{definition}
An \emph{atlas} $\mathcal{A}$ for $\h$ is a family $\{(\U_\alpha,\varphi_\alpha,\rat_p^{n_\alpha})\}_{\alpha\in A}$ of pairwise compatible charts which cover $\h$, i.e., $\h=\bigcup_{\alpha\in A}\U_\alpha$. An atlas $\mathcal{A}$ for $\h$ is called $n$-dimensional if all the charts in $\mathcal{A}$ have dimension $n$.
\end{definition}
Similarly to the standard real case, it is now natural to set the following: 
\begin{definition}\label{def.2.4}
A Hausdorff space, $\h$, together with a maximal (w.r.t.\ inclusion) atlas $\mathcal{A}$ is called a $\rat_p$\emph{-analytic manifold}. The manifold is called
$n$-dimensional if the atlas $\mathcal{A}$ is $n$-dimensional.
\end{definition}
For notational convenience, in what follows we shall denote an $n$-dimensional atlas on $\h$ as $\mathcal{A}=\{(\U_\alpha,\vf_\alpha)\}_{\alpha\in A}$; moreover, we will refer to `$\rat_p$-analytic manifold' simply as `$\rat_p$-manifold'. If $\h$, $\hy$ are two $\rat_p$-manifolds of dimension $m$ and $n$ respectively, we shall say that a map $f$ from $\h$ to $\hy$ is \emph{$\rat_p$-analytic} if, for every $x\in\h$, there exist a chart $(\U,\varphi,\rat_p^m)$ on $\h$ around $x$, and a chart $(\mathscr{V}, \psi,\rat_p^n)$ on $\hy$ around $f(x)$, such that $f(\U)\subset \mathscr{V}$, and $\psi\circ f\circ \ha\varphi\colon \varphi(\U)\rightarrow \rat_p^n$ is a $\rat_p$-analytic map.
\begin{remark}\label{rem:totdiscimopencom}
Every $\rat_p$-manifold $\h$ is both totally disconnected and locally compact (TDLC in short). In particular, the latter condition entails that for every point $x$ of $\h$, the set $\mathcal{T}_x$ of all compact open subsets in $\h$ containing $x$ forms a base at $x$ (see Lemma~7.1.1 in~\cite{igusa2000}). Therefore, the set $\mathcal{T}(\h)=\bigcup_{x\in\h}\mathcal{T}_x$ of all the compact open subsets of $\h$ forms a basis for the topology of $\h$.
\end{remark}

 Analytic differential forms on $\rat_p$-manifolds are defined in a similar fashion to the standard real setting  (see Chapter 2 in~\cite{igusa2000} for a thorough discussion). Indeed, let $\h$ be a $\rat_p$-manifold of dimension $n$, and let $\mathcal{A}=\{(\U_\alpha,\varphi_\alpha)\}_{\alpha\in A}$ be an atlas on $\h$. If $\Theta$ is a \emph{differential form of degree $k<n$} on $\h$, its restriction $\Theta_\alpha\coloneqq\Theta|_{\U_\alpha}$ --- in the local coordinates of $(\U_\alpha,\vf_\alpha)$ --- is given by
\begin{equation}
\Theta_\alpha(u)=\sum_{j_1<\ldots<j_k}\theta^\alpha_{j_1\ldots j_k}(u)dx_{j_1}\wedge\ldots\wedge dx_{j_k},   
\end{equation}
where $\theta^\alpha_{j_1\ldots j_k}$ are $\rat_p$-valued functions on $\U_\alpha$, and where we set $\vf_\alpha(u)=(x_1,\ldots,x_n)$ to denote the local coordinates of $u$ 
 in $\U_\alpha$. If, for every $\alpha\in A$, the maps $\theta^\alpha_{j_1,\ldots, j_k}$ are all $\rat_p$-analytic functions on $\U_\alpha$, we say that $\Theta$ is a \emph{$\rat_p$-analytic differential $k$-form} on $\h$. If $\Omega$ is a $\rat_p$-analytic differential $n$-form on $\h$ (i.e., of maximal degree equal to the dimension $n$ of $\h$), its local expression $\Omega_\alpha\coloneqq\Omega|_{\U_\alpha}$ can be written as
\begin{equation}  \label{eq:formamaxloca} 
\Omega_\alpha(u)=\omega_\alpha(u)\de x_1\wedge\ldots\wedge\de x_n,
\end{equation}
for $\omega_\alpha\colon\U_\alpha\rightarrow\mathbb{Q}_p$ a $\rat_p$-analytic function. In what follows, we shall abbreviate `$\rat_p$-analytic differential $k$-form’ to `differential $k$-form’.

Let $F\colon \h\rightarrow\hy$ be a $\rat_p$-analytic map between $n$-dimensional $\rat_p$-manifolds $\h$, $\hy$, and let $\Xi$ be a differential $n$-form on $\hy$. The \emph{pullback} $F^\ast\Xi$ of $\Xi$ through $F$ is a well-defined differential $n$-form on $\h$~\cite{igusa2000}; specifically, if $(\U_\alpha,\vf_\alpha)$ is a chart in $\h$, and $(\mathscr{V}_\beta,\psi_\beta)$ is a chart in $\hy$, then, on $\U_\alpha\cap F^{-1}(\mathscr{V}_\beta)$ one has:
\begin{equation}\label{eq:diffformpull}
F^\ast\Xi_\beta=F^\ast(\xi_\beta\,\de y_1\wedge\ldots\wedge \de y_n)=(\xi_\beta\circ F)(\det \mathrm{D}F)\de x_1\wedge\ldots\wedge\de x_n,
\end{equation}
where $(x_i)_{i=1}^n$ and $(y_j)_{j=1}^n$ denote the systems of local coordinates of $\U_\alpha$ and $\mathscr{V}_\beta$ respectively, and where $\mathrm{D}F$ is the \emph{Jacobian matrix} of the transformation $F$.

To conclude this subsection, we now discuss the principal object of our investigations: 
\begin{definition}\label{defliegroups}
A \emph{$p$-adic Lie group} $G$ is a $\mathbb{Q}_p$-manifold which is also a group, and such that the multiplication map 
\begin{equation}\label{prodonlig}
 G\times G \ni(g,h)\mapsto gh\in G
\end{equation}
is $\rat_p$-analytic. 
\end{definition}
From Definition~\ref{defliegroups}, it follows that the inverse map, $G\ni g\mapsto g^{-1}\in G$, is a $\rat_p$-analytic map. Moreover, it is clear that every $p$-adic Lie group is a TDLC Hausdorff space (see Definition~\ref{def.2.4} and Remark~\ref{rem:totdiscimopencom}).
\begin{remark}\label{rem.2.9}
Let $G$ be a $p$-adic Lie group. For $h\in G$, the map $\ell_h$ of \emph{left translation} by $h$ is defined as:
\begin{equation}
G\ni g\mapsto \ell_h(g)\coloneqq hg\in G.
\end{equation}
This map is the composition of the map $G\ni g\mapsto (h,g)\in G\times G$, and the multiplication map defined in~\eqref{prodonlig}; hence, it is $\rat_p$-analytic (the composition of $\rat_p$-analytic maps is a $\rat_p$-analytic map; see Lemma~$8.4$ in~\cite{schneider2011p}). Similarly, one can define the map of right translation, $r_h$, on $G$, which is $\rat_p$-analytic as well. 
\end{remark}
\begin{remark} 
A classical result by van Dantzig (see Theorem~7.7 in~\cite{hewitt1979}) states that a TDLC group admits a base at the identity consisting of compact open subgroups (and vice versa). This result provides a peculiar characterization of the topology of $p$-adic Lie groups.
\end{remark}
Since a $p$-adic Lie group, $G$, is a $\rat_p$-manifold, we can clearly define differential $k$-forms on it. In particular, we say that a differential $k$-form $\Theta$ on  $G$ is \emph{left-invariant} if $\ell_h^\ast\Theta=\Theta$ for any $h\in G$, i.e., if
\begin{equation}\label{eq.72s}
\ell_h^\ast\Theta(g)=\Theta(h^{-1}g)
\end{equation}
holds for every $g$ and $h$ in $G$. Right-invariant differential $n$-forms are defined similarly with $\ell_h$ replaced by $r_h$. By taking $h\equiv g^{-1}$ and $g\equiv e$ in~\eqref{eq.72s}, we also see that 
\begin{equation}\label{eq.73s}
\ell_{g^{-1}}^\ast\Theta(e)=\Theta(g),
\end{equation}
that is, if $\Theta$ is left-invariant on $G$, its value at every point on $G$ is determined by the value $\Theta$ assumes at the identity $e$ in $G$. In the next subsection, we shall prove that a left-invariant $n$-form on $G$ can always be constructed, and that it naturally induces the left-invariant Haar measure on $G$. For the moment, we want to stress a relevant topological feature of $p$-adic Lie groups which will turn out to be central in our later derivations.  
 We recall that a Hausdorff space $\h$ is called \emph{paracompact}, if every open covering of $\h$ can be refined into a \emph{locally finite} open covering. We say that $\h$ is \emph{strictly paracompact} if every open cover of $\h$ admits a refinement consisting of pairwise disjoint open sets. 
\begin{proposition}\label{prop.2.9}
Let $G$ be a second countable $p$-adic Lie group. Then, $G$ is a strictly paracompact space. 
\end{proposition}
\begin{proof}
By assumption, $G$ is locally compact, second countable and Hausdorff, hence $\sigma$-compact (i.e.\ union of countably many compact subspaces). Every $\sigma$-compact space is Lindel\"of, and, therefore, paracompact (cf.\ Theorem~5.1.11.\ in~\cite{engelking89}). Then, the proposition follows by the equivalence of the points $\mathrm{i}$ and $\mathrm{ii}$ of Proposition~8.7 in~\cite{schneider2011p}.   
\end{proof}
\begin{remark}
From Proposition~\ref{prop.2.9} it follows that any second countable $p$-adic Lie group $G$ can always be endowed with an atlas consisting of pairwise disjoint charts. Indeed, since every atlas is an open covering, it admits a refinement consisting of pairwise disjoint open sets. Then, the restriction of the coordinate maps of the initial atlas to the sets in the refinement provides a new system of charts for $G$.
\end{remark}
\begin{remark}\label{rem.2.31}
It is well known that a LC group $G$ is \emph{Polish} iff it admits a second countable topology  (see Theorem~5.3 in~\cite{kechris95set}). Therefore, a second countable $p$-adic Lie group is also a Polish group.
\end{remark}

In this work, $p$-adic Lie groups will always be assumed to be second countable (as so are the most important examples); hence, they are LCSC Polish groups.


\subsection{\texorpdfstring{$p$}{Lg}-Adic rotation groups}\label{subsec.3.2n}
A noteworthy class of  $p$-adic Lie groups is given by the special orthogonal groups over the $p$-adic fields. We devote this subsection to recall some of their  basic properties~\cite{our2,our1st}. 

The general definition of special orthogonal group is given in terms of \emph{quadratic forms} $Q\colon \mathrm{V}\rightarrow \mathbb{F}$, for $\mathrm{V}$ a vector space over a field $\mathbb{F}$ (see~\cite{ serre2012course,Cassels} for a thorough discussion). Quadratic forms, up to linear equivalence and scaling, lead to isomorphic special orthogonal groups~\cite{serre2012course}. In this work, we always assume that the quadratic
forms $Q$ are \emph{non-degenerate} (i.e. they have maximum rank).

If the characteristic of $\mathbb{F}$ is different from $2$ (as it is for $\mathbb{R}$ and $\mathbb{Q}_p$), a bilinear form $b(\mathbf{x},\mathbf{y})$, $\mathbf{x},\,\mathbf{y}\in \mathrm{V}$, induces a quadratic form $Q(\mathbf{x}) = b(\mathbf{x}, \mathbf{x})$ and, \emph{vice versa}, a quadratic form induces a bilinear form, i.e.,
\begin{equation}\label{eq:sclprodQbil} b(\mathbf{x},\mathbf{y})\coloneqq \frac{1}{2}\big(Q(\mathbf{x}+\mathbf{y})-Q(\mathbf{x})-Q(\mathbf{y})\big).
\end{equation}
Therefore, we are allowed to use interchangeably quadratic forms and bilinear forms.

The unique non-degenerate definite quadratic form on $\mathbb{R}^n$, for every $n\geq2$, is given (up to linear equivalence and scaling) by $Q_{\mathbb{R}}(\mathbf{x})=\sum_{i=0}^{n-1}x_i^2$. This is represented in the canonical basis by the $n$-dimensional identity matrix $\mathrm{I}_n$. 
Thus, the (compact) special orthogonal group over $\mathbb{R}$ of degree $n$ is 
\begin{align}
\so(n,\mathbb{R})&=\{L\in \mathsf{M}_n(\mathbb{R})\;\;|\;\;L^\top L=\mathrm{I}_n,\ \det(L)=1\}\\
&=\{L\in \mathsf{M}_n(\mathbb{R})\;\;|\;\;\langle L\mathbf{x},L\mathbf{y}\rangle=\langle \mathbf{x}, \mathbf{y}\rangle\ \textup{for every}\ \mathbf{x}, \mathbf{y}\in\mathbb{R}^n,\;\;\det(L)=1\},
\end{align}
where $\langle\hspace{0.3mm}\cdot\hspace{0.6mm},\cdot\hspace{0.3mm}\rangle\colon\mathbb{R}^n\times\mathbb{R}^n\rightarrow\mathbb{R}$ is the Euclidean scalar product on $\mathbb{R}^n$, and $\mathsf{M}_n(\mathbb{R})$ denotes the associative algebra of $n\times n$ matrices over the field of real numbers $\mathbb{R}$.

\medskip

The following theorem characterizes the definite quadratic forms over $\mathbb{Q}_p$ in every dimension, as explicitly derived in~\cite{our1st} (see also~\cite{serre2012course,Cassels}).
\begin{theorem}\label{quadform}
For every prime $p>2$, let $u\in\mathbb{U}_p$ be a non-square --- with $\mathbb{U}_p$ denoting the group of $p$-adic units, i.e., the group of all invertible elements of $\mathbb{Z}_p$ --- and let $v\in\mathbb{U}_p$ be defined by
\begin{equation}\label{eq:vdef}
v\coloneqq
\begin{cases}
 -1 &\text{if } p\equiv3 \mod4,\\
 -u &\text{if } p\equiv1 \mod4.
\end{cases}
\end{equation}
In the case where $p>2$, there are (precisely) three definite quadratic forms on $\mathbb{Q}_p^2$, up to linear equivalence and scaling,
\begin{equation}\label{eq:quadrform2}
 Q_{-v}(\mathbf{x}) =x_0^2-vx_1^2,\quad 
 Q_{p}(\mathbf{x})    =x_0^2+px_1^2,\quad
 Q_{\frac{p}{u}}(\mathbf{x}) =u x_0^2+px_1^2,
\end{equation}
and there are seven on $\mathbb{Q}_2^2$, namely,
\begin{alignat}{2}       Q_1(\mathbf{x})   &=x_0^2+x_1^2,\quad&       Q_{\pm2}(\mathbf{x})  &=x_0^2\pm2x_1^2,\nonumber\\       Q_{\pm5}(\mathbf{x})  &=x_0^2\pm5x_1^2,\quad&       Q_{\pm10}(\mathbf{x}) &=x_0^2\pm10x_1^2.     \label{eq:quadrform22}
\end{alignat}
There is a unique definite quadratic form on $\mathbb{Q}_p^3$ (depending on $p$), up to linear equivalence and scaling, i.e.,
\begin{equation}\label{eq:quadrform3}
Q_+(\mathbf{x})
 =\begin{cases}
    x_0^2-vx_1^2+px_2^2 &\textup{if } p>2, \\     x_0^2+x_1^2+x_2^2   &\textup{if } p=2,   \end{cases}
\end{equation}
as well as on $\mathbb{Q}_p^4$, i.e.,
\begin{equation}\label{eq:quadrform4}
  Q_{(4)}(\mathbf{x}) 
         = \begin{cases}
             x_0^2 - v x_1^2 + p x_2^2 - pv x_3^2 & \textup{ if } p>2, \\              x_0^2 + x_1^2 + x_2^2 + x_3^2      & \textup{ if } p=2.          \end{cases}
\end{equation}
No quadratic form on $\rat_p^n$ is definite for $n\geq 5$.
\end{theorem}
\begin{remark}\label{rem:equivrestricQ}
Note that all restrictions of the definite quadratic form $Q_{(4)}$ on $\mathbb{Q}_p^4$ to any three variables, and indeed to any three-dimensional subspace, are  equivalent to $Q_+$ on $\mathbb{Q}_p^3$.
\end{remark} 
We can now characterize a relevant class of special orthogonal groups over $\rat_p$. 
\begin{corollary}\label{cor:compunigru}
The $p$-adic special orthogonal groups associated with the definite quadratic forms on $\rat_p^2$ are (up to isomorphism)
\begin{equation}
    \so(2,\mathbb{Q}_p)_\kappa=\{L\in \mathsf{M}_2(\mathbb{Q}_p)\;\;|\;\; A_\kappa=L^\top A_\kappa L,\ \det(L)=1\},
\end{equation}
where $A_\kappa$ are the matrix representations, in the canonical basis of $\mathbb{Q}_p^2$, of the quadratic forms in~\eqref{eq:quadrform2} and~\eqref{eq:quadrform22}. Index $\kappa$ ranges in $\{-v,p,\frac{p}{u}\}$ whenever $p>2$, while $\kappa\in\{1,\pm2,\pm5,\pm10\}$ when $p=2$.\\
For every $p\geq2$, the special orthogonal group associated with the definite quadratic form on $\mathbb{Q}_p^3$ is (up to isomorphisms)
\begin{equation}
    \so(3,\mathbb{Q}_p)=\{L\in \mathsf{M}_3(\mathbb{Q}_p)\;\;|\;\; A=L^\top AL,\ \det(L)=1\},
\end{equation}
while the one on $\mathbb{Q}_p^4$ is
\begin{equation}
    \so(4,\mathbb{Q}_p)=\{L\in \mathsf{M}_4(\mathbb{Q}_p)\;\;|\;\; A'=L^\top A'L,\ \det(L)=1\}.
\end{equation}
$A$ and $A'$ are the matrix representations in the canonical basis of $\mathbb{Q}_p^3$ and $\mathbb{Q}_p^4$ of the quadratic forms~\eqref{eq:quadrform3} and~\eqref{eq:quadrform4} respectively.
\end{corollary}

 It is not difficult to prove that the special orthogonal groups $\so(n,\mathbb{Q}_p)$, $n=2,3,4$, in Corollary~\ref{cor:compunigru} are compact as subsets in $\rat_p^{n^2}$.  Indeed, we can introduce a $p$-adic (non-Archimedean) norm on $\so(n,\mathbb{Q}_p)$ by setting $\|M\|_p=\|(M_{ij})_{ij}\|_p\coloneqq \max_{i,j=1,\dots,n}|M_{ij}|_p$. Clearly, $\so(n,\mathbb{Q}_p)$, $n=2,3,4$, turn into topological groups, whenever they are endowed with the natural topology generated by the open balls of the $p$-adic norm. We recall that a set $K \subset \mathbb{Q}_p^m$ is \emph{compact} if and only if it is \emph{closed} and \emph{bounded} w.r.t.\ the ultrametric topology generated by (the open balls of) the $p$-adic norm $\mathrm{N}_p(\mathbf{x})\coloneqq\max_{i=1,\dots,m}|x_i|_p$
of $\mathbb{Q}_p^m$~\cite{vladimirov1994, rooij78}. $\so(n,\mathbb{Q}_p),\ n=2,3,4,$ are closed, as they are groups of solutions of a system of continuous (polynomial) equations. On the other hand, every matrix in the groups of Corollary~\ref{cor:compunigru} has bounded entries (see Theorem~5 and Remark~14 in \cite{our1st} for the details); specifically, we have:\footnote{To be precise, the statements concerning $p=2$ are not quite right in Ref.~\cite{our1st}, and are corrected here. 
}
\begin{align}  & \mathrm{SO}(2,\mathbb{Q}_p)_\kappa= \mathrm{SO}(2,\mathbb{Z}_p)_\kappa\quad \textup{except}\quad  \mathrm{SO}(2,\mathbb{Q}_2)_{-5}= \mathrm{SO}(2,2^{-1}\mathbb{Z}_2)_{-5};\nonumber\\ 
&\mathrm{SO}(3,\mathbb{Q}_p)= \mathrm{SO}(3,\mathbb{Z}_p);\label{eq:inclsSO4SO3SO2SLint}\\ 
& \mathrm{SO}(4,\mathbb{Q}_p)= \mathrm{SO}(4,\mathbb{Z}_p)\quad \textup{except}\quad \mathrm{SO}(4,\mathbb{Q}_2)= \mathrm{SO}(4,2^{-1}\mathbb{Z}_2).\nonumber \end{align}
This entails that $\|M\|_p\leq p$ for every $M\in \so(n,\mathbb{Q}_p)\subset \mathsf{M}_n(\mathbb{Q}_p)\cong \mathbb{Q}_p^{n^2}$, i.e., $\so(n,\mathbb{Q}_p)$, $n=2,3,4$, is a bounded subset of $\mathbb{Q}_p^{n^2}$. 
\begin{remark}
We have used definite quadratic forms to define the $p$-adic special orthogonal groups. It turns out that those groups defined on indefinite quadratic forms are not bounded, whence, not compact.
\end{remark}
In the light of the discussion above, the following result is now clear:
\begin{proposition}\label{prop:compatti}
The groups $\so(n,\rat_p)$, $n=2,3,4$ of Corollary~\ref{cor:compunigru} are \emph{all and the only compact} $p$-adic special orthogonal groups. 
\end{proposition}
The next theorem provides a parameterization of the compact $p$-adic special orthogonal groups in dimension two~\cite{our1st}.
\begin{theorem}\label{th.23}
Any element of $\so(2,\mathbb{Q}_p)_\kappa$ takes the following matrix form in the canonical basis of $\mathbb{Q}_p^2$:
\begin{equation}\label{genelso2}
R_\kappa(\alpha)=\begin{pmatrix}
\frac{1-\kappa\alpha^2}{1+\kappa\alpha^2} & -\frac{2\kappa\alpha}{1+\kappa\alpha^2}\\ \frac{2\alpha}{1+\kappa\alpha^2} & \frac{1-\kappa\alpha^2}{1+\kappa\alpha^2}
\end{pmatrix},\quad \alpha\in\mathbb{Q}_p\cup\{\infty\},
\end{equation}
where $R_\kappa(\infty)=-R_\kappa(0)=-\mathrm{I}_2$, and $\kappa\in\{-v,p,\frac{p}{u}\}$ for $p>2$, while $\kappa\in\{1,\pm2,\pm5,\pm10\}$ for $p=2$.
The composition of two elements in $\so(2,\mathbb{Q}_p)_\kappa$, for any fixed $\kappa$, is given by
\begin{equation}\label{comphar}
R_\kappa(\alpha)R_\kappa(\beta)=R_\kappa\left(\frac{\alpha+\beta}{1-\kappa\alpha\beta}\right),
\end{equation}
for every $\alpha,\beta\in\mathbb{Q}_p\cup\{\infty\}$.
\end{theorem}

\begin{remark}\label{oss:reductio}
By choosing $\kappa=1$, and taking $\alpha=\tan\left(\theta/2\right)$, a generic $p$-adic rotation, as given in~\eqref{genelso2}, assumes the form
\begin{equation}
\begin{pmatrix}
\cos\theta&-\sin\theta\\\sin\theta&\cos\theta
\end{pmatrix},
\end{equation}
i.e., it `formally’ reduces to a real planar rotation by an angle $\theta$. 
\end{remark}


\subsection{Integration on \texorpdfstring{$p$}{Lg}-adic manifolds}\label{subsec.3.2}
This subsection deals with integration theory on $p$-adic manifolds~\cite{igusa2000}. For our purposes, we will need a $p$-adic counterpart of the well known change-of-variables formula for multiple integrals on $\mathbb{R}^n$. Therefore, we start with the following:
\begin{theorem}[Change-of-variables formula]\label{changeofvarint}
Let $a\in\mathbb{Q}_p^n$ and let $\xi=(\xi_1,\ldots,\xi_n)\colon U\subset \mathbb{Q}_p^n\rightarrow V\subset\rat_p^n$ be a $\mathbb{Q}_p$-analytic isomorphism between an open neighborhood $U$ of $a$, and an open neighborhood $V$ of $\xi(a)$, such that 
\begin{equation}\label{Jacdet}
\det\left(\frac{\partial \xi_i}{\partial x_j}(a)\right)\neq0.   
\end{equation}
Then, for every integrable function $f$ on $V$, the following formula holds:
\begin{equation}\label{eq.29}
\int_V f(x)\de\lambda_{|V}(x)=\int_U f(\xi(x))\left\lvert\det\left(\frac{\partial \xi_i}{\partial x_j}(x)\right)\right\rvert_p\de\lambda_{|U}(x),  
\end{equation}
where $\lambda$ is the Haar measure on $\rat_p^n$.
\end{theorem}
\begin{proof}
Formula~\eqref{eq.29} is actually a special case of the abstract C.O.V.F. (see~\eqref{eq.30} in Remark~\ref{rem.3.5}) specialized to the case where the pushforward of the measure on $U$ is realized via a $\rat_p$-analytic map. See Proposition~$7.4.1$  in~\cite{igusa2000} for the technical details. 
\end{proof}

Let $\h$ be a second countable $n$-dimensional $\rat_p$-manifold, and $\Omega$ a differential $n$-form on $\h$. If $\mathcal{A}=\{(\U_\alpha,\varphi_\alpha)\}_{\alpha\in A}$ is an atlas on $\h$, $\Omega$ is expressed as in~\eqref{eq:formamaxloca} in the local coordinates of each chart in $\mathcal{A}$.
Then, we can associate a Radon measure $\mu_{\Omega}$ with $\Omega$ by setting
\begin{equation}\label{eq.32}
\mu_\Omega(\sC)\coloneqq\int_\sC|\omega_\alpha(u)|_p\de\big((\ha\vf_\alpha)_\ast\lambda\big)(u),
\end{equation}
for every compact (open) subset $\sC\subset\U_\alpha$ of $\h$, and where $\de((\ha\vf_\alpha)_\ast\lambda)(u)$ denotes the pushforward of the Haar measure $\lambda$ by $\ha\vf_\alpha$. It is not difficult to verify that this measure is well-defined:
if $(\U_\beta,\varphi_\beta)$ is another chart in $\mathcal{A}$ containing 
$\sC$ (i.e., $\sC\subset \U_\beta\cap\U_\alpha$), then \begin{equation}\label{eq.5oh}
\int_\sC|\omega_\alpha(u)|_p\de\big((\ha\vf_\alpha)_\ast\lambda\big)(u)=\int_\sC|\omega_\beta(u)|_p\de\big((\ha\vf_\beta)_\ast\lambda\big)(u),
\end{equation}
that is, $\mu_\Omega(\sC)$ does not depend on the considered chart containing $\sC$. We shall give the proof of this result in Remark~\ref{rem:dimoanticipata} below. 
\begin{remark}\label{rem.3.19}
Since a second countable $\rat_p$-manifold is $\sigma$-compact, the measure~\eqref{eq.32} is 
\emph{regular} (cf.\ Theorem~7.8 in~\cite{folland99}). Then, the measure of a Borel set $\E\subset\U_\alpha$ of $\h$ is given, by \emph{inner regularity}, by the supremum of the measures of the compact (open) sets contained in $\E$. 
\end{remark}
If $f\in \mathrm{C}_c(\h)$ is such that $\supp(f)\subset\sC\subset\U_\alpha$, its integral w.r.t.\ $\mu_\Omega$ is also well-defined, and is given by 
\begin{equation}\label{eq.34s}
\int_{\h}f\Omega\coloneqq\int_{\U_\alpha}f(u)|\omega_\alpha(u)|_p\de\big((\ha\vf_\alpha)_\ast\lambda\big)(u).
\end{equation}


Let now $\sC$ be an arbitrary compact (open) subset of $\h$. Its measure w.r.t.\ $\mu_\Omega$ can be defined as follows. First, we can decompose $\sC$ as
\begin{equation}
\sC=\bigsqcup_{i}\sC_i,\qquad\mbox{$\sC_i\subset\U_\alpha$, for some $\alpha\in A$},
\end{equation}
i.e., as a disjoint union of compact (open) subsets $\sC_i$, each contained in some $\U_\alpha$. Then, the measure of $\sC$ is given by 
\begin{equation}\label{eq:35comp}
\mu_\Omega(\sC)=\sum_i\mu_\Omega(\sC_i).
\end{equation}
Similarly, we can then extend the measure~\eqref{eq:35comp} to arbitrary Borel sets $\E$ in $\h$ (see Remark~\ref{rem.3.19}). Exploiting~\eqref{eq.34s}, it is then not difficult to define the integral of an arbitrary function $f\in \mathrm{C}_c(\h)$ w.r.t.\ $\mu_\Omega$, as well.

We can consider the pushforward of the measure $\mu_\Omega$ via $\vf_\alpha$ to a measure on $\rat_p^n$. This allows us to treat the integration theory on a manifold $\h$ via integrals on $\mathbb{Q}_p^n$.
Indeed, using formula~\eqref{eq.2}, we have: \begin{equation}\label{eq.6} 
\de\big((\vf_\alpha)_\ast\mu_{\Omega}\big)(x)=(\vf_\alpha)_\ast\Big(|\omega_\alpha(u)|_p\de\big((\ha\vf_\alpha)_\ast\lambda\big)(u)\Big)=|(\omega_\alpha\circ\ha\vf_\alpha)(x)|_p\de\lambda(x), \end{equation}
where $\varphi_\alpha(u)=(x_1,\ldots,x_n)\eqqcolon x$ denotes the coordinate representation of the point $u\in \h$. Hence, using the (abstract) C.O.V.F. (cf.\ relation~\eqref{eq.30} in Remark~\ref{rem.3.5}) with $f=\chi_\sC$, 
we obtain
\begin{equation}\label{eq.37}     
\mu_{\Omega}(\sC)=\int_\sC|\omega_\alpha(u)|_p\de\big((\ha\vf_\alpha)_\ast\lambda\big)(u)=\int_{\vf_\alpha(\sC)}|(\omega_\alpha\circ\ha\vf_\alpha)(x)|_p\de\lambda(x). \end{equation} 
Furthermore, if $f\in \mathrm{C}_c(\h)$ is a function with $\supp(f)\subset\sC\subset\U_\alpha$, it is clear that 
\begin{equation}\label{eq.s38} 
\int_{\U_\alpha}f(u)|\omega_\alpha(u)|_p\de\big((\ha\vf_\alpha)_\ast\lambda\big)(u)=\int_{\vf_\alpha(\U_\alpha)}(f\circ\ha\vf_\alpha)(x)|(\omega_\alpha\circ\ha\vf_\alpha)(x)|_p\de\lambda(x). 
\end{equation}

With the above discussion, we get to the following two conclusions. First, from~\eqref{eq.32}, we see that
\begin{equation}
\mu_{\Omega}|_{\U_\alpha}\ll(\ha\vf_\alpha)_\ast\lambda, \qquad \frac{\de\mu_{\Omega}|_{\U_\alpha}}{\de(\ha\vf_\alpha)_\ast\lambda}=|\omega_\alpha|_p,
\end{equation}
i.e., $\mu_{\Omega}|_{\U_\alpha}$ is an \emph{absolutely continuous measure} w.r.t.\ $(\ha\vf_\alpha)_\ast\lambda$, with \emph{Radon-Nikodym derivative} given by $|\omega_\alpha|_p\colon\U_\alpha\rightarrow \mathbb{R}^+_{\ast}$. Secondly,~\eqref{eq.37} entails that $(\vf_\alpha)_\ast\mu_{\Omega}\ll\lambda|_{\vf_\alpha(\U_\alpha)}$ as well, with Radon-Nikodym derivative $|\omega_\alpha\circ \ha\vf_\alpha|_p\colon \vf_\alpha(\U_\alpha)\rightarrow\mathbb{R}^+_{\ast}$. The latter condition means that the pushforward of the measure $\mu_\Omega$ on $\h$ via the maps $\vf_\alpha$, for every $\alpha\in A$, provides an absolutely continuous measure w.r.t.\ the Haar measure $\lambda$ on $\mathbb{Q}_p^n$. Their Radon-Nikodym derivative --- which,  for notational convenience, hereafter we will simply denote by $\eta$ --- is globally defined, and it is \emph{uniquely defined} up to a set of points of null measure (any other Radon-Nikodym derivative is equal to $\eta$ \emph{almost everywhere}). Accordingly, we shall denote by  $\eta_\alpha\coloneqq|\omega_\alpha\circ\ha\vf_\alpha|_p$ the restriction of $\eta$ on $\vf_\alpha(\U_\alpha)\subset \mathbb{Q}_p^n$, for every $\U_\alpha$ in the covering atlas  $\mathcal{A}$ of $\h$. Exploiting this notation, and recalling condition~\eqref{eq.37}, we can then write
\begin{equation}\label{eq.s40}
\mu_{\Omega}(\sC)=\int_{\vf_\alpha(\sC)}\eta_\alpha(x)\de\lambda(x),
\end{equation}
for every compact (open) subset $\sC\subset \U_\alpha\subset \h$. Similarly, we can express the integral w.r.t.\ $\Omega$ of every function $f\in \mathrm{C}_c(\h)$ --- with $\supp(f)\subset\sC\subset \U_\alpha$ --- as 
\begin{equation}
\int_{\h}f\Omega=\int_{\U_\alpha}f(u)|\omega_\alpha(u)|_p\de\big((\ha\vf_\alpha)_\ast\lambda\big)(u)=\int_{\vf_\alpha(\U_\alpha)}(f\circ\ha\vf_\alpha)(x)\eta_\alpha(x)\de\lambda(x).  
\end{equation}
\begin{remark}\label{rem:dimoanticipata}
Using the local representation~\eqref{eq.s40}, it is now not difficult to prove the equality of integrals in~\eqref{eq.5oh}. Indeed, let $\sC\subset \U_\alpha\cap\U_\beta$ be a compact (open) set in $\h$. Then, we want to show that 
\begin{equation}\label{eq:prequivincord}
\int_{\vf_\alpha(\sC)}\eta_\alpha(x)\de\lambda(x)=\int_{\vf_\beta(\sC)}\eta_\beta(y)\de\lambda(y).
\end{equation}
We first consider the change of variable $y=(\vf_\beta\circ\ha\vf_\alpha)(x)$ in the r.h.s.\ of~\eqref{eq:prequivincord}. Then, using Theorem~\ref{changeofvarint}, we see that~\eqref{eq:prequivincord} holds iff
\begin{equation}\label{eq.44}
\eta_\alpha(x)=(\eta_\beta\circ\vf_\beta\circ\ha\vf_\alpha)(x)|\det\mathrm{D}(\vf_\beta\circ\ha\vf_\alpha)(x)|_p,
\end{equation}
where $\det\mathrm{D}(\vf_\beta\circ\ha\vf_\alpha)(x)$ denotes the Jacobian of the transformation $\vf_\beta\circ\ha\vf_\alpha$. 
On the other hand, the pullback formula~\eqref{eq:diffformpull} also shows that
\begin{align}\label{eq.44s}
(\omega_\alpha\circ\ha\vf_\alpha)(x)\de x_1\wedge\ldots\wedge\de x_n&=(\ha\vf_\alpha)^\ast\Omega\nonumber\\
&=(\vf_\beta\circ\ha\vf_\alpha)^\ast(\ha\vf_\beta)^\ast\Omega\nonumber\\
&=(\vf_\beta\circ\ha\vf_\alpha)^\ast(\omega_\beta\circ\ha\vf_\beta)(y)\de y_1\wedge\ldots\wedge\de y_n\nonumber\\
&=(\omega_\beta\circ\ha\vf_\beta\circ\vf_\beta\circ\ha\vf_\alpha)(x)\det\mathrm{D}(\vf_\beta\circ\ha\vf_\alpha)(x)\de x_1\wedge\ldots\wedge\de x_n.
\end{align}
Therefore, taking the $p$-adic absolute value of the l.h.s.\ and of the last equality in~\eqref{eq.44s} entails 
that~\eqref{eq.44} (and, hence,~\eqref{eq:prequivincord}) holds. 
\end{remark}
To conclude this subsection, we prove
that it is always possible to construct an essentially unique --- i.e., uniquely defined up to a multiplicative constant --- (left-)invariant differential $n$-form on every $n$-dimensional $p$-adic Lie group. We will then show that it is naturally  associated with the (left) Haar measure on the group. This will draw a parallel with the standard theory of (real) Lie groups~\cite{knappLie, helgason01}.

Let us first note that, also in the $p$-adic setting the tangent space $\mathrm{T}_eG$ to $G$ at $e\in G$ has a natural structure of \emph{Lie algebra} $\mathfrak{g}$, whenever the elements $X\in \mathrm{T}_eG$ are identified with the corresponding \emph{left-invariant vector fields} $\widetilde{X}$ on $G$~\cite{glockner}. Let $X_1,\dots, X_n$ be a basis of $\mathrm{T}_eG$, and let $\widetilde{X}_1,\dots,\widetilde{X}_n$ be the corresponding left-invariant vector fields in $\mathfrak{g}$. We can now define, for all $g$ in $G$, the $1$-forms 
 $\omega_1,\dots,\omega_n$ on $G$ via the condition 
 \begin{equation}
(\omega_i)_g\left((\widetilde{X}_j)_g\right)=\delta_{ij},\quad j=1,\ldots,n.
 \end{equation}
 By construction, $\omega_1,\dots,\omega_n$ are left-invariant $1$-forms on $G$, as follows by observing that
\begin{equation}\label{eq.47s}
(\ell_g^\ast \omega_i)(\tilde{X}_j)=\omega_i(\ell_g^\ast\tilde{X}_j)=\omega_i(\tilde{X}_j).
\end{equation}
In particular, this also entails that $\omega_1,\ldots,\omega_n$ form a basis of the dual space of $\mathrm{T}_gG$ for every $g\in G$. Therefore, the differential form $\Oinv$ defined as
\begin{equation}
\Oinv\coloneqq\omega_1\wedge\dots\wedge \omega_n,
\end{equation}
is a (nowhere vanishing) left-invariant $n$-form  on $G$. Indeed,  since the pullback $\ell_g^\ast$ commutes with $\wedge$, we have:
\begin{equation}\label{eq.48s}
\ell_g^\ast\Oinv=\ell_g^\ast(\omega_1\wedge\ldots\wedge\omega_n)=\ell_g^\ast\omega_1\wedge\ldots\wedge\ell_g^\ast\omega_n=\omega_1\wedge\ldots\wedge\omega_n=\Oinv,
\end{equation}
that is, $\Oinv$ is left-invariant. It is clear that any constant multiple of $\Oinv$ is a left-invariant $n$-form as well. Conversely, if $\breve{\Omega}$ is another left-invariant $n$-form on $G$, there must exist $c\in\rat_p$ such that $\breve{\Omega}(e)=c\Oinv(e)$. But then, the left-invariance condition~\eqref{eq.73s} entails that $\breve{\Omega}(g)=c\Oinv(g)$ for every $g$ in $G$. 

We want now to show that if $\Oinv$ is the left-invariant differential $n$-form on $G$, its induced measure, $\mu_{\Oinv}$, is the left Haar measure on $G$ (up to multiplicative constants). Indeed, we already know that $\mu_{\Oinv}$ is a Radon measure. To conclude that it is a Haar measure, we have to show that it is left-invariant. Let $\sC$ be a compact (open) set in $G$. From the left-invariance of $\Oinv$ we see that
\begin{equation}
\omega_\alpha\circ\ha\vf_\alpha=(\ha\vf_\alpha)^\ast\Oinv=(\ha\vf_\alpha)^\ast L_{g^{-1}}^\ast\Oinv,
\end{equation}
for every $g\in G$. This entails that $\mu_{\Oinv}(g\sC)=\mu_{\Oinv}(\sC)$, for every compact (open) set $\sC\subset G$, and $g$ in $G$ (see~\eqref{eq.37}). Moreover, since $G$ is second countable, $\mu_{\Oinv}$ is regular. In particular, inner regularity entails that
\begin{align}
\mu_{\Oinv}(\E)&=\sup\{\mu_{\Oinv}(K)\mid \mbox{$K\subset \E$ compact}\}\nonumber\\
&=\sup\{\mu_{\Oinv}(gK)\mid \mbox{$K\subset\E$ compact}\}\nonumber\\
&=\mu_{\Oinv}(g\E),
\end{align}
for every Borel set $\E$ in $G$, and every $g\in G$. Concluding, we proved that $\mu_{\Oinv}$ is a left-invariant Radon measure on $G$, and since the Haar measure is essentially uniquely defined, it must coincide with the Haar measure on $G$ up to a multiplicative constant.

\section{The Haar measure on \texorpdfstring{$p$}{Lg}-adic Lie groups}\label{cosHaarmes}
In this section, we show how to construct a left Haar measure $\mu$ on a (second countable) $p$-adic Lie group $G$. Our approach exploits the peculiar topological features of $p$-adic Lie groups, and relies on the possibility to construct a \emph{quasi-invariant measure} for $G$. Eventually, we will prove that the measure thus constructed coincides with the measure induced by the left-invariant differential $n$-form $\Oinv$ on $G$ (see Subsection~\ref{subsec.3.2}).

We begin by recalling the notion of a quasi-invariant measure~\cite{folland2016course}. Let $G$ be a $p$-adic Lie group, and let $\nu$ be a Radon measure on it. For $h\in G$, we can define the left translation $\nu^h$, of $\nu$ by $h$, as
\begin{equation}\label{eq.46}
\nu^h(\E)\coloneqq \nu(h\E),
\end{equation}
for every Borel set $\E\in\mathcal{B}_G$. We say that $\nu$ is \emph{quasi-invariant} if the measures $\nu^h$ are all equivalent, i.e., mutually absolutely continuous~\cite{folland99}. In such a case, we have:
\begin{equation}\label{eq.48}
\de\nu^h(g)=\eta(h,g)\de\nu(g),
\end{equation}
where $\eta\colon G\times G\rightarrow \mathbb{R}^+_{\ast}$ is a positive map on $G\times G$. The function $\eta$ is the Radon-Nikodym derivative $\de\nu^h/\de\nu$. For $h,h'\in G$, since $\nu^{hh'}=(\nu^h)^{h'}$, the chain rule for the Radon-Nikodym derivative entails the following \emph{cocycle formula}:
\begin{equation}\label{eq.47}
\eta(hh',g)=\eta(h,h'g)\eta(h',g),
\end{equation}
for every $g\in G$. In particular, using~\eqref{eq.47} it is not difficult to prove the following result.
\begin{lemma}\label{lemm.4.1}
Let $G$ be a $p$-adic Lie group, and let $\nu$ be a quasi-invariant measure on it. The measure defined as
\begin{equation}\label{eq.4}  \de\mu(g)\coloneqq\eta(g,e)^{-1}\de\nu(g)
\end{equation}
--- where $e$ denotes the identity element in $G$ --- is a left Haar measure on $G$.
\end{lemma}
\begin{proof}
Let $\mu^h$ be the left translation, by $h$ in $G$, of the measure $\mu$, as defined in~\eqref{eq.46}. For every Borel set $\E$ in $\mathcal{B}_G$, we have:
\begin{equation}\label{eq.50'}
\mu^h(\E)=\int_{h\E}\eta(g,e)^{-1}\de\nu(g)=\int_\E\eta(hg,e)^{-1}\de\nu^h(g),
\end{equation}
where in the last equality we have used the change of variable $h^{-1}g\mapsto g$. Then, taking into account condition~\eqref{eq.48} for quasi-invariant measures, and exploiting the cocycle formula~\eqref{eq.47}, we have: 
\begin{equation}
\eta(hg,e)=\eta(h,g)\eta(g,e),\qquad \de\nu^h(g)=\eta(h,g)\de\nu(g),
\end{equation}
which yield
\begin{equation}\label{eq.52}
\int_\E\eta(hg,e)^{-1}\de\nu^h(g)=\int_\E\eta(h,g)^{-1}\eta(g,e)^{-1}\eta(h,g)\de\nu(g)=\int_{\E}\eta(g,e)^{-1}\de\nu(g)=\mu(\E).
\end{equation}
Therefore, the first equality in~\eqref{eq.50'} and the last one in~\eqref{eq.52} give the desired result.
\end{proof}
From Lemma~\ref{lemm.4.1} we see that it is always possible to construct a left Haar measure on a $p$-adic Lie group $G$ once known a quasi-invariant measure on it. Hence, our next step is to show how to explicitly construct a quasi-invariant measure on $G$. 

Let $\mathcal{A}=\{(\U_\alpha,\vf_\alpha)\}_{\alpha\in A}$ be a \emph{disjoint atlas} on $G$ (cf.\ Proposition~\ref{prop.2.9}). We can construct a (regular) Radon measure $\nu$ on $G$ as follows. First, in every chart $(\U_\alpha,\vf_\alpha)$ in $\mathcal{A}$, we define a measure $\nu_\alpha$ on $\U_\alpha$ by setting
\begin{equation}\label{eq.56}
\nu_\alpha\coloneqq (\ha \vf_\alpha)_\ast\lambda_\alpha,\qquad\lambda_\alpha=\lambda\rvert_{\vf_\alpha(\U_\alpha)}, 
\end{equation}
that is, $\nu_\alpha$ is the pushforward measure, via $\ha\vf_\alpha\colon \vf_\alpha(\U_\alpha)\rightarrow \U_\alpha$, of the restricted Haar measure $\lambda\rvert_{\vf_\alpha(\U_\alpha)}$ on $\rat_p^n$. 
Note that since $\nu_\alpha$ is finite on compact sets, it is a Radon measure. In this way, we have constructed a Radon measure on every chart $(\U_\alpha,\vf_\alpha)$ in $\mathcal{A}$. To obtain a Radon measure $\nu$ on the whole group $G$, we can then act as follows. Given any Borel set $\E$ in $\mathcal{B}_G$, we express it as the disjoint union $\E=\bigsqcup_{\alpha\in A} \E_\alpha$, where $\E_\alpha\coloneqq \E\cap \U_\alpha$, and set
\begin{equation}\label{eq.54}
\nu(\E)\coloneqq\sum_{\alpha\in A}\nu_\alpha(\E_\alpha).
\end{equation}
Since $A$ is countable, the series in~\eqref{eq.54} contains a countable number of non-null terms. 
It is now easily proved that the measure defined in~\eqref{eq.54} is a (regular) Radon measure on $G$. Indeed, $\nu$ takes values in $[0,+\infty]$ as so do all the $\nu_\alpha$s, and $\nu(\emptyset)=0$. If $\{\E_i\}_i$ is a countable family of Borel sets in $G$, then $\nu(\cup_i \E_i)=\sum_i\nu(\E_i)$, as follows by observing that the $\nu_\alpha$s are $\sigma$-additive, and that the summation order can be exchanged in the double series $\sum_\alpha\sum_i\nu_\alpha(\E_i\cap\U_\alpha)$ by positivity of the $\nu_\alpha$s. Moreover, $\nu$ is clearly finite on compact sets, and since $G$ is second countable, we can conclude that $\nu$ is a regular, and hence Radon, measure on $G$ (cf.\ Theorem~7.8 in~\cite{folland99}).

Our next step is to show that this measure is quasi-invariant. To this end, let $h\in G$ be some fixed point, and let us set, for any $\alpha,\beta\in A$,
\begin{equation}
\U_{\alpha,\beta}^h\coloneqq\{g\in\U_\alpha\mid hg\in\U_\beta\}=h^{-1}\Big((h\U_\alpha)\cap \U_\beta)\Big)=\U_\alpha\cap(h^{-1}\U_\beta).
\end{equation}
Note that $\U_{\alpha,\beta}^h\subset \U_\alpha$ is an \emph{open} set and 
\begin{equation}
\U_\alpha=\bigsqcup_{\beta}\U_{\alpha,\beta}^h.
\end{equation}
Assuming that $\U_{\alpha,\beta}^h\neq\emptyset$, for every $j=1,\dots,n$, and at given $h\in G$, we put
\begin{equation}\label{eq.55}
\vf_\alpha(\U_{\alpha,\beta}^h)\ni x\mapsto \zeta_{\beta,j}(h;x)\coloneqq\vf_{\beta,j}(h\ha\vf_\alpha(x))\in\mathbb{Q}_p,
\end{equation}
where $\vf_{\beta,j}$ is the $j$-th vector component of $\varphi_\beta\colon\U_\beta\rightarrow\mathbb{Q}_p^n$, i.e. $\varphi_\beta=\left(\vf_{\beta,1},\dots,\vf_{\beta,j},\dots,\vf_{\beta,n}\right)$. Moreover, the definition of $\zeta_{\beta,j}(h;\,\cdot\,)$ can be extended to the whole open set $\vf_\alpha(\U_\alpha)=\bigsqcup_\beta\vf_\alpha(\U_{\alpha,\beta}^h)$ by varying $\beta$ in $A$. In this way, we obtain a map $\zeta_{\beta,j}(h;\,\cdot\,)\colon \vf_\alpha(\U_\alpha)\rightarrow \rat_p$ (for suitable labels $\beta$ depending on the charts as in \eqref{eq.55}), for any given $h\in G$. We can then define a function 
\begin{equation}
 \rho_\beta(h;\,\cdot\,)\colon \vf_\alpha(\U_\alpha)\rightarrow \mathbb{R}^+_{\ast},\quad\rho_\beta(h;x)\coloneqq\left\vert\det\left[\frac{\partial \zeta_{\beta,j}}{\partial x_k}(h;x)\right]_{1\leq j,k\leq n}\right\vert_p.
\end{equation}
Eventually, we obtain a function $\eta\colon G\times G\rightarrow \mathbb{R}^+_{\ast}$, defined as follows:
\begin{equation}\label{eq.62s}
 \eta(h,g)\coloneqq \rho_\beta(h;\vf_\alpha(g)),\quad g\in\U_\alpha,\,\,\alpha\in A.   
\end{equation}
Let us now define a (regular) Radon measure $\mu_h$ on $G$ by setting 
\begin{equation}
    \de\mu_h(g)=\eta(h,g)\de\nu(g).
\end{equation}
We want to prove that $\nu$ is quasi-invariant and, moreover, $\mu_h=\nu^h$, so that 
\begin{equation}
\frac{\de\nu^h}{\de\nu}(g)=\eta(h,g).    
\end{equation}
Since $\nu$ is a regular measure, then $\nu^h$ and $\mu_h$ are \emph{regular} measures. Hence, by outer regularity, it is sufficient to show that 
\begin{equation}\label{eq.65}
    \mu_h(\Os)=\nu^h(\Os),
\end{equation}
for every open set $\Os\subset G$. Actually, since 
\begin{equation}
\Os=\bigsqcup_\alpha(\Os\cap\U_\alpha),
\end{equation}
it is sufficient to prove~\eqref{eq.65}
on every open subset $\Os\equiv \Os_\alpha$ of $\U_\alpha$, for $\alpha\in A$. Moreover, since for any open $\Os_\alpha\subset\U_\alpha$,
\begin{equation}
\Os_\alpha=\bigsqcup_\beta(\Os_\alpha\cap\U_{\alpha,\beta}^h)= \bigsqcup_\beta\Os_{\alpha,\beta}^h,    \qquad \Os_{\alpha,\beta}^h\coloneqq \Os_\alpha\cap\U_{\alpha,\beta}^h,
\end{equation}
it is enough to prove~\eqref{eq.65} on every open subset $\Os\equiv \Os_{\alpha,\beta}^h$ of $\U_{\alpha,\beta}^h$. Assuming that $\Os_{\alpha,\beta}^h\neq \emptyset$ (otherwise there is nothing to prove), we have:
\begin{align}
\nu^h(\Os_{\alpha,\beta}^h)&=\int_{h\Os_{\alpha,\beta}^h}\de\nu(g)&\nonumber\\
&=\int_{\vf_\beta(h\Os_{\alpha,\beta}^h)}\de\lambda(x)&(\mbox{since  $h\Os_{\alpha,\beta}^h\subset\U_\beta$})\nonumber\\
&=\int_{\vf_\alpha(\Os_{\alpha,\beta}^h)}\rho(h;x)\de\lambda(x)& (\mbox{by C.O.V.F.~\eqref{eq.29}})\nonumber\\
&=\int_{\Os_{\alpha,\beta}^h}\eta(h,g)\de\nu(g)&(\mbox{by~\eqref{eq.56}-\eqref{eq.54} and~\eqref{eq.62s}})\nonumber\\
&=\mu_h(\Os_{\alpha,\beta}^h).
\end{align}
In conclusion, we have $\nu^h=\mu_h$. Therefore, $\nu^h$ and $\nu$ are mutually absolutely continuous, for every $h\in G$ --- namely, $\nu$ is quasi-invariant --- and $\frac{\de\nu^h}{\de\nu}(g)=\eta(h,g)$.

As a direct consequence of Lemma~\ref{lemm.4.1}, the left Haar measure $\mu$ on $G$ is of the form
\begin{equation}
\de\mu(g)=\eta(g,e)^{-1}\de\nu(g).
\end{equation}
With the above construction, we have proved the following result:
\begin{theorem}\label{theo.4.3}
Let $G$ be a $p$-adic Lie group, and let $\mathcal{A}=\{(\U_\alpha,\vf_\alpha)\}_{\alpha\in A}$ be a disjoint atlas on $G$. If $\mu$ is the left Haar measure on $G$ then, for every Borel set $\E$ in $\mathcal{B}_G$, and any $\U_\alpha$ in $\mathcal{A}$, the following equality holds:
\begin{equation}\label{eq:harrr}
\mu(\E\cap\U_\alpha)=\int_{\varphi_\alpha(\E\cap\U_\alpha)}\left|\det\left[\frac{\partial \zeta_{\alpha,j}}{\partial x_{k}}\big( \ha\vf_\alpha(y);\varphi_0(e)\big)\right]_{1\leq j,k\leq n}\right|_{p}^{-1}\de\lambda(y),
\end{equation}
where $(\U_0,\vf_0)$ is the chart around the identity $e\in G$, $(x_k)_{k=1}^n$ denotes a system of local coordinates w.r.t.\ $(\U_0,\vf_0)$, and $\zeta_j$ is the map defined in~\eqref{eq.55}.
\end{theorem}
\begin{remark}
In~\eqref{eq:harrr}, the functions $\zeta_{\alpha,j}$ are correctly labelled by $\alpha$. In fact, their derivatives are performed in a neighborhood of $x=\varphi_0(e)$, and, thus, $\zeta_{\alpha,j}(\ha\vf_\alpha(y);x)=\varphi_{\alpha,j}(\ha\vf_\alpha(y)\ha\vf_0(x))$ whenever $\ha\vf_0(x)\in\U_0$ is `sufficiently
close to $e$' such that $\ha\vf_\alpha(y)\ha\vf_0(x)\in\U_\alpha$.
\end{remark}

Now, we prove that Theorem~\ref{theo.4.3} still holds in the case of an atlas including possibly overlapping charts. Indeed, let  $\mathcal{A}=\{(\U_\alpha,\vf_\alpha)\}_{\alpha\in A}$ be an arbitrary atlas on $G$. Since $G$ is strictly paracompact (see Proposition~\ref{prop.2.9}), we can always find a refinement $\mathcal{A}'$ of $\mathcal{A}$ consisting of pairwise disjoint charts. Then, Theorem~\ref{theo.4.3} provides us with a left Haar measure on (every chart of) $\mathcal{A}'$. To show that this measure is well-defined on $\mathcal{A}$ as well, we have to prove that for every Borel set $\E$ in $\mathcal{B}_G$ contained in the intersection of two charts in $\mathcal{A}$, the value of the integral in~\eqref{eq:harrr} is the \emph{same} w.r.t.\ the local coordinates of  the two charts; that is, we want to prove
\begin{equation}\label{eq.75n}
\mu(\E\cap\U_\alpha)=\mu(\E\cap\U_\beta),
\end{equation}
for every Borel set $\E$ in $G$ such that $\E\subset\U_\alpha\cap\U_\beta$, $\U_\alpha,\U_\beta\in\mathcal{A}$. To start with, the r.h.s.\ of~\eqref{eq.75n} explicitly is
\begin{equation}
\mu(\E\cap\U_\beta)=\int_{\vf_\beta(\E\cap\U_\beta)}\left|\det\left[\frac{\partial \zeta_{\beta,j}}{\partial x_{k}}\big( \ha\vf_\beta(z);\varphi_0(e)\big)\right]_{1\leq j,k\leq n}\right|_{p}^{-1}\de\lambda(z),
\end{equation}
where we have denoted by $y$ the local coordinates in the chart $(\U_\beta,\varphi_\beta)$. Then, the change of variable $z=\vf_\beta\circ\ha\vf_\alpha(y)$ immediately yields 
\begin{align}\label{eq.77}
&\int_{\vf_\beta(\E\cap\U_\beta)}\left|\det\left[\frac{\partial \zeta_{\beta,j}}{\partial x_{k}}\big( \ha\vf_\beta(z);\varphi_0(e)\big)\right]\right|_{p}^{-1}\de\lambda(z)\nonumber\\
&=\int_{\vf_\alpha(\E\cap\U_\alpha)}\left|\det\left[\frac{\partial\zeta_{\beta,j}}{\partial x_k}\Big((\ha\vf_\beta\circ\vf_\beta\circ\ha\vf_\alpha)(y);\vf_0(e)\Big)\right]\right|_p^{-1}\,\bigg|\det\bigg[\frac{\partial(\vf_\beta\circ\ha\vf_\alpha)_j}{\partial y_k}(y)\bigg]\bigg|_p\de\lambda(y)\nonumber\\
&=\int_{\vf_\alpha(\E\cap\U_\alpha)}\left|\det\left[\frac{\partial\zeta_{\alpha,j}}{\partial x_k}(\ha\vf_\alpha(y);\vf_0(e))\right]\right|_p^{-1}\,\bigg|\det\bigg[\frac{\partial(\vf_\beta\circ\ha\vf_\alpha)_j}{\partial y_k}(y)\bigg]\bigg|_p^{-1}\,\bigg|\det\bigg[\frac{\partial(\vf_\beta\circ\ha\vf_\alpha)_j}{\partial y_k}(y)\bigg]\bigg|_p\de\lambda(y)\nonumber\\
&=\int_{\vf_\alpha(\E\cap\U_\alpha)}\left|\det\left[\frac{\partial\zeta_{\alpha,j}}{\partial x_k}(\ha\vf_\alpha(y);\vf_0(e))\right]\right|_p^{-1}\de\lambda(y)=\mu(\E\cap\U_\alpha),
\end{align}
where, for notational convenience, we have omitted $1\leq j,k\leq n$ in the Jacobians. Note that 
in the second equality of~\eqref{eq.77}, we have used the C.O.V.F.\ for multiple integrals in $\rat_p^n$ (cf.\ Theorem~\ref{changeofvarint}). 
Moreover, in the third equality, we have used the fact that $\zeta_{\beta,j}$ and $\zeta_{\alpha,j}$ are related 
via the condition $\zeta_{\alpha,j}=\vf_{\alpha,j}\circ\ha\vf_{\beta,j}\circ\zeta_{\beta,j}$, and then we have exploited the usual chain rule for the Jacobian of a composite function. Therefore,~\eqref{eq.75n} shows that the (left) Haar measure in Theorem~\ref{theo.4.3} is well-defined over overlapping charts; that is, it does not depend on the particular chosen chart in $\mathcal{A}$. Concluding, we have the following Corollary of Theorem~\ref{theo.4.3}.
\begin{corollary}\label{cor.4.3}
Let $G$ be a $p$-adic Lie group, and let $\mathcal{A}=\{(\U_\alpha,\vf_\alpha)\}_{\alpha\in A}$ be a (not necessarily disjoint) 
atlas on $G$. The left Haar measure $\mu$ on $G$ is expressed, in the local coordinates of any given chart $(\U_\alpha,\varphi_\alpha)$ in $\mathcal{A}$, as
\begin{equation}\label{eq.71}
\mu(\E\cap\U_\alpha)=\int_{\varphi_\alpha(\E\cap\U_\alpha)}\left|\det\left[\frac{\partial \zeta_{\alpha,j}}{\partial x_{k}}\big( \ha\vf_\alpha(y);\varphi_0(e)\big)\right]_{1\leq j,k\leq n}\right|_{p}^{-1}\de\lambda(y),
\end{equation}
for every Borel set $\E\in\mathcal{B}_G$, where $(\U_0,\vf_0)$ is the chart around $e\in G$, and $(x_k)_{k=1}^n$ denotes a system of local coordinates w.r.t.\ $(\U_0,\vf_0)$.
\end{corollary}
To conclude this section, we now show that the Haar measure~\eqref{eq.71} \emph{coincides} with the measure on $G$ associated with the left-invariant differential $n$-form $\Oinv$ on $G$, as constructed in Subsection~\ref{subsec.3.2}. Indeed, let us denote with $\breve{\Omega}$ the differential $n$-form on $G$ whose local expression  $\breve{\Omega}_\alpha$, in every chart $\U_\alpha$ in $\mathcal{A}$, is given by
\begin{equation}\label{eq.71s}
\breve{\Omega}_\alpha(g)=\det[\mathrm{D}\zeta_\alpha(g;\vf_0(e))]^{-1}\de x_1\wedge\ldots\wedge\de x_n,    
\end{equation}
where, as usual,  $\mathrm{D}\zeta_\alpha$ denotes the Jacobian matrix of $\zeta_\alpha=(\zeta_{\alpha,j})_{j=1}^{n}$, and where we set $\vf_\alpha(g)=(x_1,\ldots,x_n)$. It is clear that the measure $\mu_{\breve{\Omega}}$, associated with $\breve{\Omega}$ via relation~\eqref{eq.37},  coincides with the Haar measure in~\eqref{eq.71}. (It is worth noting that, from Corollary~\ref{cor.4.3}, it follows that $\breve{\Omega}$ does not depend on the particular chosen chart on $G$, i.e. it is a well-defined differential $n$-form on $G$). To prove that the form~\eqref{eq.71s} \emph{coincides} with the left-invariant differential $n$-form $\Oinv$ on $G$, it is enough to show that condition~\eqref{eq.72s} holds, i.e., $\ell_h^\ast\breve{\Omega}(hg)=\breve{\Omega}(g)$, for every $h,g$ in $G$. Indeed, this will prove that $\breve{\Omega}$ is a left-invariant differential $n$-form on $G$, and due to its essential uniqueness,  we can then conclude that it coincides with $\Oinv$ (up to a multiplicative constant). In fact, we have:
\begin{align}
\ell_h^\ast\breve{\Omega}(hg)&=\ell_h^\ast\Big(\det[\mathrm{D}\zeta_\beta(\cdot\,;\vf_0(e))]^{-1}(hg)\de y_1\wedge\ldots\wedge\de y_n\Big)\nonumber\\
&=\det[\mathrm{D}(\zeta_\beta(\cdot\,;\vf_0(e))\circ\ell_h)]^{-1}(hg)\det[\mathrm{D}\ell_h]\de x_1\wedge\ldots\wedge\de x_n\nonumber\\
&=\det[\mathrm{D}\zeta_\alpha(h^{-1}hg;\vf_0(e))]^{-1}\det[\mathrm{D}\ell_h]^{-1}\det[\mathrm{D}\ell_h]\de x_1\wedge\ldots\wedge\de x_n\nonumber\\
&=\breve{\Omega}(g),
\end{align}
where we set $\vf_\alpha(hg)=(y_1,\ldots,y_n)$. Note that, in the second equality we have used the pullback formula~\eqref{eq:diffformpull} for differential forms, while in the third equality we have used the formula for the Jacobian of a composite function, taking into account the relation $\zeta_{\alpha,j}=\vf_{\alpha,j}\circ\ha\vf_{\beta,j}\circ\zeta_{\beta,j}$ between $\zeta_{\alpha,j}$ and $\zeta_{\beta,j}$. Hence, since the Haar measure is essentially uniquely defined,  we get to the conclusion that the left Haar measure~\eqref{eq.71} must \emph{coincide} (up to a multiplicative constant) with the measure $\mu_{\Oinv}$ induced by the left-invariant differential $n$-form $\Oinv$ on $G$.
\begin{remark}\label{intonG}
Let us clarify how the local formula~\eqref{eq.71} for the Haar measure $\mu$ on $G$ allows us to globally integrate a function on $G$.
Given $f\in \mathrm{C}_c(G)$, its Haar integral $\int_Gf(g)\de\mu(g)$ can be computed by splitting $f$ as a sum of its components on local supports contained in the domains of the charts in an atlas for $G$.
This is done by making use of a partition of unity $\{\chi_\alpha\}_{\alpha\in A}$ under an atlas $\{(\U_\alpha,\varphi_\alpha)\}_{\alpha\in A}$ of $G$. Then, the following relations hold:
\begin{equation}\label{eq.38}
\int_Gf(g)\de\mu(g)=\int_G\sum_{\alpha\in A} \chi_\alpha f(g)\de\mu(g) = \sum_{\alpha\in A}\int_{\U_\alpha}\chi_\alpha f(g)\de\mu(g).
\end{equation}
Each integral in the summation can be computed by using the local formulas~\eqref{eq.71}.
\end{remark}

\section{Applications}\label{sec5}
As previously observed (Proposition~\ref{prop:compatti}), the groups $\so(n,\mathbb{Q}_p)$, $n=2,3,4$, are compact. Hence, they admit a (left and right) Haar measure, which is essentially uniquely defined, i.e., unique up to a normalization constant factor. The construction of the Haar measure on $\so(2,\rat_p)_\kappa$ immediately follows by formula~\eqref{eq.71}. On the other hand, we will explicitly construct the Haar integrals on $\so(3,\rat_p)$ and $\so(4,\rat_p)$.  A fruitful approach is to introduce a suitable \emph{$p$-adic quaternion algebra}, $\pqal$, and exploit its relations with the $p$-adic special orthogonal groups in dimension three and four. In particular, we will prove that the latter groups can be realized as suitable quotients of the quaternion groups $\mathbb{H}_p^\times$ and $\mathbb{P}(\mathbb{H}_p^\times)$ respectively (c.f.\ Theorems~\ref{theo:SO3quat} and~\ref{theo:SO4quat}), whose Haar measures are determined, once again, by means of a direct application of~\eqref{eq.71}. Then, exploiting the Weil-Mackey-Bruhat formula introduced in Subsection~\ref{sec2n}, we will express the Haar integrals on $\mathrm{SO}(3,\mathbb{Q}_p)$ and $\mathrm{SO}(4,\mathbb{Q}_p)$ as lifts to the Haar integrals on the covering quaternion groups (see the forthcoming Theorems~\ref{theoso3} and~\ref{theoso4}).

\subsection{The Haar measure on \texorpdfstring{$\so(2,\mathbb{Q}_p)_\kappa$}{Lg}}\label{subsec.5.1}
In this subsection, we explicitly construct a left and right Haar measure on every $\so(2,\rat_p)_\kappa$, as in Corollary~\ref{cor:compunigru}. 

According to parameterization~\eqref{genelso2}, $\so(2,\rat_p)_\kappa$ is homeomorphic to the $p$-adic projective line, and it is covered by two disjoint charts. One coordinate map, say $\varphi_{(\kappa)}$, is defined on $\so(2,\rat_p)_\kappa\backslash\{-\mathrm{I}\}$ to $\rat_p$, and it is such that $\ha\varphi_{(\kappa)}(x)\equiv R(\alpha)$ (cf.\ Theorem~\ref{th.23}); the other one maps $-\mathrm{I}\in \so(2,\rat_p)_\kappa$ to $\infty$. Since the groups $\so(2,\rat_p)_\kappa$ are compact and infinite (uncountable), the singleton $\{-\mathrm{I}\}$ has \emph{zero} Haar measure. The Jacobian in~\eqref{eq.71} is now easily computed: by recalling the composition law~\eqref{comphar}, we find
\begin{equation}
\Big[\frac{\partial \zeta_{(\kappa)}}{\partial \beta}(\ha\vf_{(\kappa)}(\alpha);\beta)\Big]_{1\leq j,k\leq n}\equiv \frac{\de}{\de\beta}\bigg(\frac{\alpha+\beta}{1-\kappa\alpha\beta}\bigg)=\frac{1+\kappa\alpha^2}{(1-\kappa\alpha\beta)^2}.
\end{equation}
(Note that $-\kappa$ is never a square~\cite{our1st}, i.e., $1+\kappa\alpha^2\neq0$ for every $\alpha\in\mathbb{Q}_p$). Therefore, an application of~\eqref{eq.71} --- with $\beta=\varphi_{(\kappa)}(\mathrm{I})=0$ --- immediately yields the Haar measure of every Borel subset $\mathcal{E}$ in $\so(2,\mathbb{Q}_p)_\kappa$:
\begin{equation}\label{haarmes2}
\mu_2^{(\kappa)}(\mathcal{E})=\int_{\varphi_{(\kappa)}(\mathcal{E})}\frac{1}{|1+\kappa\alpha^2|_{p}}\de\lambda(\alpha),
\end{equation}
with $\de\lambda(\alpha)$ the Haar measure on $\rat_p$.
\begin{remark}
One can directly verify that the measure in~\eqref{haarmes2} is a Haar measure, i.e., left- and right-invariant. Indeed, let us consider the functional invariance condition in~\eqref{lefinv}:
\begin{equation}
\int\limits_{\so(2,\rat_p)_\kappa}\hspace{-4mm}  L_gf(x)\de\mu_2^{(\kappa)}(x)=\hspace{-2mm} \int\limits_{\alpha\in\rat_p}\hspace{-2mm}L_{R_\kappa(\beta)}f(R_\kappa(\alpha))\frac{\de\lambda(\alpha)}{|1+\kappa\alpha^2|_p}=\hspace{-2mm} \int\limits_{\alpha\in\rat_p}\hspace{-2mm}f(R_\kappa(-\beta)R_\kappa(\alpha))\frac{\de\lambda(\alpha)}{|1+\kappa\alpha^2|_p},
\end{equation}
for $f\in \mathrm{C}\big(\so(2,\rat_p)_\kappa\big)$ a compactly supported function on $\so(2,\rat_p)_\kappa$ (recall that $\mathrm{C}_c(X)=\mathrm{C}(X)$, whenever $X$ is compact), and where $g=R_\kappa(\beta)$, for some $\beta\in\rat_p$. In the last integral, we have also used the fact that $L_gf(x)=f(g^{-1}x)$ (i.e., the left translation of functions on $\so(2,\rat_p)_\kappa$), together with $R_\kappa(\beta)^{-1}=R_\kappa(-\beta)$. Recalling  formula~\eqref{comphar}, we have:
\begin{equation}\label{intver}
\int\limits_{\so(2,\rat_p)_\kappa}\hspace{-2mm}L_gf(x)\de\mu_2^{(\kappa)}(x) = \int\limits_{\alpha\in\rat_p}f\left(R_\kappa\left(\frac{\alpha-\beta}{1+\kappa\alpha\beta}\right)\right)\frac{1}{|1+\kappa\alpha^2|_p}\de\lambda(\alpha).
\end{equation}
Let us now set $\varpi=(\alpha-\beta)/(1+\kappa\alpha\beta)$. We have:
\begin{equation}\label{changever}
\alpha=\frac{\varpi+\beta}{1-\kappa\beta\varpi},\qquad \de\lambda(\alpha)=\frac{1+\kappa\beta^2}{(1-\kappa\varpi\beta)^2}\de\varpi,
\end{equation}
and, by inserting~\eqref{changever} into~\eqref{intver}, we obtain
\begin{align}
&\int\limits_{\varpi\in\rat_p}\hspace{-2mm}f(\varpi)\frac{|1-\kappa\varpi\beta|_p^2}{|(1-\kappa\varpi\beta)^2+\kappa(\varpi+\beta)^2|_p}\frac{|1+\kappa\beta^2|_p}{|1-\kappa\varpi\beta|_p^2}\de\varpi=\hspace{-2mm}\int\limits_{\varpi\in\rat_p}\hspace{-1.2mm}f(\varpi)\frac{|1+\kappa\beta^2|_p}{|1+\kappa\varpi^2+\kappa\beta^2+\kappa^2\varpi^2\beta^2|_p}\de\varpi\nonumber\\
&=\int\limits_{\varpi\in\rat_p}f(\varpi)\frac{|1+\kappa\beta^2|_p}{|(1+\kappa\varpi^2)(1+\kappa\beta^2)|_p}\de\varpi=\int\limits_{\varpi\in\rat_p}f(\varpi)\frac{1}{|1+\kappa\varpi^2|_p}\de\varpi=\int\limits_{\so(2,\rat_p)_\kappa}\hspace{-2mm}f(x)\de\mu_2^{(\kappa)}(x).
\end{align}
This shows the left-invariance of the measure. On the other hand, since the group is compact, this also entails the right-invariance of the measure~\eqref{haarmes2}.
\end{remark}

\begin{remark}\label{rem:Haarptreal}
The Haar measure of any Borel subset $\mathcal{F}$ of $\so(2,\mathbb{R})$ is given by
\begin{equation}
\mu(\mathcal{F})=\lambda(\varphi(\mathcal{F}))
,
\end{equation}
where $\lambda$ denotes the Haar measure on $\mathbb{R}$, and the coordinate map on $\so(2,\mathbb{R})$ is given by $\varphi
\begin{pmatrix}
\cos\theta & -\sin\theta\\\sin\theta&\cos\theta
\end{pmatrix}=\theta\in[0,2\pi[$. On the other hand, with $\kappa=1$ and $\alpha=\tan\big(\frac{\theta}{2}\big)$, an element of $\so(2,\mathbb{Q}_p)_1$ becomes formally identical to an element of $\so(2,\mathbb{R})$ (cf.\ Remark~\ref{oss:reductio}). Therefore, one may expect that such a `reduction’ applies also for the Haar measure. Indeed, using the C.O.V.F.\ for $p$-adic integrals (see Theorem~\ref{changeofvarint}) we have:
\begin{equation}
\mu_2^{(\kappa)}(\mathcal{E})=\int_{\varphi_{(\kappa)}(\mathcal{E})}\frac{1}{|1+\kappa\alpha^2|_{p}}\de\lambda(\alpha)\rightarrow \int_{\varphi(\mathcal{F})}\left\lvert\frac{1}{1+\tan^{2}(\theta/2)}\right\rvert_{p}\left\lvert\frac{1}{\cos^{2}(\theta/2)}\right\rvert_{p}\de\lambda(\theta)=\int_{\varphi(\mathcal{F})}\de\lambda(\theta),
\end{equation}
i.e., the Haar measure on $\so(2,\mathbb{Q}_p)_\kappa$ reduces to that on $\so(2,\mathbb{R})$, up to the normalization constant factor.
\end{remark}

\subsection{The quaternion algebra \texorpdfstring{$\pqal$}{Lg}}\label{sec:quatalp}
The study of real quaternions was originally motivated by their property to model Euclidean orthogonal transformations of $\mathbb{R}^3$ and $\mathbb{R}^4$~\cite{voight2021}. It turns out that this familiar picture keeps some of its main futures — but also requires some 
essential modifications — when switching from the real to the $p$-adic setting.  In what follows, we will describe the \emph{quaternion algebra} $\pqal$ over the field  $\rat_p$ of $p$-adic numbers~\cite{kochubei2001pseudo}, in a way that closely mimics its real counterpart (briefly reminded in Appendix~\ref{sec.1}); later, (cf.\ Subsection~\ref{sec:relquatrotp}), we shall clarify its relations with the $p$-adic special orthogonal groups in dimension three and four. The cases where $p>2$ and $p=2$ will be discussed separately.

\subsubsection{Case \texorpdfstring{$p>2$}{Lg}}
In the standard real case, the quaternion algebra $\qal$ is the vector space $\mathbb{R}^4\cong \mathbb{R}\times\mathbb{R}^3$ equipped with a suitable standard basis, namely, the one consisting of the vectors $1,\I,\J,\K$ in $\mathbb{R}^4$ satisfying the commutation rules~\eqref{comrel} of Appendix~\ref{sec.1}. From this, one can then define an isomorphism which realizes $\qal$ as a subalgebra of $\mathsf{M}_2(\mathbb{C})$.  Switching to the $p$-adic setting, it is then natural to set the following
\begin{definition}\label{quatalgodd}
Let $p>2$ be an odd prime. By a \emph{$p$-adic quaternion algebra} we mean a four-dimensional vector space $\pqal\cong \mathbb{Q}_p\times\mathbb{Q}_p^3$ over $\rat_p$ which is a $\rat_p$-algebra, and satisfies the following conditions:
\begin{enumerate}[label=\tt{(\alph*)}]
\item There exist $\I,\J$ in $\pqal$ such that, denoting by $1$ the multiplicative identity in $\pqal$, the set $\{1,\I,\J,\K\coloneqq\J\I\}$ is a $\rat_p$-basis in $\pqal$.
\item The basis vectors $\I,\J,\K$ in $\pqal$ satisfy the following commutation rules:
\begin{equation}
\I^2=v,\quad \J^2=-p,\quad \K^2=pv,\quad \J\I=-\I\J,\quad
 \K\I=-\I\K=v\J,\quad \K\J=-\J\K=p\I,
\end{equation}
for $v\in\rat_p$ a non-quadratic $p$-adic unit.
\end{enumerate}
\end{definition}
\begin{remark}\label{rem.5.2}
By means of a direct calculation, one verifies that the centre of the quaternion algebra $\pqal$ coincides with the base field $\rat_p$. This is reminiscent, to some extent, of the standard real case where, similarly, one shows that the field of real numbers $\mathbb{R}$ is the centre of the real quaternion algebra $\mathbb{H}$.
\end{remark}
On the quaternion algebra $\pqal$, we can define a natural \emph{involutive anti-automorphism} 
by setting
\begin{equation}\label{conjquat}
\pqal\ni\xi=q_0+\I q_1+\J q_2+\K q_3\mapsto \overline{\xi}\coloneqq q_0-\I q_1-\J q_2-\K q_3,\quad \xi\in\pqal.
\end{equation}
Then, it is easily checked that, for every $\xi\in\pqal$, the product of $\xi$ and $\overline{\xi}$ results into 
\begin{equation}
\xi\overline{\xi}=Q_{(4)}(q_0,q_1,q_2,q_3)=q_0^2-vq_1^2+pq_2^2-pvq_3^2, 
\end{equation}
that is, the unique (up to linear equivalence and scaling) four-dimensional definite quadratic form over $\mathbb{Q}_p$, for $p>2$ (cf.~\eqref{eq:quadrform4} in Theorem~\ref{quadform}). Therefore, we can express the inverse $\xi^{-1}$ of every (non-null) $p$-adic quaternion as
\begin{equation}\label{eq.40}
\xi^{-1}=\frac{\overline{\xi}}{Q_{(4)}(q_0,q_1,q_2,q_3)}.
\end{equation}
In what follows, we shall denote by 
\begin{equation}
\pqal^\times\coloneqq\{\xi\in\pqal\mid\xi\neq 0\}=\{\xi=q_0+\I q_1+\J q_2+\K q_3\in\pqal\mid Q_{(4)}(q_0,q_1,q_2,q_3)\neq 0\}
\end{equation}
the multiplicative group of \emph{invertible quaternions}.
\begin{remark}\label{redborm}
In the literature (e.g., see~\cite{voight2021}), the \emph{reduced norm} is defined as the map 
\begin{equation}\label{eq.50}
\mathbb{H}_p\ni\xi\mapsto\rn(\xi)\coloneqq\xi\overline{\xi}=Q_{(4)}(q_0,q_1,q_2,q_3)\in\mathbb{Q}_p.
\end{equation}
It is easily checked that $\rn$ is a multiplicative map; namely, $\rn(\xi\eta)=\xi\eta\overline{\xi\eta}=\xi\eta\overline{\eta}\overline{\xi}=\xi\rn(\eta)\overline{\xi}=\rn(\eta)\xi\overline{\xi}=\rn(\xi)\rn(\eta)$, for every $\xi,\eta\in \mathbb{H}_p$. Moreover, for every $\alpha\in\mathbb{Q}_p$ and $\xi\in\pqal$, 
$\rn(\alpha \xi)=\alpha^2\rn(\xi)$, and $\rn(\overline{\xi})=Q_{(4)}(q_0,-q_1,-q_2,-q_3)=Q_{(4)}(q_0,q_1,q_2,q_3)=\rn(\xi)$. In what follows, we shall denote by $\overline{\xi}/\rn(\xi)$ the inverse element~\eqref{eq.40} of a quaternion  $\xi\in\mathbb{H}_p^\times$.
\end{remark}

In the group of invertible quaternions $\pqal^\times$, it is possible to single out the subgroup of the so-called \emph{unit quaternions}, namely, the group:
\begin{equation}
\un(\pqal)\coloneqq\{\xi\in\pqal^\times\mid \xi^{-1}=\overline{\xi}\}\equiv\{\xi\in\pqal\mid\rn(\xi)=1\}.
\end{equation}

We want now to show that, as in the standard real case, $\pqal$ can be realized as a suitable matrix algebra. To begin with, we recall that in the quadratic form $Q_{(4)}(x)=x_0^2-vx_1^2+px_2^2-pvx_3^2$ on $\rat_p$,   $v\in\rat_p$ is a non-quadratic $p$-adic unit, i.e., $v\notin(\rat_p^\times)^2$ and $|v|_p=1$.  
Accordingly, we set $\mathbb{Q}_p(\sqrt{v})$ to denote the \emph{quadratic field extension} of $\mathbb{Q}_p$ by $\sqrt{v}$. Let $\mathsf{M}_2(\rat_p(\sqrt{v}))$ denote the algebra of two-dimensional matrices over $\rat_p(\sqrt{v})$, and let $\mpqal$ be the subalgebra of the matrices $M$ in $\mathsf{M}_2(\mathbb{Q}_p(\sqrt{v}))$ of the form
\begin{equation}
M=\begin{pmatrix}
x_0+\sqrt{v}x_1 & -x_2+\sqrt{v}x_3\\
p(x_2+\sqrt{v}x_3) & x_0-\sqrt{v}x_1
\end{pmatrix},
\end{equation}
where $x_i\in\mathbb{Q}_p$, $i=0,1,2,3$. It is easily checked that $\mpqal$ is a (unital) \emph{$\rat_p$-division algebra}, where the inverse of every non-null element $M\in\mpqal$ is 
\begin{equation}
M^{-1}=\frac{1}{\det(M)}
\begin{pmatrix}
x_0-\sqrt{v}x_1 & x_2-\sqrt{v}x_3\\
-p(x_2+\sqrt{v}x_3) & x_0+\sqrt{v}x_1
\end{pmatrix},
\end{equation}
and where $\det(M)=Q_{(4)}(x_0,x_1,x_2,x_3)$.
Let us now introduce the matrices $\I,\J,\K$ in $\mathsf{M}_2(\mathbb{Q}_p(\sqrt{v}))$ defined as
\begin{equation}\label{maquat}
\I\coloneqq
\begin{pmatrix}
\sqrt{v} & 0\\
0 &-\sqrt{v}
\end{pmatrix},
\quad\J\coloneqq
\begin{pmatrix}
0 & -1\\
p & 0
\end{pmatrix},
\quad\K\coloneqq
\begin{pmatrix}
0 & \sqrt{v}\\
p\sqrt{v} & 0
\end{pmatrix}.
\end{equation}
It is clear that every $M$ in $\mpqal$ can be expressed as follows:
\begin{equation}
M=x_0+\I x_1+\J x_2+ \K x_3,\quad x_0,x_1,x_2,x_3\in\rat_p,
\end{equation}
(here, we are omitting the identity matrix $\mathrm{I}_2$ multiplying $x_0$); that is, $\mpqal$ coincides with the \emph{$\mathbb{Q}_p$-linear span} of the set $\{\mathrm{I}_2,\I,\J,\K\}$. Moreover, $\I, \J, \K$ satisfies the following commutation rules:
\begin{align}\label{commrule} \I^2=v\mathrm{I}_2,\quad \J^2=-p\mathrm{I}_2,\quad \K^2=pv\mathrm{I}_2,\quad \J\I=-\I\J=\K,\quad
 \K\I=-\I\K=v\J,\quad \K\J=-\J\K=p\I,
 \end{align} 
from which we can argue that $\mpqal$ is a \emph{non-commutative} $\rat_p$-division algebra.
\begin{remark}\label{rem4.1}
As in the complex case, the subset of invertible elements in $\mpqal$ forms a group
\begin{equation}
\mpqal^\times \coloneqq \{M\in \mpqal\mid M\neq 0_{2}\}=\{M\in \mpqal\mid \det(M)\neq0\},
\end{equation}
where $0_2$ denotes the null $2\times 2$ matrix on $\rat_p(\sqrt{v})$. Moreover, we can single out the subgroup $\un(\mpqal)$ of elements in $\mpqal^\times$ having unit determinant, i.e.,
\begin{equation}
\un(\mpqal)\coloneqq\{M\in\mpqal^\times\mid\det(M)=1\},
\end{equation}
which provides the $p$-adic counterpart of~\eqref{eq:uintquat2} in Appendix~\ref{sec.1}.
\end{remark}
In the light of the discussion above, it is now not difficult to prove the following result:
\begin{proposition}\label{prop.4.7q}
For every prime $p>2$, the $p$-adic quaternion algebra $\pqal$ is isomorphic to the $\rat_p$-division subalgebra $\mpqal$ of $M_2(\rat_p(\sqrt{v}))$.
\end{proposition}
\begin{proof}
Let us consider the map 
\begin{equation}\label{eq:invpquad} \theta_p:\mathbb{H}_p\ni
\xi=q_0+\mathbf{i}q_1+\mathbf{j}q_2+\mathbf{k}q_3\mapsto \theta_p(\xi)\coloneqq \begin{pmatrix} q_0+\sqrt{v}q_1 & -q_2+\sqrt{v}q_3\\ p(q_2+\sqrt{v}q_3) & q_0-\sqrt{v}q_1 \end{pmatrix}\in\mathbf{H}_p.
\end{equation}
It is clear that $\theta_p$ is one-one, onto and linear, i.e., it is an isomorphism of vector spaces. Also, $\theta_p$ is a ring homomorphism, since $\theta_p(\xi\eta)=\theta_p(\xi)\theta_p(\eta)$ for every $\xi,\eta\in\mathbb{H}_p$. Hence, it defines an \emph{algebra isomorphism} from $\pqal$ to $\mpqal$.
\end{proof}
The algebra isomorphism $\theta_p$ identifies the basis vectors $1,\I,\J,\K$ of $\pqal$ with $\rm{I}_2$ and the matrices~\eqref{maquat} in the spanning set of $\mpqal$, respectively. This then also justifies our abuse of notation in using the same symbols for the basis elements of both $\pqal$ and $\mpqal$.

\begin{remark}\label{rem:center}
 Exploiting the algebra isomorphism $\pqais$, one can easily check that
 \begin{equation}
\rn(\xi)=\det(\pqais(\xi))=Q_{(4)}(q_0,q_1,q_2,q_3).
\end{equation}
Therefore, we can interchangeably use $\rn(\xi)$, $\det(\pqais(\xi))$ and $Q_{(4)}(q_0,q_1,q_2,q_3)$ to denote the reduced norm of $\xi=q_0+\I q_1+\J q_2+\K q_3$ in $\pqal$. 
\end{remark}
\begin{remark}
Using the isomorphism $\theta_p$, 
it is clear that the subgroups $\un(\pqal)$ and $\mathbb{H}_p^\times$ of $\mathbb{H}_p$ are isomorphic, respectively, to the subgroups $\un(\mpqal)$ and $\mpqal^\times$ of $\mpqal$ (cf.\ Remark~\ref{rem4.1}).
\end{remark}
\subsubsection{Case \texorpdfstring{$p=2$}{Lg}}
As for the $p>2$ case, we start by giving the following
\begin{definition}
Let $p=2$. By a \emph{$2$-adic quaternion algebra} we mean a four-dimensional vector space $\mathbb{H}_2\cong\rat_2\times\rat_2^3$ over $\rat_2$ which is a $\rat_2$-algebra, and satisfies the following conditions:
\begin{enumerate}[label=\tt{(\alph*)}]
\setcounter{enumi}{2}
\item There exist $\I,\J$ in $\mathbb{H}_2$ such that, denoting by $1$ the multiplicative identity in $\mathbb{H}_2$, the set $\{1,\I,\J,\K\}$ is a $\rat_2$-basis in $\mathbb{H}_2$.
\item The basis vectors $\I,\J,\K$ satisfy the following commutation rules:
\begin{equation}
\I^2=\J^2=\K^2=-1,\quad \I\J=-\J\I=\K,\quad \J\K=-\K\J=\I,\quad 
\K\I=-\I\K=\J.
\end{equation}
\end{enumerate}
\end{definition}
We can endow $\mathbb{H}_2$ with the involution~\eqref{conjquat}, thus turning it into an  involutive algebra. Then,  the inverse $\xi^{-1}$  of every non-null $2$-adic quaternion $\xi$ can be expressed as
\begin{equation}
\xi^{-1}=\frac{\overline{\xi}}{\rn(\xi)}.
\end{equation}
Moreover, we can single out the subgroup $\mathbb{H}_2^\times\leq\mathbb{H}_2$ 
of invertible $2$-adic quaternions by putting
\begin{equation}
\mathbb{H}_2^\times=\{\xi\in\mathbb{H}_2\mid\xi\neq 0\}\equiv\{\xi\in\mathbb{H}_2\mid\rn(\xi)\neq 0\},
\end{equation}
as well as the subgroup $\un(\mathbb{H}_2)\leq \mathbb{H}_2^\times$ of unit quaternions defined as
\begin{equation}
\un(\mathbb{H}_2)=\{\xi\in\mathbb{H}_2^\times\mid \rn(\xi)=1\}.
\end{equation}

We want now prove that $\mathbb{H}_2$ can be made in a one to one correspondence with a suitable matrix algebra. To this end, we recall that the definite quadratic form of $\mathbb{Q}_2^4$ is now given by~\eqref{eq:quadrform4}; moreover, since  $-1$ is not a square in $\rat_2$, we can consider the quadratic extension $\rat_2(\sqrt{-1})$ of $\rat_2$ by $-1$. Let $\mathsf{M}_2(\rat_2(\sqrt{-1}))$ denote the algebra of two-dimensional matrices on $\rat_2(\sqrt{-1})$, and let $\mathbf{H}_2\subset\mathsf{M}_2(\mathbb{Q}_2(\sqrt{-1}))$ be the subalgebra of matrices $M$ defined by 
\begin{equation}\label{eq:matrK2}
M\coloneqq\begin{pmatrix}
x_0+\sqrt{-1}x_1 & x_2+ \sqrt{-1}x_3\\
-x_2+\sqrt{-1}x_3 & x_0-\sqrt{-1}x_1
\end{pmatrix},
\quad x_i\in\mathbb{Q}_2,\ \forall i=0,\dots,3.
\end{equation}
By construction, we have that $\det(M)=Q_{(4)}(x_0,x_1,x_2,x_3)=x_0^2+x_1^2+x_2^2+x_3^2$. Hence, every non-zero $M\in\mathbf{H}_2$ is invertible, with inverse given by
\begin{equation}
M^{-1}=\frac{1}{\det(M)}
\begin{pmatrix}
x_0-\sqrt{-1}x_1 & -x_2- \sqrt{-1}x_3\\
x_2-\sqrt{-1}x_3 & x_0+\sqrt{-1}x_1
\end{pmatrix};
\end{equation}
i.e., $\mathbf{H}_2$ is an associative (unital) $\mathbb{Q}_2$-division algebra. Next, let us introduce the matrices $\I,\J,\K$ in $\mathsf{M}_2(\rat_2(\sqrt{-1}))$ defined by
\begin{equation}
\I\coloneqq
\begin{pmatrix}
\sqrt{-1} & 0\\
0 & -\sqrt{-1}
\end{pmatrix},
\quad\J\coloneqq
\begin{pmatrix}
0 & 1\\
-1 & 0
\end{pmatrix},
\quad\K\coloneqq
\begin{pmatrix}
0 & \sqrt{-1}\\
\sqrt{-1} & 0
\end{pmatrix}.
\end{equation}
Every $M$ in $\mathbf{H}_2$ can be expressed as $M=x_0+\I x_1+\J x_2+\K x_3$ (we have omitted the identity $\mathrm{I}_2$ multiplying $x_0$); that is, $\mathbf{H}_2$ can be realized as the $\rat_2$-linear span of $\{\mathrm{I}_2,\I,\J,\K\}$.
By further noting that $\I,\J,\K$ satisfy the commutation rules
\begin{equation}\label{comrel2}
\mathbf{i}^2=\mathbf{j}^2=\mathbf{k}^2=-\mathrm{I}_2,\quad \mathbf{i}\mathbf{j}=-\mathbf{j}\mathbf{i}=\mathbf{k},\quad \mathbf{j}\mathbf{k}=-\mathbf{k}\mathbf{j}=\mathbf{i},\quad 
\mathbf{k}\mathbf{i}=-\mathbf{i}\mathbf{k}=\mathbf{j},
\end{equation}
we also see that $\mathbf{H}_2$ is a \emph{non-commutative} $\rat_2$-algebra.
\begin{remark}
As in the $p>2$ case, we can introduce the group 
\begin{equation}
\mathbf{H}_2^\times\coloneqq\left\{M\in \mathbf{H}_2\mid M\neq 0_2\right\}\equiv\left\{M\in \mathbf{H}_2\mid \det(M)\neq 0\right\}
\end{equation}
of the invertible matrices in $\mathbf{H}_2$, as well as the subgroup 
\begin{equation}
\un(\mathbf{H}_2)=\left\{M\in \mathbf{H}_2^\times\mid \det(M)=1\right\}.
\end{equation}
\end{remark}
The following result is a straightforward adaptation of Proposition~\ref{prop.4.7q}
\begin{proposition}
Let $p=2$. Then, the $2$-adic quaternion algebra $\mathbb{H}_2$ is isomorphic with the subalgebra $\mathbf{H}_2$ of $\mathsf{M}_2(\rat_2(\sqrt{-1}))$.
\end{proposition}
\begin{proof}
It suffices to consider the map
\begin{equation}\label{eq:padunitquat} \theta_2\colon \mathbb{H}_2\ni\xi=q_0+\mathbf{i}q_1+\mathbf{j}q_2+\mathbf{k}q_3\mapsto \theta_2(\xi)\coloneqq \begin{pmatrix} q_0+\sqrt{-1}q_1 & q_2+ \sqrt{-1}q_3\\ -q_2+\sqrt{-1}q_3 & q_0-\sqrt{-1}q_1 \end{pmatrix} \in\mathbf{H}_2,
\end{equation}
and observe that it provides the desired algebra isomorphism.
\end{proof}
\begin{remark}\label{remqudfrpnorm}
The quaternion algebra $\mathbb{H}_2$ shares some analogies with the standard real quaternion algebra $\qal$. In particular, the matrix representation of a $2$-adic quaternion is `essentially the same’ as in the standard case (just set $\sqrt{-1}\coloneqq i$ for the square root of the non quadratic element $-1\in \rat_2$). This is what one expects upon considering the `formal equivalence’ of the real four-dimensional quadratic form $Q_{\mathbb{R}}$ with the four-dimensional quadratic form $Q_{(4)}$ on $\rat_2$. However, the analogies between standard and $p$-adic quaternion algebras cannot be pursued too far. Indeed, a fundamental difference between $\mathbb{H}_p$, for every prime $p\geq2$, and $\mathbb{H}$ is the following. For the latter, we have that $Q_{\mathbb{R}}(q_0,q_1,q_2,q_3)=\|(q_0,q_1,q_2,q_3)\|_{\mathbb{R}^4}^2$, i.e., the definite quadratic form $Q_{\mathbb{R}}$ on $\mathbb{R}^4$ \emph{coincides} with the squared Euclidean norm of $\mathbb{R}^4$. (This also entails that the reduced norm of $\qal$ is equivalent to the (square of) the Euclidean norm of $\mathbb{R}^4$. See Remarks~\ref{redborm} and~\ref{rem:reducednorm}). 
On the other hand, in the $p$-adic setting, we only have the equivalence $Q_{(4)}\equiv\rn$, i.e., the reduced norm of $\pqal$ \emph{does not coincide} with the square of the $p$-adic norm of $\mathbb{Q}_p^4$.
\end{remark}

\subsection{Relation between \texorpdfstring{$p$}{Lg}-adic quaternions and special orthogonal groups}\label{sec:relquatrotp}
This subsection clarifies the relations between $p$-adic quaternions and the $p$-adic groups of rotations $\so(3,\mathbb{Q}_p)$ and $\so(4,\mathbb{Q}_p)$, for every $p\geq 2$.
We begin with $\so(3,\mathbb{Q}_p)$. Let us consider the action by \emph{conjugation} of the group $\pqal^\times$ of invertible quaternions on $\pqal$; namely, the map 
\begin{equation}\label{eq:actionk}
\pqal\ni \eta\mapsto\xi\eta\xi^{-1}\in\pqal
,
\end{equation}
where $\xi\in\pqal^\times$, and $p\geq 2$.
This map is an \emph{isometric linear transformation} of $\pqal$, 
since it preserves the reduced norm of every quaternion $\eta$ in $\pqal$:
\begin{align}
\rn(\xi\eta\xi^{-1})
&=\rn(\xi)\rn(\eta)\rn(\xi^{-1})\nonumber\\
&=\rn(\xi)\rn(\xi^{-1})\rn(\eta)\nonumber\\&
=\rn(\xi\xi^{-1})\rn(\eta)\nonumber\\
&=\rn(\eta);\label{eq:isoQdet}
\end{align}
equivalently, the action by conjugation of $\pqal^\times$ preserves the definite quadratic form of $\mathbb{Q}_p^4$.
Mo\-reover, the operation $\eta\mapsto\xi\eta\xi^{-1}$  leaves the centre $\mathbb{Q}_p$ of $\pqal$ pointwise fixed and, hence, also leaves the orthogonal subspace $\mathbb{Q}_p^3$ invariant.

(Note: Here, we refer to the orthogonality w.r.t.\ the inner product induced by the definite quadratic form of $\mathbb{Q}_p^4$, as defined in~\eqref{eq:sclprodQbil}).

Let us now consider the restriction of~\eqref{eq:actionk} to the subset $\pqal^0\coloneqq\{\nu\in\pqal\mid \nu=\I q_1+\J q_2+\K q_3\}$ of \emph{pure imaginary quaternions} in $\pqal$; that is, let us consider the map
\begin{equation}
\pqrot(\xi)\colon \pqal^0\ni\nu\mapsto\pqrot(\xi)\nu\coloneqq \xi\nu\xi^{-1},\quad\xi\in\pqal^\times.
\end{equation}
By noting that $\pqal^0\cong\rat_p^3$, and reminding that the action~\eqref{eq:actionk} is an isometric transformation of $\pqal$, we deduce that $\pqrot(\xi)$ preserves the restriction of $Q_{(4)}$ to $\mathbb{Q}_p^3$, i.e., the (equivalent) quadratic form $Q_+$ (see Remark~\ref{rem:equivrestricQ}). Hence, we deduce that $\pqrot(\xi)\in \mathrm{O}(3,\mathbb{Q}_p)\cong \{L\in\text{End}(\mathbb{Q}_p^3)\mid Q_+(L\mathbf{x})=Q_+(\mathbf{x}),\ \forall \mathbf{x}\in \mathbb{Q}_p^3\}$ represents an orthogonal transformation in $\rat_p^3$. Next, by observing that, for every $\xi,\rho\in\pqal$ and $\nu\in\pqal^0$, the equalities $\pqrot(\xi\rho)\nu=(\xi\rho)\nu(\xi\rho)^{-1}=\xi\big(\rho \nu\rho^{-1}\big)\xi^{-1}=\pqrot(\xi)\pqrot(\rho)\nu$ hold, we can conclude that $\pqrot:\pqal^\times\rightarrow \mathrm{O}(3,\mathbb{Q}_p)$ provides a group homomorphism. 

Let us now explicitly derive its action on a pure imaginary quaternion $\nu$ in $\pqal^0$. If $\xi=q_0+\I q_1+\J q_2+\K q_3\in\pqal^\times$ and $\nu=\I s_1+\J s_2+\K s_3\in\pqal^0$, the action of $\pqrot(\xi)$ on $\nu$ is given by
\begin{equation}
\xi \nu \xi^{-1}=(q_0+\I q_1+\J q_2+\K q_3)(\I s_1+\J s_2+\K s_3)(q_0-\I q_1-\J q_2-\K q_3)\frac{1}{\rn(\xi)},
\end{equation}
where we have used the fact that $\xi^{-1}=\overline{\xi}/\rn(\xi)$ (see Remark~\ref{redborm}).
Expanding the above products, one sees that the scalar part vanishes, as expected, and, by collecting the terms in $\I,\J$ and $\K$, 
we get
\begin{equation}\label{eq:matrkxi}
\pqrot(\xi)=\frac{1}{\rn(\xi)}
\begin{pmatrix}
q_0^2-vq_1^2-pq_2^2+pvq_3^2 & 2p(q_0q_3+q_1q_2) & -2p(q_0q_2+vq_1q_3)\\
2v(q_0q_3-q_1q_2) & q_0^2+vq_1^2+pq_2^2+pv q_3^2 & -2v(q_0q_1+pq_2q_3)\\
2(q_0q_2-vq_1q_3) & 2(-q_0q_1+pq_2q_3) & q_0^2+vq_1^2-pq_2^2-pvq_3^2
\end{pmatrix}
\end{equation}
for $p>2$, and
\begin{equation}\label{eq:kappaxi2}
\kappa_2(\xi)=\frac{1}{\rn(\xi)}
\begin{pmatrix}
q_0^2+q_1^2-q_2^2-q_3^2 & 2(q_1q_2-q_3q_0) & 2(q_2q_0+q_3q_1)\\
2(q_1q_2+q_0q_3) & q_0^2-q_1^2+q_2^2-q_3^2 & 2(q_2q_3-q_1q_0)\\
2(q_1q_3-q_2q_0) & 2(q_1q_0+q_2q_3) & q_0^2-q_1^2-q_2^2+q_3^2
\end{pmatrix}
\end{equation}
for $p=2$. A direct calculation shows that the transformations~\eqref{eq:matrkxi} and~\eqref{eq:kappaxi2} have unit determinant, i.e., 
\begin{equation}\label{eq:isoQdet2}
\det(\pqrot(\xi))=\frac{1}{\rn(\xi)^3}Q_{(4)}(q_0,q_1,q_2,q_3)^3=1.
\end{equation}
Therefore, we get to the conclusion that, for every prime $p\geq2$, and every $\xi\in\mathbb{H}_p^\times$, $\pqrot(\xi)\in \so(3,\mathbb{Q}_p)$ is a \emph{three-dimensional $p$-adic rotation}. 

The above discussion shows that $\kappa_p(\mathbb{H}_p^\times)\subseteq \so(3,\mathbb{Q}_p)$.
We are now going to prove that, actually, also the reverse inclusion $\so(3,\mathbb{Q}_p)\subseteq \kappa_p(\mathbb{H}_p^\times)$ holds. Indeed, let us first introduce the map 
$\tau_\rho\colon\pqal^0\rightarrow\pqal^0$ defined, for every $\rho\in\pqal^\times\cap \pqal^0$, as
\begin{equation}\label{eq:hyperreflec}
\tau_\rho(\nu)\coloneqq\nu-\frac{2b(\nu,\rho)}{\rn(\rho)}\rho,
\end{equation}
where $b$ denotes the bilinear form associated with the quadratic form $Q_+$ in $\rat_p^3$ (cf.\ Section~\ref{sec2}). It is easily shown that this map satisfies the conditions $\tau_\rho(\rho)=-\rho$ and $Q_+(\tau_\rho(\nu))=Q_+(\nu)$, for any $\nu\in\pqal^0\cong\rat_p^3$; namely, $\tau_\rho\in \mathrm{O}(3,\rat_p)\setminus \so(3,\rat_p)$ defines a hyperplane reflection (w.r.t.\ $\rho$) in $\pqal^0$. 
Moreover, by taking into account the defining properties of $b$, $\rn$ and $\pqal^0$, and recalling that, for a pure imaginary quaternion $\nu$ in $\pqal^0$, one has $\overline{\nu}=-\nu$, we see that the reflection~\eqref{eq:hyperreflec}
is explicitly given by
$\tau_\rho(\nu)=-\rho\nu\rho^{-1}$, i.e. $\tau_\rho\equiv-\kappa_p(\rho)$. On the other hand, by a classical theorem of Cartan and Dieudonné (cf.\ Theorem~4.5.7.\ in~\cite{voight2021}),  every special orthogonal transformation in $\so(3,\rat_p)$ can be written as the composition of two such reflections, i.e. $g=\tau_{\rho_1}\tau_{\rho_2}$, for all $g\in \so(3,\rat_p)$, and suitable $\rho_1,\rho_2\in\pqal^\times\cap \pqal^0$. Therefore, every $p$-adic rotation in $\so(3,\rat_p)$ is expressed by
\begin{equation}
g=\tau_{\rho_1}\tau_{\rho_2}=(-\tau_{\rho_1})(-\tau_{\rho_2})=\kappa_p(\rho_1)\kappa_p(\rho_2)=\kappa_p(\rho_1\rho_2)=\kappa_p(\xi),
\end{equation}
for $\xi\coloneqq\rho_1\rho_2\in\pqal^\times$. This then shows that
$\pqrot(\pqal^\times)=\so(3,\rat_p)$, i.e. that $\kappa_p$ is surjective.

The following result is now straightforward and crucial for our purposes.
\begin{theorem}\label{theo:SO3quat} The group $\so(3,\mathbb{Q}_p)$ is isomorphic to the quotient of the group $\mathbb{H}_p^\times$ of invertible quaternions, and the multiplicative group $\rat_p^\times$ of non-null elements in $\rat_p$, namely
\begin{equation}\label{eq:isoquatrot}
    \so(3,\mathbb{Q}_p)\cong \mathbb{H}_p^\times/\mathbb{Q}_p^\times.
\end{equation}
\end{theorem}
\begin{proof}
To prove the group isomorphism~\eqref{eq:isoquatrot}, we can equivalently show that the following 
\begin{equation}\label{eq:shortexact}
1\rightarrow \mathbb{Q}_p^\times \hookrightarrow \mathbb{H}_p^\times\ {\stackrel{\text{\upshape $\kappa_p$}}{\longrightarrow}}\ \so(3,\mathbb{Q}_p)\rightarrow 1
\end{equation}
is a short exact sequence. We already know that $\kappa_p$ is surjective. Furthermore, the kernel of $\pqrot$,  $\ker(\pqrot)$, coincides with the image $\rat_p^\times$ of the embedding in the short sequence:
\begin{align}
\ker(\kappa_p)&=\{\xi\in\mathbb{H}_p^\times\mid \kappa_p(\xi)=\mathrm{I}\in \so(3,\mathbb{Q}_p)\}\nonumber\\
& =\{\xi\in\mathbb{H}_p^\times\mid \kappa_p(\xi)\nu=\nu\ \textup{for every}\ \nu\in\pqal^0\}\nonumber\\
&=\{\xi\in\mathbb{H}_p^\times\mid \xi\nu=\nu\xi\ \textup{for every}\ \nu\in\pqal^0\}\nonumber\\
&= \{\xi\in\mathbb{H}_p^\times\mid \xi\rho=\rho\xi\ \textup{for every}\ \rho\in\mathbb{H}_p\}
 = \mathbb{Q}_p^\times,
\end{align}
as $\rat_p^\times$ is the centre of $\pqal^\times$ (see Remark~\ref{rem.5.2}).
\end{proof}
The exact sequence~\eqref{eq:shortexact} is reminiscent, to some extent, of the exact sequence
\begin{equation}\label{realexac}
1\rightarrow \{\pm1\}\hookrightarrow \un(\mathbb{H})\cong \mathrm{SU}(2,\mathbb{C}) \twoheadrightarrow \so(3,\mathbb{R})\rightarrow 1,
\end{equation}
of the standard real case (cf.\ the isomorphism~\eqref{eq:so3real2quat} in Appendix~\ref{sec.A2}). Here, the main difference with the sequence~\eqref{eq:shortexact} is provided by the fact that the groups $\un(\qal)$ and $\mathbb{F}_2=\{\pm 1\}$ are replaced, in the $p$-adic setting, by the groups $\pqal^\times$ and $\rat_p^\times$ respectively. The reason for this discrepancy is related to the peculiar features of the base field $\rat_p$. Indeed, it is possible to prove~\cite{voight2021,Lam} that a sequence as in~\eqref{realexac} is exact if and only if $\rn(\qal^\times)\subset (\mathbb{F}^\times)^2$, namely, iff the reduced norm of every invertible quaternion is a quadratic element of the field. In the case where $\mathbb{F}=\mathbb{R}$, this is certainly true. Instead, in the $p$-adic setting, $\rn(\pqal^\times)\subset (\rat_p^\times)^2$ is \emph{never} true.
\bigskip

We want now to show that $\so(3,\rat_p)$ and $\pqal^\times/\rat_p^\times$ are homeomorphic. This fact will indeed  play a fundamental role in our construction of the lift of the Haar integrals on $\so(3,\rat_p)$ to $\pqal^\times$. 

Let us preliminary recall that every LCSC Hausdorff space is a \emph{standard Borel space} once endowed with its Borel $\sigma$-algebra. Accordingly, one calls a space $X$ a \emph{standard Borel $G$-space} if $X$ is a $G$-space (cf.\ Subsection~\ref{sec2n}), its Borel structure is standard, and if the action of $G$ on $X$ is a \emph{Borel map}. If $X$ is a standard Borel $G$-space, and $x\in X$ is a fixed point, let $G_x\coloneqq\{g\in G\mid g[x]=x\}$ be the \emph{stability subgroup} at $x$. One can show (cf.\ Corollary~$5.8$ in~\cite{varadarajan}) that $G_x$ is a closed subgroup of $G$. Moreover, 
denoting by $q\colon G\rightarrow G/G_x$ the projection homomorphism, the map
\begin{equation}\label{eq:76} 
G/G_x\ni q(g)\mapsto g[x]\in X 
\end{equation}
is a Borel isomorphism, and it is a homeomorphism whenever $X$ is LCSC (cf.\ Theorem~$5.11$ in~\cite{varadarajan}). Therefore, in such a case, $X\cong G/G_x$ are homeomorphic spaces in a natural way. We are now ready to prove the following result.
\begin{proposition}\label{prop:omeoSO3}
The group isomorphism~\eqref{eq:isoquatrot} between $\so(3,\rat_p)$ and $\pqal^\times/\rat_p^\times$ is also an isomorphism of topological groups.
\end{proposition}
\begin{proof} The proof we give here is based on general measure-theoretical arguments on $G$-spaces; for a more specific proof, involving the reduced norm of $p$-adic quaternions, see Appendix~\ref{alternhomeo}.

As a vector space, $\pqal\cong \rat_p\times\rat_p^3=\rat_p^4$, and we can provide $\pqal$ with the product topology (on $\rat_p$, we consider the natural (ultra-)metric topology generated by the $p$-adic absolute value). Similarly, $\mathbb{H}_p^\times$ and $\mathbb{Q}_p^4-\{0\}$ are homeomorphic topological spaces whenever they are equipped with the induced topology as subspaces of $\mathbb{H}_p$ and $\mathbb{Q}_p^4$ respectively. The continuity of the group operations (multiplication and inverse) of $\mathbb{H}_p^\times$ is inherited from the continuity of the addition, inner multiplication (according to the commutation relations among the basis elements) and multiplication by scalars of $\mathbb{Q}_p^4-\{0\}$; therefore $\mathbb{H}_p^\times$ is a topological group. Also, $\mathbb{H}_p^\times$ is LCSC, as $\mathbb{Q}_p^4-\{0\}$ is so (being an open 
subspace of the locally compact Hausdorff space $\mathbb{Q}_p^4$). We have already observed that $\so(3,\rat_p)$ is a compact second countable Hausdorff group, once endowed with the topology introduced in Section~\ref{sec2}. Hence, $\so(3,\rat_p)$, supplied with its Borel $\sigma$-algebra is a standard Borel space. We want now to show that $\so(3,\rat_p)$ is a standard Borel $\mathbb{H}_p^\times$-space. To this end, we have to find a Borel action of $\mathbb{H}_p^\times$ on $\so(3,\mathbb{Q}_p)$. 

Let us introduce the map from $\pqal^\times\times \so(3,\rat_p)$ to $\so(3,\rat_p)$ defined as
\begin{equation}\label{eq:azione}
\pqal^\times\times \so(3,\rat_p)  \ni  (\xi,R)\mapsto \xi[R]\coloneqq \kappa_p(\xi)R\in \so(3,\rat_p).
\end{equation}
It is easily shown that the map~\eqref{eq:azione} provides a continuous left action of $\pqal^\times$ on $\so(3,\rat_p)$. Indeed, continuity follows from 
that of $\kappa_p$ and of the matrix multiplication in $\so(3,\rat_p)$. Next, we have that $\xi[\nu[R]]=\kappa_p(\xi)\big(\kappa_p(\nu)R\big)=\big(\kappa_p(\xi)\kappa_p(\nu)\big)R = \kappa_p(\xi\nu)R=(\xi\nu)[R]$, for every $\xi,\nu\in\mathbb{H}_p^\times$, $R\in \so(3,\mathbb{Q}_p)$. Moreover $R\mapsto \xi[R]$ is a homeomorphism for every fixed $\xi\in\mathbb{H}_p^\times$, as follows by observing that  $R\mapsto \xi[R]$ is surjective (since the multiplication in $\so(3,\mathbb{Q}_p)$ by the matrix $\kappa_p(\xi)$ is so), and injective (since if $\kappa_p(\xi)R_1=\kappa_p(\xi)R_2$, then $R_1=R_2$ by the invertibility of $\kappa_p(\xi)\in \so(3,\mathbb{Q}_p)$), and both the map and its inverse are continuous (as they are just matrix multiplications and inverses). This shows that~\eqref{eq:azione} is a continuous (actually, Borel) left action of $\mathbb{H}_p^\times$ on $\so(3,\mathbb{Q}_p)$. This action is also transitive, since it exists an element $R\in \so(3,\mathbb{Q}_p)$ such that its orbit $\{\kappa_p(\xi)R\mid \xi\in\mathbb{H}_p^\times\}$
is the whole space $\so(3,\rat_p)$ (it is enough to consider $R=\rm{I}$, and the surjectivity of $\kappa_p$). Therefore, we can argue that $\so(3,\rat_p)$ is a standard Borel (transitive) $\pqal^\times$-space. On the other hand, the stability subgroup at every $R\in \so(3,\mathbb{Q}_p)$ is given by $\{\xi\in\mathbb{H}_p^\times\mid \xi[R]=R\}=\{\xi\in\mathbb{H}_p^\times\mid \kappa_p(\xi)R=R\}=\{\xi\in\mathbb{H}_p^\times\mid \kappa_p(\xi)=\mathrm{I}\}=\ker(\kappa_p)=\rat_p^\times$; hence, 
we can conclude that $\so(3,\rat_p)$ and $\pqal^\times/\rat_p^\times$ are homeomorphic spaces. In particular, the homeomorphism is as in~\eqref{eq:76} with, for instance, the stability subgroup at $\mathrm{I}\in \so(3,\mathbb{Q}_p)$. Explicitly, the homeomorphism is 
$\pqal^\times/\rat_p^\times\ni \xi\mathbb{Q}_p^\times\mapsto \kappa_p(\xi)\in\so(3,\rat_p)$. This is, indeed, the same map providing the isomorphism in Theorem~\ref{theo:SO3quat}.
\end{proof}

\bigskip

Proposition~\ref{prop:omeoSO3} concludes our discussion on the relations between $p$-adic quaternions and rotations in $\so(3,\rat_p)$.
Now, we carry out a similar analysis to clarify the relation between quaternions and the elements in $\so(4,\mathbb{Q}_p)$. To begin with, let us introduce the left action of $\mathbb{H}_p^\times\times \mathbb{H}_p^\times$ on $\mathbb{H}_p$ defined by
\begin{equation}\label{eq:2leftactio}
\mathbb{H}_p\ni\eta\mapsto \xi\eta\varrho^{-1}\in \mathbb{H}_p,\quad (\xi,\varrho)\in \mathbb{H}_p^\times\times\mathbb{H}_p^\times.
\end{equation}
This action is by \emph{similarities}, as follows by noting that
\begin{align}  \rn(\xi\eta\varrho^{-1})&=\rn(\xi)\rn(\eta)\rn(\varrho^{-1})=\frac{\rn(\xi)}{\rn(\varrho)}\rn(\eta).
\end{align}
In particular, the action is by \emph{isometries} whenever $\rn(\xi)=\rn(\varrho)$. Hence, let us introduce the group \begin{equation}\label{eq:subgrpquatpairs}
\mathbb{P}(\pqal^\times)\coloneqq\{(\xi,\varrho)\in \mathbb{H}_p^\times\times \mathbb{H}_p^\times\mid \rn(\xi)=\rn(\varrho)\}.
\end{equation}
The restriction of the action~\eqref{eq:2leftactio} to a pair $(\xi,\varrho)\in\mathbb{P}(\mathbb{H}_p^\times)\leq\pqal^\times\times\pqal^\times$ is denoted by $\kappa_p'(\xi,\varrho)$; namely, we set
\begin{equation}
\kappa_p'(\xi,\varrho):\mathbb{H}_p\ni\eta\mapsto \kappa_p'(\xi,\varrho)\eta\coloneqq \xi\eta\varrho^{-1}\in \mathbb{H}_p,\quad (\xi,\varrho)\in\mathbb{P}(\mathbb{H}_p^\times).
\end{equation}
Since this action is by isometries, and $\mathbb{H}_p\cong \mathbb{Q}_p^4$, then $\kappa_p'(\xi,\varrho)\in \mathrm{O}(4,\mathbb{Q}_p)\cong\{L\in\text{End}(\mathbb{Q}_p^4)\mid Q_{(4)}(Lx)=Q_{(4)}(x),\ \text{for every}\ x\in \mathbb{Q}_p^4\}$. It can be easily checked that, as done for the maps $\kappa_p(\xi)$ in the three-dimensional case, $\kappa_p'(\xi,\varrho)\in \so(4,\mathbb{Q}_p)$, for every $(\xi,\varrho)\in\mathbb{P}(\mathbb{H}_p^\times)$. Also, $\kappa_p'\colon\mathbb{P}(\mathbb{H}_p^\times)\rightarrow \so(4,\mathbb{Q}_p)$ is a group homomorphism, and we get to the following result:
\begin{theorem}\label{theo:SO4quat}
The group $\so(4,\mathbb{Q}_p)$ is isomorphic to the quotient of the group $\mathbb{P}(\mathbb{H}_p^\times)$ and the multiplicative group $\mathbb{Q}_p^\times$ of non-null $p$-adic numbers:
\begin{equation}\label{so4iso}
\so(4,\mathbb{Q}_p)\cong \mathbb{P}(\pqal^\times)/\rat_p^\times.
\end{equation}
\end{theorem}
\begin{proof}
Since $\mathrm{char}(\rat_p)\neq 2$, the isomorphism~\eqref{so4iso} follows from Proposition 4.5.17.\ in~\cite{voight2021}.  In particular, to prove~\eqref{so4iso}, it suffices to show that the following 
\begin{equation}\label{eq:shortexact4}
1\rightarrow\rat_p^\times\hookrightarrow\mathbb{P}(\pqal^\times)\ {\stackrel{\text{\upshape $\kappa_p'$}}{\longrightarrow}}\ \so(4,\rat_p)\rightarrow 1
\end{equation}
is a short exact sequence. This is done similarly to the proof of Theorem~\ref{theo:SO3quat}: Surjectivity of the map $\kappa_p'\colon\mathbb{P}(\pqal^\times)\rightarrow \so(4,\rat_p)$ again 
 follows by the Cartan-Dieudonné Theorem (cf.\ Theorem~4.5.7.\ in~\cite{voight2021}), and its kernel is
 \begin{align}
\ker(\kappa_p')&=\{(\xi,\varrho)\in\mathbb{P}(\mathbb{H}_p^\times)\mid \kappa_p'(\xi,\varrho)=\mathrm{I}\in \so(4,\mathbb{Q}_p)\}\nonumber\\
& =\{(\xi,\varrho)\in\mathbb{P}(\mathbb{H}_p^\times)\mid \kappa_p'(\xi,\varrho)\eta=\eta\ \textup{for every}\ \eta\in\pqal\}\nonumber\\
&=\{(\xi,\varrho)\in\mathbb{P}(\mathbb{H}_p^\times)\mid \xi\eta=\eta\varrho\ \textup{for every}\ \eta\in\pqal\}.
\end{align}
In particular, the last condition must hold for $\eta=1\in\mathbb{H}_p$, providing the necessary condition $\xi=\varrho$; hence,
\begin{align}
\ker(\kappa_p')&=\{(\xi,\xi)\in\mathbb{P}(\mathbb{H}_p^\times)\mid \xi\eta=\eta\xi\ \textup{for every}\ \eta\in\pqal\}\nonumber\\
& \cong \{\xi\in\mathbb{H}_p^\times\mid \xi\eta=\eta\xi\ \textup{for every}\ \eta\in\pqal\} = \mathbb{Q}_p^\times.
\end{align}
That is, the kernel of $\kappa_p'$ is the diagonally embedded $\mathbb{Q}_p^\times\cong \mathbb{Q}_p^\times(1,1)$ in $\mathbb{P}(\mathbb{H}_p^\times)$.
\end{proof}
\begin{remark}\label{rem:standquatcov}
The short exact sequences~\eqref{eq:shortexact4} is the $p$-adic counterpart of the following sequence for the standard real setting:
\begin{equation}\label{realseq}
\quad 1\rightarrow\{\pm1\}\hookrightarrow \un(\mathbb{H})\times \un(\mathbb{H})\twoheadrightarrow \so(4,\mathbb{R})\rightarrow 1,
\end{equation}
where $\un(\qal)$ denotes the group of unit quaternions (see~\eqref{eq:uintquat} in Appendix~\ref{sec.1}). This then entails the well known group isomorphism~\eqref{eq:so4real2quat}. The main difference with the $p$-adic case is provided by the fact that $\un(\qal)\times \un(\qal)$ and $\mathbb{F}_2=\{\pm 1\}$ are now replaced by $\mathbb{P}(\pqal^\times)$ and $\rat_p^\times$ respectively. Once again, this discrepancy is a consequence of the fact that in the $p$-adic setting, $\rn(\pqal^\times)\not\subset (\rat_p^\times)^2$.
\end{remark}

Similarly to what we did for $\so(3,\mathbb{Q}_p)$, we are interested in proving that $\so(4,\mathbb{Q}_p)$ and $\mathbb{P}(\mathbb{H}_p^\times)/\mathbb{Q}_p^\times$ are homeomorphic; this will allow us to consider the lift of the Haar integrals on $\so(4,\mathbb{Q}_p)$ to that on $\mathbb{P}(\mathbb{H}_p^\times)$.
\begin{proposition}\label{prop:omeoSO4}
The group isomorphism~\eqref{so4iso} between $\so(4,\rat_p)$ and $\mathbb{P}(\mathbb{H}_p^\times)/\mathbb{Q}_p^\times$ is also an isomorphism of topological groups.
\end{proposition}
\begin{proof}
Consider the group $\mathbb{P}(\mathbb{H}_p^\times)$ with the subspace topology induced by $\mathbb{Q}_p^8$ (the latter, being endowed with $p$-adic topology). The group operations are continuous, hence $\mathbb{P}(\mathbb{H}_p^\times)$ is a topological group. It is also Hausdorff and second countable, being a subspace of the Hausdorff second countable space $\mathbb{Q}_p^8$. In addition, $\mathbb{P}(\mathbb{H}_p^\times)$ is a closed subspace of the locally compact Hausdorff space $\mathbb{Q}_p^8$, hence it is locally compact as well. We are now going to show that, actually, $\so(4,\mathbb{Q}_p)$ is a standard Borel $\mathbb{P}(\mathbb{H}_p^\times)$-space. The group $\so(4,\mathbb{Q}_p)$ with $p$-adic topology is compact, second countable and Hausdorff. Thus, $\so(4,\mathbb{Q}_p)$ along with its Borel $\sigma$-algebra is a standard Borel space. Let us introduce the map
\begin{equation}\label{eq:aazionee}
\mathbb{P}(\mathbb{H}_p^\times)\times \so(4,\mathbb{Q}_p) \ni \big((\xi,\rho),\,R\big)\mapsto (\xi,\rho)[R]\coloneqq  \kappa_p'(\xi,\rho) R \in \so(4,\mathbb{Q}_p).
\end{equation}
This map is continuous, and such that $(\xi,\rho)\big[(\nu,\eta)[R]\big] =\kappa_p'(\xi,\rho)\big(\kappa_p'(\nu,\eta)R\big)=\kappa_p'\big((\xi,\rho)(\nu,\eta)\big)R= \big((\xi,\rho)(\nu,\eta)\big)[R]$, for every $(\xi,\rho),(\nu,\eta)\in \mathbb{P}(\mathbb{H}_p^\times)$, $R\in \so(4,\mathbb{Q}_p)$ (here, we have used the fact $\kappa_p'$ is a homomorphism). Moreover, the map $R\mapsto (\xi,\rho)[R]$ is a homeomorphism, for every fixed $(\xi,\rho)\in \mathbb{P}(\mathbb{H}_p^\times)$. Therefore, the map~\eqref{eq:aazionee} is an action of $\mathbb{P}(\mathbb{H}_p^\times)$ on $\so(4,\mathbb{Q}_p)$, which is transitive by surjectivity of $\kappa_p'$. Actually,  it is also a Borel map and, hence, $\so(4,\mathbb{Q}_p)$ is a standard Borel $\mathbb{P}(\mathbb{H}_p^\times)$-space. Now, we observe that the stability subgroup at any $R\in \so(4,\mathbb{Q}_p)$ is $\{(\xi,\rho)\in \mathbb{P}(\mathbb{H}_p^\times)\mid (\xi,\rho)[R]=R\}=\{(\xi,\rho)\in \mathbb{P}(\mathbb{H}_p^\times)\mid \kappa_p'(\xi,\rho)R=R\}=\{(\xi,\rho)\in \mathbb{P}(\mathbb{H}_p^\times)\mid \kappa_p'(\xi,\rho)=\mathrm{I}\}=\ker(\kappa_p')=\mathbb{Q}_p^\times$. Thus, we can argue that $\mathbb{P}(\mathbb{H}_p^\times)/\mathbb{Q}_p^\times$ and $\so(4,\mathbb{Q}_p)$ are homeomorphic, the homeomorphism being provided, once again, by~\eqref{eq:76}. In particular, if we consider the stability subgroup, for instance, at $\mathrm{I}\in \so(4,\mathbb{Q}_p)$, the homeomorphism is explicitly given by $\mathbb{P}(\mathbb{H}_p^\times)/\mathbb{Q}_p^\times\ni (\xi,\rho)\mathbb{Q}_p^\times \mapsto \kappa_p'(\xi,\rho)\in \so(4,\rat_p)$, and coincides with the isomorphism of Theorem~\ref{theo:SO4quat}.
\end{proof}

\subsection{The Haar integral on \texorpdfstring{$\so(3,\mathbb{Q}_p)$}{Lg}}\label{subsec:haarso3p}
The construction of the Haar integral on $\so(3,\rat_p)$ can be conveniently carried out by exploiting the conclusions of Theorem~\ref{th.2.17} and Proposition~\ref{prop:omeoSO3}. In particular, this will bring us along two main steps: First, we shall construct the Haar measure on $\pqal^\times$ and, hence, its associated Haar integral. Then, owing to the result in Theorem~\ref{th.2.17}, we will show that there is a natural lift of the Haar integral on $\so(3,\rat_p)$ to that of $\pqal^\times$.

To begin with, let us notice that, since $\pqal^\times$ is locally compact, it admits a left Haar measure.
\begin{proposition}\label{unih}
The group $\pqal^\times$ of invertible quaternions is unimodular.
\end{proposition}
\begin{proof}
We exploit the well known result that a locally compact group is unimodular whenever there exists a compact neighborhood of the identity element which is invariant under the inner automorphisms of the group (see Chapter V in~\cite{gaal}). In the present case, $1\in \pqal^\times$ is an element of $\rat_p^\times\leq \pqal^\times$. Since $\rat_p^\times$ is the centre of $\mathbb{H}_p^\times$, every subset of $\rat_p^\times$ is invariant under the inner automorphisms of $\pqal^\times$. $\mathbb{Z}_p^\times$ provides an instance of such an invariant compact neighborhood of the identity $1$.
\end{proof}

As a consequence of Proposition~\ref{unih}, the left and the right Haar measures on $\pqal^\times$ coincide, and we can construct it by directly exploiting formula~\eqref{eq.71}. In particular, since $\pqal^\times\cong\rat_p^4-\{0\}$ as topological spaces, the map $\varphi$ defining this homeomorphism provides us a \emph{global} coordinate map for the elements of $\pqal^\times$. Specifically, if $\xi=q_0+\I q_1+\J q_2+\K q_3\in \pqal^\times$, its coordinates are given by $\varphi(\xi)\coloneqq (q_0,q_1,q_2,q_3)$. 
Therefore, the density function $\eta$ (cf.\ Section~\ref{cosHaarmes}) characterizing the Haar measure on $\pqal^\times$ will be globally defined on the whole $\pqal^\times$. We will discuss the cases $p>2$ and $p=2$ separately. 

Let us assume $p>2$ first.
If $\xi=q_0+\I q_1+\J q_2+\K q_3$, and $\chi=x_0+\I x_1+\J x_2+\K x_3$ are two quaternions in $\mathbb{H}_p^\times$, their composition will result in another quaternion $\zeta=z_0+\I z_1+\J z_2+\K z_3$; namely
\begin{align}\label{eq.76}
\zeta=z_0+\I z_1+\J z_2+ \K z_3&=(q_0+\I q_1+\J q_2+\K q_3)(x_0+\I x_1+\J x_2+\K x_3)\nonumber\\
&=q_0x_0+\I q_0x_1+\J q_0x_2+\K q_0x_3+\I q_1x_0+vq_1x_1\nonumber\\
&-\K q_1x_2-\J vq_1x_3+\J q_2x_0+\K q_2x_1-pq_2x_2-\I pq_2x_3\nonumber\\
&+\K q_3x_0+\J vq_3x_1+\I pq_3x_2+ pvq_3x_3,
\end{align}
from which we argue that
\begin{align}
&z_0=q_0x_0+vq_1x_1-pq_2x_2+pvq_3x_3, &&z_1=q_0x_1+q_1x_0-pq_2x_3+pq_3x_2,\nonumber\\
&z_2=q_0x_2+q_2x_0-vq_1x_3+vq_3x_1,   &&z_3=q_0x_3+q_3x_0-q_1x_2+q_2x_1. \label{eq:ugualexentrambi}
\end{align}
We can now compute the function $\frac{\partial \zeta_{j}}{\partial x_{k}}\big(\ha\vf(q);\varphi(e)\big)$, where the vectors of coordinates of $e$ and $\xi$ are $(1,0,0,0)$ and $(q_0,q_1,q_2,q_3)\coloneqq q$ respectively:
\begin{equation}\label{eq:JacperHp}
\frac{\partial \zeta_{j}}{\partial x_{k}}\big(\ha\vf(q);\varphi(e)\big)=
\left.
\begin{pmatrix}
\frac{\partial \zeta_0}{\partial x_0} & \frac{\partial \zeta_0}{\partial x_1} & \frac{\partial \zeta_0}{\partial x_2} &  \frac{\partial \zeta_0}{\partial x_3}\\
\frac{\partial \zeta_1}{\partial x_0} & \frac{\partial \zeta_1}{\partial x_1} & \frac{\partial \zeta_1}{\partial x_2} &  \frac{\partial \zeta_1}{\partial x_3}\\
\frac{\partial \zeta_2}{\partial x_0} & \frac{\partial \zeta_2}{\partial x_1} &  \frac{\partial \zeta_2}{\partial x_2} &  \frac{\partial \zeta_2}{\partial x_3}\\
\frac{\partial \zeta_3}{\partial x_0} & \frac{\partial \zeta_3}{\partial x_1} &  \frac{\partial \zeta_3}{\partial x_2} &  \frac{\partial \zeta_3}{\partial x_3}
\end{pmatrix}\right\rvert_{\substack{\hspace{-7.2mm}x_0=1,\\ x_1,x_2,x_3=0}}=
\begin{pmatrix}
q_0&vq_1&-pq_2& pvq_3\\q_1& q_0 & pq_3 & -pq_2\\
q_2&vq_3 & q_0 & -vq_1\\
q_3 & q_2 & -q_1 & q_0
\end{pmatrix}.
\end{equation}
This yields
\begin{equation}\label{eq.83}
\det\left(\frac{\partial \zeta_{j}}{\partial x_{k}}\big(\ha\vf(q);\varphi(e)\big)\right)=(q_0^2-vq_1^2+pq_2^2-pvq_3^2)^2,
\end{equation}
which, as anticipated, is globally defined on $\pqal^\times$. Then, using~\eqref{eq.71}, we can conclude that the Haar measure of any Borel subset $\mathcal{E}$ of $\mathbb{H}_p^\times$ is
\begin{equation}
\mu_{\mathbb{H}_p^\times}(\mathcal{E})=\int_{\varphi(\mathcal{E})}\frac{1}{|q_0^2-vq_1^2+pq_2^2-pvq_3^2|_p^2}\de\lambda(q),
\end{equation} 
where we recall that $\de\lambda(q)$ denotes the Haar measure on $\rat_p^4$ (cf.\ Example~\ref{exa.2.3}).

For the $p=2$ case, a similar discussion to the one leading from~\eqref{eq.76} to~\eqref{eq.83} shows that 
\begin{equation}
\det\left(\frac{\partial \zeta_{i}}{\partial x_{j}}\big(\ha\vf(q);\varphi(e)\big)\right)=\det\begin{pmatrix}
q_0&-q_1&-q_2&-q_3\\q_1& q_0 & -q_3 & q_2\\
q_2&q_3 & q_0 & -q_1\\
q_3 & -q_2 & q_1 & q_0
\end{pmatrix}=(q_0^2+q_1^2+q_2^2+q_3^2)^2.
\end{equation}
Then, using~\eqref{eq.71}, the Haar measure of any Borel subset $\mathcal{E}$ of $\mathbb{H}_2^\times$ is
\begin{equation}
\mu_{\mathbb{H}_2^\times}(\mathcal{E})=\int_{\varphi(\mathcal{E})}\frac{1}{|q_0^2+q_1^2+q_2^2+q_3^2|_2^2}\de\lambda(q).
\end{equation} 
We summarize the above results with the following
\begin{proposition}
Let $p\geq 2$ be a prime number, and let $\pqal^\times$ be the group of invertible $p$-adic quaternions. The Haar measure $\mu_{\pqal^\times}$ on $\pqal^\times$ is given by
\begin{equation}\label{mesonhp}
\mu_{\pqal^\times}(\mathcal{E})=\int_{\varphi(\mathcal{E})}\frac{1}{|Q_{(4)}(q_0,q_1,q_2,q_3)|_p^2}\de\lambda(q),
\end{equation}
for every Borel set $\mathcal{E}\subset \pqal^\times$, where $\xi=q_0+\I q_1+\J q_2+\K q_3$, $\varphi(\xi)=(q_0,q_1,q_2,q_3)$, $\de\lambda(q)=\de q_0\de q_1\de q_2\de q_3$ is the Haar measure on $\mathbb{Q}_p^4$, and $Q_{(4)}$ is the definite quadratic form of $\mathbb{Q}_p^4$ (see~\eqref{eq:quadrform4} in Theorem~\ref{quadform}).
\end{proposition}  

Exploiting the results of Theorem~\ref{th.2.17}, we shall now prove that there exists a one-one correspondence between the Haar integrals on $\so(3,\rat_p)$ and those of 
$\pqal^\times$. Indeed, let us consider the quotient group $\pqal^\times/\rat_p^\times$. According to the results in Subsection~\ref{sec2n}, denoting by $\lambda$ the Haar measure on $\rat_p$ (see Example~\ref{exa.2.2}), and with $\ts\colon\pqal^\times/\rat_p^\times\rightarrow\pqal^\times$ a Borel cross section, the map $\widehat{P}\colon \mathrm{L}^1(\pqal^\times)\rightarrow\mathrm{L}^1(\pqal^\times/\rat_p^\times)$ defined as
\begin{equation}
 (\widehat{P}f)(x)\coloneqq \int_{\rat_p} \de\lambda(\alpha)f(\ts(x)\alpha), \qquad x\in \pqal^\times/\rat_p^\times,\,\,f\in\mathrm{L}^1(\pqal^\times)
\end{equation}
is a well-defined \emph{surjection} of $\mathrm{L}^1(\pqal^\times)$ onto $\mathrm{L}^1(\pqal^\times/\rat_p^\times)$ (cf.\ Remark~\ref{rem.2.14s}).  
For $K$ a compact subset of $\pqal^\times$, we can then define the set (cf.\ \eqref{eq.15s2}):
\begin{equation}
\Psi_K\coloneqq\{\psi\in \mathrm{C}_c^+(\pqal^\times)\mid (P\psi)(q)=1,\,\,\forall q\in K\}.
\end{equation}
In particular, adhering to the notation used in Subsection~\ref{sec2n}, we set $\Psi\equiv\Psi_{\pqal^\times/\rat_p^\times}$. Moreover, for every $\psi\in\Psi$, we also denote by $\widehat{\mathscr{L}}_\psi\colon\mathrm{L}^1(\pqal^\times/\rat_p^\times)\rightarrow \mathrm{L}^1(\pqal^\times)$ --- i.e., the extended WMB lift --- the right inverse of $\widehat{P}$, as defined in~\eqref{eq.27s}. We are now ready to prove the following
\begin{theorem}\label{theoso3}
Let $\mu_3$ and $\mu_{\pqal^\times}$ be the Haar measures on $\so(3,\rat_p)$ and $\pqal^\times$ respectively. For every prime $p\geq 2$, and any $\phi\in \mathrm{L}^1\big(\so(3,\rat_p)\big)$, the following equality holds:
\begin{equation}\label{intonso}
\int_{\so(3,\rat_p)}\de\mu_3(R)\phi(R)=\int_{\pqal^\times}\de\mu_{\pqal^\times}(q)\big(\widehat{\mathscr{L}}_\psi\phi\big)(q),
\end{equation}
where $\widehat{\mathscr{L}}_\psi\phi\in \mathrm{L}^1(\pqal^\times)$ is the (extended) WMB lift of the map $\phi$.   
\end{theorem}
\begin{proof}
From Proposition~\ref{prop:omeoSO3}, $\so(3,\rat_p)$ is homeomorphic to $\pqal^\times/\rat_p^\times$, hence this two spaces are Borel isomorphic. In particular, by exploiting the homeomorphism between $\so(3,\rat_p)$ and $\pqal^\times/\rat_p^\times$, it is clear that, for any $\phi\in \mathrm{L}^1\big(\so(3,\rat_p)\big)$, the Haar integral in the l.h.s.\ of~\eqref{intonso}
can be expressed in terms of a Haar integral of the function $\phi$ on the homogeneous space $\pqal^\times/\rat_p^\times$. On the other hand, the same homeomorphism also entails that $\pqal^\times/\rat_p^\times$ is a compact group. But then, the equality in~\eqref{intonso} directly follows from~\eqref{eq.20s2} in Theorem~\ref{th.2.17}.
\end{proof}

\begin{remark}
We stress that the result of Theorem~\ref{theoso3} provides an \emph{equivalence of Haar integrals}. In other terms, computing the Haar integral of a function $\phi$ in $\mathrm{L}^1\big(\so(3,\rat_p)\big)$ is equivalent to performing the integration of the function $\widehat{\mathscr{L}}_\psi\phi$, namely
\begin{equation}
\int\limits_{\so(3,\rat_p)} \de\mu_3(R)\phi(R)=\int\limits_{\pqal^\times}\frac{(\widehat{\mathscr{L}}_\psi\phi)(q_0,q_1,q_2,q_3)}{|Q_{(4)}(q_0,q_1,q_2,q_3)|_p^2}\de q_0\de q_1\de q_2\de q_3,
\end{equation}
where we have used the explicit form of the Haar measure~\eqref{mesonhp}.
\end{remark}

\subsection{The Haar integral on \texorpdfstring{$\so(4,\mathbb{Q}_p)$}{Lg}}\label{subsec:haarso4p}
Here we extend to $\so(4,\rat_p)$ the results of the last subsection. In particular, in complete analogy to what was done for $\so(3,\rat_p)$, we will provide a suitable lift of the Haar integral on $\so(4,\rat_p)$ to that on $\mathbb{P}(\pqal^\times)$.

We start by observing that the group $\mathbb{P}(\pqal^\times)$
is locally compact and, hence, it admits a left Haar measure. Moreover, $\mathbb{P}(\mathbb{H}_p^\times)$ is unimodular, since it is a subgroup of the unimodular group $\mathbb{H}_p^\times\times \mathbb{H}_p^\times$ (being the direct product of the unimodular groups $\pqal^\times$)~\cite{gaal}. Since the measure on every chart covering $\mathbb{P}(\pqal^\times)$ can be obtained by translating the measure around
its identity element $e$, it is enough to explicitly provide the latter by exploiting formula~\eqref{eq.71}.

Consider the pairs of quaternions $(\alpha,\beta),\,(\gamma,\delta)\in \mathbb{P}(\mathbb{H}_p^\times)$. From the defining condition of the group $\mathbb{P} (\mathbb{H}_p^\times)$, it must be true that $\rn(\alpha)=\rn(\beta)$ and $\rn(\gamma)=\rn(\delta)$. 
Explicitly, let $\alpha,\beta,\gamma,\delta$ be given by
\begin{align}
    &\alpha=a_0+\I a_1+\J a_2+\K a_3, && \beta=b_0+\I b_1+\J b_2+\K b_3,\nonumber\\ 
    & \gamma=c_0+\I c_1+\J c_2+\K c_3, && \delta=d_0+\I d_1+\J d_2+\K d_3.\label{eq:gamadeltaPHp}
\end{align}
We shall denote the composition of the two elements $(\alpha,\beta),(\gamma,\delta)$ in $\mathbb{P}(\mathbb{H}_p^\times)$ by $\zeta\coloneqq (\zeta_1,\zeta_2)=(\alpha\gamma,\beta\delta)$, where $\zeta_1=z_0+\I z_1+\J z_2+\K z_3$ and $\zeta_2=z_0'+\I z_1'+\J z_2'+\K z_3'$. Clearly, $\zeta$ is a function of the parameters $a_i,b_i,c_i,d_i$ for $i=0,1,2,3$. Now, to construct the Haar measure on $\mathbb{P}(\pqal^\times)$, we have first to compute the Jacobian of the function $\zeta$. In particular, we shall consider $(\alpha,\beta)$ as fixed, and $(\gamma,\delta)$ as variables. In what follows, we will treat separately the cases $p>2$ and $p=2$. 

Assume $p>2$ at first. The components $z_i$ and $z_i'$, $i=0,1,2,3$, of $\zeta_1$ and $\zeta_2$, can be computed in the same way in which we found~\eqref{eq:ugualexentrambi}:
\begin{align}
&z_0=a_0c_0+va_1c_1-pa_2c_2+pva_3c_3, && z_0'=b_0d_0+vb_1d_1-pb_2d_2+pvb_3d_3, \nonumber\\
&z_1=a_0c_1+a_1c_0-pa_2c_3+pa_3c_2, && z_1'=b_0d_1+b_1d_0-pb_2d_3+pb_3d_2, \nonumber\\
&z_2=a_0c_2+a_2c_0-va_1c_3+va_3c_1, && z_2'=b_0d_2+b_2d_0-vb_1d_3+vb_3d_1, \nonumber\\
& z_3=a_0c_3+a_3c_0-a_1c_2+a_2c_1,  &&  z_3'=b_0d_3+b_3d_0-b_1d_2+b_2d_1.\label{eq:compcompZ}
\end{align}
As anticipated, $z_i$ and $z_i^\prime$ are functions of the parameters $c_i,d_i$, $i=0,1,2,3$ (we are assuming $a_i,b_i$ to be fixed). The defining condition $\rn(\gamma)=\rn(\delta)$ of the group $\mathbb{P}(\pqal^\times)$ is equivalent to $c_0^2=Q_{(4)}(d_0,d_1,d_2,d_3)+vc_1^2-pc_2^2+pvc_3^2$, and imposes a constraint on the variables $c_i,d_i$. Actually, this condition allows us to only consider $(c_1,c_2,c_3,d_0,d_1,d_2,d_3)$ as \emph{independent} parameters in a neighborhood of $e\in\mathbb{P}(\pqal^\times)$: As the forthcoming remark will clarify, in such a neighborhood $Q_{(4)}(d_0,d_1,d_2,d_3)+vc_1^2-pc_2^2+pvc_3^2$ is a quadratic element, and its square root will then provide $c_0$ up to a sign. 

\begin{remark}\label{rem:Henselrootp}The identity element of $\mathbb{P}(\mathbb{H}_p^\times)$ is identified in $\mathbb{Q}_p^8$ by  $(c_0,c_1,c_2,c_3,d_0,d_1,d_2,d_3) =(1,0,0,0,1,0,0,0)$. Now, consider an open ball in $\mathbb{Q}_p^8$ of centre $(1,0,0,0,1,0,0,0)$ and radius $1$, in the usual $p$-adic topology. Here, $d_0=1+py_0,\,d_i=py_i,\, c_i=px_i$ with $y_0,y_i,x_i\in\mathbb{Z}_p$, $i=1,2,3$. Then, $Q_{(4)}(d_0,d_1,d_2,d_3)=1+pt$ and $Q_{(4)}(d_0,d_1,d_2,d_3)+vc_1^2-pc_2^2+pvc_3^2 =1+pt'$, where $t\coloneqq 2y_0+pQ_{(4)}(y_0,y_1,y_2,y_3),\, t'\coloneqq t+p(vx_1^2-px_2^2+pvx_3^2)\in\mathbb{Z}_p$. We can now resort to Hensel’s Lemma (Chapter II, Section 2.2 of~\cite{serre2012course}) to show that, actually, $1+pz$, $z\in\mathbb{Z}_p$, is a square in $\mathbb{Z}_p$, i.e. that $f(x)\coloneqq x^2-1-pz$ admits roots in $\mathbb{Z}_p$. First, $f(x)= x^2-1\mod p$ has roots $x= \pm1\mod p$, where the derivative $\frac{df(x)}{dx}=2x$ takes values $\pm2\neq0\mod p$. Then, Hensel's Lemma allows us to (uniquely) lift each of these simple roots to a root of the same function modulo $p^n$, $n\in\mathbb{Z}_{>1}$, converging to a $p$-adic root. This proves that $Q_{(4)}(d_0,d_1,d_2,d_3)$ and $Q_{(4)}(d_0,d_1,d_2,d_3)+vc_1^2-pc_2^2+pvc_3^2$ are squares. Hence, we can write
\begin{equation}\label{eq:c0dependentp}
c_0=\pm\sqrt{Q_{(4)}(d_0,d_1,d_2,d_3)+vc_1^2-pc_2^2+pvc_3^2},
\end{equation}
at least for $(c_0,c_1,c_2,c_3,d_0,d_1,d_2,d_3)$ in a ball in $\mathbb{Q}_p^8$ centred in $(1,0,0,0,1,0,0,0)$ of radius $1$.
\end{remark}
From Remark~\ref{rem:Henselrootp} above, we know that the domain of definition of the square root~\eqref{eq:c0dependentp} is non empty, and contains a neighborhood of the coordinates of the identity element of $\mathbb{P}(\mathbb{H}_p^\times)$, where $c_1,c_2,c_3,d_0,d_1,d_2,d_3$ can be assumed as independent variables. Here, we are referring to the coordinate map on such a neighborhood of $e\in \mathbb{P}(\mathbb{H}_p^\times)$ as 
\begin{equation}\label{eq:coordmapPHp}
    \mathbb{P}(\mathbb{H}_p^\times)\ni (\gamma,\delta)\mapsto \varphi_0((\gamma,\delta))\coloneqq (c_1,c_2,c_3,d_0,d_1,d_2,d_3)\in \mathbb{Q}_p^7,
\end{equation} where $\gamma$ and $\delta$ are as in~\eqref{eq:gamadeltaPHp}. 
The same argument can be repeated for the condition $\rn(\zeta_1)=\rn(\zeta_2)$, to express $z_0$ in terms of the other independent coordinates $z_0',z_i,z_i',\,i=1,2,3$. In conclusion, we are left with the following $7\times 7$ Jacobian matrix
\begin{equation}\label{eq:JacPHp}
\left.\frac{\partial \zeta_{0,j}\big(\ha\vf_0(a_i,b_i);\varphi_0(e)\big)}{\partial x_{k}}\right\rvert_{1\leq j,k\leq7}\hspace{-0.3cm}=
\left.
\begin{pmatrix}
\frac{\partial z_1}{\partial c_1} & \frac{\partial z_1}{\partial c_2} &  \frac{\partial z_1}{\partial c_3} & \frac{\partial z_1}{\partial d_0} & \frac{\partial z_1}{\partial d_1} & \frac{\partial z_1}{\partial d_2} &  \frac{\partial z_1}{\partial d_3} \\
\frac{\partial z_2}{\partial c_1} & \frac{\partial z_2}{\partial c_2} &  \frac{\partial z_2}{\partial c_3} & \frac{\partial z_2}{\partial d_0} & \frac{\partial z_2}{\partial d_1} & \frac{\partial z_2}{\partial d_2} &  \frac{\partial z_2}{\partial d_3} \\
\frac{\partial z_3}{\partial c_1} & \frac{\partial z_3}{\partial c_2} &  \frac{\partial z_3}{\partial c_3} & \frac{\partial z_3}{\partial d_0} & \frac{\partial z_3}{\partial d_1} & \frac{\partial z_3}{\partial d_2} &  \frac{\partial z_3}{\partial d_3} \\
\frac{\partial z_0'}{\partial c_1} & \frac{\partial z_0'}{\partial c_2} &  \frac{\partial z_0'}{\partial c_3} & \frac{\partial z_0'}{\partial d_0} & \frac{\partial z_0'}{\partial d_1} & \frac{\partial z_0'}{\partial d_2} &  \frac{\partial z_0'}{\partial d_3} \\
\frac{\partial z_1'}{\partial c_1} & \frac{\partial z_1'}{\partial c_2} &  \frac{\partial z_1'}{\partial c_3} & \frac{\partial z_1'}{\partial d_0} & \frac{\partial z_1'}{\partial d_1} & \frac{\partial z_1'}{\partial d_2} &  \frac{\partial z_1'}{\partial d_3} \\
\frac{\partial z_2'}{\partial c_1} & \frac{\partial z_2'}{\partial c_2} &  \frac{\partial z_2'}{\partial c_3} & \frac{\partial z_2'}{\partial d_0} & \frac{\partial z_2'}{\partial d_1} & \frac{\partial z_2'}{\partial d_2} &  \frac{\partial z_2'}{\partial d_3} \\
\frac{\partial z_3'}{\partial c_1} & \frac{\partial z_3'}{\partial c_2} &  \frac{\partial z_3'}{\partial c_3} & \frac{\partial z_3'}{\partial d_0} & \frac{\partial z_3'}{\partial d_1} & \frac{\partial z_3'}{\partial d_2} &  \frac{\partial z_3'}{\partial d_3} \\
\end{pmatrix}\right\rvert_{\substack{c_1=c_2=c_3=0 \\ d_0=1,d_1=d_2=d_3=0}},
\end{equation}
where, in the l.h.s., $\vf_0(e)=(0,0,0,1,0,0,0)$ and $(x_k)_{k=1}^7=(c_1,c_2,c_3,d_0,d_1,d_2,d_3)$.

By using~\eqref{eq:c0dependentp}, the partial derivatives of the dependent variable $c_0$ w.r.t.\ the independent ones are
\begin{equation}\label{eq:partiaSO4p0}
\left.\frac{\partial c_0}{\partial c_i}\right\rvert_{\substack{c_1=c_2=c_3=0 \\ d_0=1,d_1=d_2=d_3=0}}=\left.\frac{\partial c_0}{\partial d_i}\right\rvert_{\substack{c_1=c_2=c_3=0 \\ d_0=1,d_1=d_2=d_3=0}}=0,\quad \textup{for}\ i=1,2,3,
\end{equation}
and
\begin{equation}\label{eq:partiaSO4p1}
\left.\frac{\partial c_0}{\partial d_0}\right\rvert_{\substack{c_1=c_2=c_3=0 \\ d_0=1,d_1=d_2=d_3=0}}=\pm1.
\end{equation}
Hence, using the expressions in~\eqref{eq:compcompZ}, the Jacobian matrix~\eqref{eq:JacPHp} is straightforwardly computed and reads: 
\begin{equation}
\left.\frac{\partial \zeta_{0,j}\big(\ha\vf_0(a_i,b_i);\vf_0(e)\big)}{\partial x_{k}}\right\rvert_{1\leq j,k\leq7}= \begin{pmatrix}
a_0&pa_3 & -pa_2 & \pm a_1 & 0 & 0 & 0\\
va_3 & a_0 & -va_1 & \pm a_2 & 0 & 0 & 0\\
a_2 & -a_1 & a_0 & \pm a_3 & 0 & 0 & 0\\
0 & 0 & 0 & b_0 & vb_1 & -pb_2 & pvb_3 \\
0 & 0 & 0 & b_1 & b_0 & pb_3 & -pb_2\\
0 & 0 & 0 & b_2 & vb_3 & b_0 & -vb_1\\
0 & 0 & 0 & b_3 & b_2 & -b_1 & b_0
\end{pmatrix}.
\end{equation}
The $p$-adic absolute value of the determinant of such a matrix is
\begin{align}
\left\lvert \det\frac{\partial \zeta_{0,j}\big(\ha\vf_0(a_i,b_i);\vf_0(e)\big)}{\partial x_{k}}\right\rvert_p &= \left\lvert a_0(a_0^2-va_1^2+pa_2^2-pva_3^2)(b_0^2-vb_1^2+pb_2^2-pvb_3^2)^2\right\rvert_p\nonumber \\
& = \left\lvert
\sqrt{Q_{(4)}(b_0,b_1,b_2,b_3)+va_1^2-pa_2^2+pva_3^2\,}\,Q_{(4)}(b_0,b_1,b_2,b_3)^3\right\rvert_p.\label{eq:haarSO4p}
\end{align}
For the last equality, we used again the condition $\rn(\xi)=\rn(\rho)$ in a suitable neighborhood for the coordinates of the identity of $\mathbb{P}(\mathbb{H}_p^\times)$, where $a_0=\pm \sqrt{Q_{(4)}(b_0,b_1,b_2,b_3)+va_1^2-pa_2^2+pva_3^2}$ is well-defined.

Let us now switch to the case where $p=2$. The components of $\zeta_1$ and $\zeta_2$ are now given by: 
\begin{align}
&z_0=a_0c_0-a_1c_1-a_2c_2-a_3c_3, && z_0'=b_0d_0-b_1d_1-b_2d_2-b_3d_3 \nonumber\\
&z_1=a_0c_1+a_1c_0+a_2c_3-a_3c_2, && z_1'=b_0d_1+b_1d_0+b_2d_3-b_3d_2 \nonumber\\
&z_2=a_0c_2+a_2c_0-a_1c_3+a_3c_1, && z_2'=b_0d_2+b_2d_0-b_1d_3+b_3d_1 \nonumber\\
& z_3=a_0c_3+a_3c_0+a_1c_2-a_2c_1,  &&  z_3'=b_0d_3+b_3d_0+b_1d_2-b_2d_1.
\end{align}
The defining condition $\rn(\nu)=\rn(\varrho)$ is equivalent to $c_0^2=Q_{(4)}(d_0,d_1,d_2,d_3)-c_1^2-c_2^2-c_3^2$. Then, analogously to the case where $p>2$, it is not difficult to prove that, in a suitable neighborhood of $(1,0,0,0,1,0,0,0)$ in $\mathbb{Q}_p^8$, $c_0$ can be expressed in terms of the independent variables $c_1,c_2,c_3,d_0,d_1,d_2,d_3$, as the forthcoming Remark will clarify.
\begin{remark}\label{rem:Henselroot2}
Consider an open ball in $\mathbb{Q}_2^8$ of centre $(1,0,0,0,1,0,0,0)$ and radius $1/2$, in which $d_0=1+4y_0,\,d_i=4y_i,\,c_i=4x_i$, with $y_0,y_i,x_i\in\mathbb{Z}_2$, $i=1,2,3$. In this case, $Q_{(4)}(d_0,d_1,d_2,d_3)=(1+4y_0)^2+(4y_1)^2+(4y_2)^2+(4y_3)^2=1+8\big(y_0+2Q_{(4)}(y_0,y_1,y_2,y_3)\big)$ and $Q_{(4)}(d_0,d_1,d_2,d_3)-c_1^2-c_2^2-c_3^2 = 1+8\big[y_0+2\big(Q_{(4)}(y_0,y_1,y_2,y_3)-x_1^2-x_2^2-x_3^2\big)\big]$ are squares in $\mathbb{Z}_2$, as they are congruent to $1$ modulo $8$. Therefore, we can write
\begin{equation}\label{eq:c0dependent2}
c_0=\pm\sqrt{Q_{(4)}(d_0,d_1,d_2,d_3)-c_1^2-c_2^2-c_3^2},
\end{equation}
at least for $(c_0,c_1,c_2,c_3,d_0,d_1,d_2,d_3)$ in an open ball in $\mathbb{Q}_2^8$ centred in $(1,0,0,0,1,0,0,0)$, and of radius $1/2$.
\end{remark}
As a consequence, the coordinate map on a suitable neighborhood of $e\in \mathbb{P}(\mathbb{H}_2^\times)$ to $\rat_2^7$ is as in~\eqref{eq:coordmapPHp}. An analogous discussion can be carried out for the condition $\rn(\zeta_1)=\rn(\zeta_2)$, to show that $z_0$ can be expressed as a function of the (independent) variables  $z_0',z_i,z_i',\,i=1,2,3$, in a suitable neighborhood of the identity. It follows that the Jacobian matrix for $p=2$ is of the same form~\eqref{eq:JacPHp} for $p>2$ and,  as one can easily check, the partial derivatives of the dependent variable $c_0$ are again given by~\eqref{eq:partiaSO4p0} and~\eqref{eq:partiaSO4p1}. Thus, we obtain
\begin{equation}
\left.\frac{\partial \zeta_{0,j}\big(\ha\vf_0(a_i,b_i);\vf_0(e)\big)}{\partial x_{k}}\right\rvert_{1\leq j,k\leq7}= \begin{pmatrix}
a_0&-a_3 & a_2 & \pm a_1 & 0 & 0 & 0\\
a_3 & a_0 & -a_1 & \pm a_2 & 0 & 0 & 0\\
-a_2 & a_1 & a_0 & \pm a_3 & 0 & 0 & 0\\
0 & 0 & 0 & b_0 & -b_1 & -b_2 & -b_3 \\
0 & 0 & 0 & b_1 & b_0 & -b_3 & b_2\\
0 & 0 & 0 & b_2 & b_3 & b_0 & -b_1\\
0 & 0 & 0 & b_3 & -b_2 & b_1 & b_0
\end{pmatrix},
\end{equation}
which then yields
\begin{align}
\left\lvert \det\frac{\partial \zeta_{0,j}\big(\ha\vf_0(a_i,b_i);\vf_0(e)\big)}{\partial x_{k}}\right\rvert_2 &= \left\lvert a_0(a_0^2+a_1^2+a_2^2+a_3^2)(b_0^2+b_1^2+b_2^2+b_3^2)^2\right\rvert_2\nonumber \\
& = \left\lvert \sqrt{Q_{(4)}(b_0,b_1,b_2,b_3)-a_1^2-a_2^2-a_3^2\,}\,Q_{(4)}(b_0,b_1,b_2,b_3)^3\right\rvert_2.\label{eq:haarSO42}
\end{align}
Once more, this expression is valid in the domain of definition of the $p$-adic square root, containing a neighborhood for the coordinates of the identity of $\mathbb{P}(\mathbb{H}_2^\times)$.

By exploiting~\eqref{eq.71}, the $p$-adic absolute values~\eqref{eq:haarSO4p} and~\eqref{eq:haarSO42} immediately yield the Haar measure on a neighborhood of the identity element of $\mathbb{P}(\mathbb{H}_p^\times)$.
\begin{proposition}\label{propHaarPhPx}
Let $p\geq2$ be a prime number, and let $\mathbb{P}(\mathbb{H}_p^\times)$ be the group of $p$-adic quaternion pairs~\eqref{eq:subgrpquatpairs}. For any Borel subset $\mathcal{E}$ of $\mathbb{P}(\mathbb{H}_p^\times)$, the following equalities hold:
\begin{equation}\label{eq.123}
\mu_{\mathbb{P}(\mathbb{H}_p^\times)}(\mathcal{E}\cap\U_0)=\int_{\varphi_0(\mathcal{E}\cap\U_0)}\left\lvert \sqrt{Q_{(4)}(b_0,b_1,b_2,b_3)+va_1^2-pa_2^2+pva_3^2\,}\,Q_{(4)}(b_0,b_1,b_2,b_3)^3\right\rvert_p^{-1} \de\lambda(q),
\end{equation}
for $p>2$, while for $p=2$,
\begin{equation}\label{eq.124}
\mu_{\mathbb{P}(\mathbb{H}_2^\times)}(\mathcal{E}\cap\U_0)=\int_{\varphi_0(\mathcal{E}\cap\U_0)}\left\lvert \sqrt{Q_{(4)}(b_0,b_1,b_2,b_3)-
a_1^2-a_2^2-a_3^2\,}\,Q_{(4)}(b_0,b_1,b_2,b_3)^3\right\rvert_2^{-1}\de\lambda(q).
\end{equation}
Here, $\mu_{\mathbb{P}(\mathbb{H}_p^\times)}$ and $\de\lambda(q)=da_1da_2da_3db_0db_1db_2db_3$ denote the Haar measure on $\mathbb{P}(\mathbb{H}_p^\times)$ and on $\mathbb{Q}_p^7$ respectively, $
Q_{(4)}$ denotes the definite quadratic form of $\rat_p^4$, and $\U_0$ is a suitable neighborhood of the identity element $e\in \mathbb{P}(\mathbb{H}_p^\times)$ where the coordinate map $\varphi_0$ (cf.\ \eqref{eq:coordmapPHp}) is defined. 
\end{proposition} 
Since, by translation invariance, one can `move' the measure on the fixed chart around $e$ all over the group, Proposition~\ref{propHaarPhPx} is enough to compute any Haar integral on the whole $\mathbb{P}(\mathbb{H}_p^\times)$.

\medskip

At this point, we are now ready to construct the Haar integral on $\so(4,\rat_p)$. Indeed, as done for $\so(3,\rat_p)$, we can define a suitable (surjective) map $\widehat{P}\colon \mathrm{L}^1\big(\mathbb{P}(\pqal^\times)\big)\rightarrow\mathrm{L}^1\big(\mathbb{P}(\pqal^\times)/\rat_p^\times\big)$ such that
\begin{equation}
(\widehat{P}f)(x)\coloneqq\int_{\rat_p}\de\lambda(\alpha)f(\ts(x)\alpha),\quad x\in\mathbb{P}(\pqal^\times/\rat_p^\times),\,\,f\in\mathrm{L}^1\big(\mathbb{P}(\pqal^\times)\big),
\end{equation}
where $\ts\colon\mathbb{P}(\pqal^\times)/\rat_p^\times\rightarrow\mathbb{P}(\pqal^\times)$ is any Borel cross section of $\mathbb{P}(\pqal^\times)/\rat_p^\times$ onto $\mathbb{P}(\pqal^\times)$. Again, for any compact subset $K$ of $\mathbb{P}(\pqal^\times)$, we define
\begin{equation}
\Psi_K\coloneqq\{\psi\in \mathrm{C}_c^+\big(\mathbb{P}(\pqal^\times)\big)\mid (P\psi)(q)=1,\,\,\forall q\in K\},
\end{equation}
and set $\Psi\equiv\Psi_{\mathbb{P}(\pqal^\times)/\rat_p^\times}$. Then, for every $\psi\in\Psi$, the map $\widehat{\mathscr{L}}_\psi\colon\mathrm{L}^1\big(\mathbb{P}(\pqal^\times)/\rat_p^\times\big)\rightarrow \mathrm{L}^1\big(\mathbb{P}(\pqal^\times)\big)$ --- i.e., the extended WMB lift --- provides a right inverse of $\widehat{P}$. 

In the light of the above discussion, the following result is now clear:
\begin{theorem}\label{theoso4}
Let $\mu_4$ and $\mu_{\mathbb{P}(\mathbb{H}_p^\times)}$ be the Haar measures on $\so(4,\rat_p)$, and $\mathbb{P}(\mathbb{H}_p^\times)$ respectively. For every prime $p\geq 2$, and any $\phi\in \mathrm{L}^1\big(\so(4,\rat_p)\big)$, the following equality holds:
\begin{equation}\label{intonso4}
\int\limits_{\so(4,\rat_p)}\de\mu_4(R)\phi(R) =\int\limits_{\mathbb{P}(\pqal^\times)}\de\mu_{\mathbb{P}(\mathbb{H}_p^\times)}(q)(\widehat{\mathscr{L}}_\psi\phi)(q),
\end{equation}
where $\widehat{\mathscr{L}}_\psi\phi\in \mathrm{L}^1\big(\mathbb{P}(\pqal^\times)\big)$ is the (extended) WMB lift of the map $\phi$. 
\end{theorem}
\begin{proof}
By Proposition~\ref{prop:omeoSO4}, $\so(4,\rat_p)$ and $\mathbb{P}(\mathbb{H}_p^\times)/\rat_p^\times$ are homeomorphic and, then, Borel isomorphic. This then entails that, for any given function $\phi\in \mathrm{L}^1\big(\so(4,\rat_p)\big)$, we can express its Haar integral (w.r.t.\ the Haar measure $\mu$ on $\so(4,\rat_p)$) as an Haar integral on $\mathbb{P}(\pqal^\times)/\rat_p^\times$.
Moreover, the same homeomorphism also implies that $\mathbb{P}(\pqal^\times)/\rat_p^\times$ is a compact group. But then, the equality in~\eqref{intonso4} directly follows from~\eqref{eq.20s2} of Theorem~\ref{th.2.17}.
\end{proof}


\section{Conclusions}\label{sec6}
In this work, we provided a general expression of the Haar measure on a $p$-adic Lie group. 
Considering this measure as naturally induced by the invariant volume form on the group, we addressed the problem of determining the Haar measure on the $p$-adic special orthogonal groups in dimension two, three and four (for every prime number $p$). Let us briefly summarize our main achievements:
\begin{itemize}
\item The first part of this work is devoted to the discussion of the main tools and techniques used later on, in the application part of the paper. Specifically, in Subsection~\ref{sec2n}, we first recall the Weil-Mackey-Bruhat formula and the related notion of Weil-Mackey-Bruhat lift, which provide powerful tools allowing one to express the Haar integral on a quotient group $X=G/H$ (where $H$ is any closed normal subgroup of $G$) as a suitable lift of the Haar integral defined on the (LCH) group $G$. We next provide, in Subsection~\ref{subsec.2.1}, an overview of $p$-adic manifolds and $p$-adic Lie groups, suitable for our purposes, especially focusing on the relevant topological properties.  
\item In Section~\ref{cosHaarmes}, by exploiting the results discussed in Subsection~\ref{subsec.2.1} --- in particular, the total disconnectedness of $p$-adic manifolds ---  we provide a general method for constructing the invariant measure on a $p$-adic Lie group.
\item Next, in Section~\ref{sec5}, we work out various applications of the general techniques developed in the previous sections. In particular, a direct application of the general formula derived in Section~\ref{cosHaarmes} provides us with the Haar measure on $\mathrm{SO}(2,\mathbb{Q}_p)$.
\item As for the groups $\mathrm{SO}(3,\mathbb{Q}_p)$ and $\mathrm{SO}(4,\mathbb{Q}_p)$, instead, we argue that a more convenient approach consists in resorting to a quaternionic realization of these groups, eventually proving the useful group isomorphisms $\mathrm{SO}(3,\mathbb{Q}_p) \cong \mathbb{H}_p^\times/\mathbb{Q}_p^\times$ and $\mathrm{SO}(4,\mathbb{Q}_p) \cong \mathbb{P}(\mathbb{H}_p^\times)/\mathbb{Q}_p^\times$  (c.f.\ Theorems~\ref{theo:SO3quat} and~\ref{theo:SO4quat}).
\item It is precisely at this stage that the machinery previously developed in Subsection~\ref{sec2n} comes into play. Indeed, once the Haar measures on $\mathbb{H}_p^\times$ and $\mathbb{P}(\mathbb{H}_p^\times)$ (equivalently, regarding such measures as functionals, the associated Haar integrals) have been determined, we can construct the Haar integrals on the quotient groups $\mathbb{H}_p^\times/\mathbb{Q}_p^\times$ and $\mathbb{P}(\mathbb{H}_p^\times)/\mathbb{Q}_p^\times$ --- thus, on $\mathrm{SO}(3,\mathbb{Q}_p)$ and $\mathrm{SO}(4,\mathbb{Q}_p)$ respectively --- in terms of a suitable Weil-Mackey-Bruhat lift (see Theorems~\ref{theoso3} and~\ref{theoso4}). 
\end{itemize}

Several further developments can be foreseen from the present study. First, the possibility of deriving the explicit expression of the Haar measure on $\so(3,\rat_p)$ by exploiting its parameterization in terms of  Cardano --- a.k.a.\ nautical --- angles (see~\cite{our1st}). This is enlightened by the relation between $p$-adic rotations and the values modulo squares of the reduced norms of $p$-adic quaternions.  However, it can work for all primes but $p=2$, where no Cardano (or Euler) decomposition exists.

Second, since compact $p$-adic special orthogonal groups are profinite, another possible approach to their Haar measure is through a suitable notion of inverse limit of an inverse family of measure spaces on their projections modulo $p^k$, $k\in\mathbb{Z}^+_\ast$. 

Looking further ahead, the relevance of the subject treated here lies in the fact that the compactness of $\so(3,\rat_p)$ entails that all its irreducible unitary representations occur (and can be studied) as sub-representations of the regular one, according to the celebrated Peter-Weyl theorem. In turn, the Haar measure on $\so(3,\rat_p)$ plays a fundamental role in the study of its regular representation and, more generally, of its irreducible projective unitary representations. Those of dimension two can be regarded as a model of $p$-adic qubit (see~\cite{our2}) and would be the core of a quantum information processing based on $p$-adic numbers.

\section*{Acknowledgments}
IS is supported by the Istituto Nazionale di Fisica Nucleare (INFN), Sezione di Perugia, and by the Spanish MICIN (project 
PID2022-141283NB-I00) with the support of FEDER funds. AW is supported by the European Commission QuantERA grant ExTRaQT (Spanish MICIN project PCI2022-132965), by the Spanish MICIN (project PID2022-141283NB-I00) with the support of FEDER funds, by the Spanish MICIN with funding from European Union NextGenerationEU (PRTR-C17.I1) and the Generalitat de Catalunya, by the Spanish MTDFP through the QUANTUM ENIA project: Quantum Spain, funded by the European Union NextGenerationEU within the framework of the ``Digital Spain 2026 Agenda'', by the Alexander von Humboldt Foundation, and by the Institute for Advanced Study of the Technical University Munich.

\appendix
\section{The real quaternion algebra and its relations with \texorpdfstring{$\so(3,\mathbb{R})$}{Lg} and  \texorpdfstring{$\so(4,\mathbb{R})$}{Lg}}\label{app:A}

We devote this appendix to a brief account on the real quaternion algebra $\qal$, along with a discussion of the quaternionic realization of the elements in $\so(3,\mathbb{R})$ and $\so(4,\mathbb{R})$. This will also give us the opportunity to highlight analogies and differences with the $p$-adic case of Section~\ref{sec:quatalp}.

\subsection{The real quaternion algebra \texorpdfstring{$\qal$}{Lg}}\label{sec.1}
 There are several ways of describing the real quaternion algebra $\mathbb{H}$~\cite{folland2016course,jacobson2012basic}. As a real vector space $\mathbb{H}\cong \mathbb{R}\times\mathbb{R}^3$, and any element in $\mathbb{H}$ is written as $\xi=(a,\mathbf{x})$, with $a\in\mathbb{R}$ and $\mathbf{x}\in\mathbb{R}^3$. The multiplication law is given by:
\begin{equation}\label{eq.88}
(a,\mathbf{x})(b,\mathbf{y})=(ab-\mathbf{x}\cdot\mathbf{y}, b\mathbf{x}+a\mathbf{y}+\mathbf{x}\times\mathbf{y}),
\end{equation}
where $\mathbf{x}\cdot\mathbf{y}$ and $\mathbf{x}\times\mathbf{y}$ are  respectively the usual inner product and vector product  between vectors in $\mathbb{R}^3$.  It is easily verified that the product~\eqref{eq.88} is associative. The centre of $\qal$ is given by the subspace $\mathbb{R}\times \{\boldsymbol{0}\}\cong \mathbb{R}$. Likewise, we identify $\{0\}\times\mathbb{R}^3$ with $\mathbb{R}^3$, in such a way that every element in $\qal$ can be expressed as $\xi=a+\mathbf{x}$, $a\in\mathbb{R}$, $\mathbf{x}\in\mathbb{R}^3$. Denoting by $\mathbf{i},\mathbf{j},\mathbf{k}$ the canonical basis vectors of $\mathbb{R}^3$, $\xi$ can be expressed as
\begin{equation}
\xi=a+x_1\mathbf{i}+x_2\mathbf{j}+x_3\mathbf{k},
\end{equation}
where $a,x_1,x_2,x_3\in\mathbb{R}$. Then, the multiplication law between quaternions is given by specifying the products between the basis vectors $\I,\J,\K$~\cite{ folland2016course,smirnov2011linear}:
\begin{equation}\label{comrel}
\mathbf{i}^2=\mathbf{j}^2=\mathbf{k}^2=-1,\;\;\;\;\mathbf{i}\mathbf{j}=-\mathbf{j}\mathbf{i}=\mathbf{k},\;\;\;\; \mathbf{j}\mathbf{k}=-\mathbf{k}\mathbf{j}=\mathbf{i}\;\;\;\;
\mathbf{k}\mathbf{i}=-\mathbf{i}\mathbf{k}=\mathbf{j}.
\end{equation}
It is straightforward to realize that $\mathbb{H}$ is a non-abelian algebra.

$\mathbb{H}$ is an involutive algebra, as the map $\mathbb{H}\ni \xi=a+\mathbf{x}\mapsto \overline{\xi}\coloneqq a-\mathbf{x}\in\mathbb{H}$ is an involutive anti-automorphism.
Moreover, $\xi\overline{\xi}=|\xi|^2=a^{2}+|\mathbf{x}|^2 = Q_{\mathbb{R}}(a,x_1,x_2,x_3)$, where $Q_{\mathbb{R}}$ denotes the definite quadratic form of $\mathbb{R}^4$. Thus, every non-zero element in $\qal$ is invertible, with $\xi^{-1}=\overline{\xi}/|\xi|^2$, and so $\qal$ is a division algebra. Those elements in $\qal$ for which $|\xi|=1$ are called \emph{unit quaternions}. They form a group in $\qal$, denoted by $\un(\qal)$:
\begin{equation}\label{eq:uintquat}
\un(\qal)=\{\xi\in\qal\;|\; |\xi|=1\}=\{\xi\in\qal \;|\;  \xi^{-1}=\overline{\xi}\}.
\end{equation}

\begin{remark}\label{rem:reducednorm}
In the literature (e.g.,~\cite{voight2021}), the quantity $\xi\overline{\xi}$ is referred to as the \emph{reduced norm} of $\xi$ in $\mathbb{H}$, and denoted by $\rn(\xi)$ (see Remark~\ref{redborm}). From the definition, it is clear that the reduced norm on the real quaternion algebra $\mathbb{H}$ is equivalent to the square of the Euclidean norm on $\mathbb{R}^4$ (since the definite quadratic form on $\mathbb{R}^4$ induces the Euclidean inner product, and vice versa). However, this is not the case when one defines a quaternion algebra over a generic field w.r.t.\ some quadratic form (see, for instance~\cite{voight2021, Lam}): The latter does not necessarily induce the considered norm on the vector-space structure of that algebra.
\end{remark}

There is another (yet equivalent) way in which $\qal$ can be described~\cite{jacobson2012basic}. Let us consider the subset $\mqal$ of the algebra $\mathsf{M}_2(\mathbb{C})$ of complex $2\times 2$ matrices of the form
\begin{equation}\label{eq:subdivalg}
M=
\begin{pmatrix}
a & b\\
-\overline{b} & \overline{a}
\end{pmatrix}
=
\begin{pmatrix}
q_0+iq_1 & q_2+ iq_3\\
-q_2+iq_3 & q_0-iq_1
\end{pmatrix},
\quad q_j\in\mathbb{R},
\end{equation}
for every $j=0,1,2,3$, where $i$ denotes the imaginary unit. One can easily verify that $\mqal$ is a subalgebra (actually, a \emph{division} algebra) in $\mathsf{M}_2(\mathbb{C})$. In particular, that every non-null element $M\in\mqal$ is invertible easily follows by observing that
\begin{equation}
\det(M)=\det
\begin{pmatrix}
q_0+iq_1 & q_2+iq_3\\
-q_2+iq_3 & q_0-iq_1
\end{pmatrix}
= q_0^2+q_1^2+q_2^2+q_3^2
\end{equation}
is the (non-degenerate) four-dimensional definite quadratic form over $\mathbb{R}$. 
From~\eqref{eq:subdivalg},
we also see that every element $M\in\mqal$ can be written as 
$M=q_0+\I q_1+\J q_2+\K q_3$, where
\begin{equation}
\I\coloneqq
\begin{pmatrix}
i & 0\\
0 & -i
\end{pmatrix},
\quad\J\coloneqq
\begin{pmatrix}
0 & 1\\
-1 & 0
\end{pmatrix},
\quad\K\coloneqq
\begin{pmatrix}
0 & i\\
i & 0
\end{pmatrix},
\end{equation}
and where we omitted the identity matrix $\mathrm{I}_2$ multiplying $q_0$. Moreover, it can be easily checked that $\mathbf{i},\,\mathbf{j},\,\mathbf{k}$ obey  commutation relations which are analogous to the quaternion commutation relations~\eqref{comrel}. Indeed, the correspondence
\begin{equation}
\theta\coloneqq \qal\ni\xi=(q_0,q_1,q_2,q_3)\mapsto \theta(\xi):= \begin{pmatrix}
q_0+iq_1 & q_2+iq_3\\
-q_2+iq_3 & q_0-iq_1
\end{pmatrix}
\in\mqal
\end{equation}
defines an \emph{algebra isomorphism} from the quaternions $\qal$ to the algebra of complex matrices $\mqal$~\cite{folland2016course}.
In particular, unit quaternions are identified in $\mqal$ by
\begin{equation}\label{eq:uintquat2}
\un(\mqal)=\{M\in \mqal\mid \det(M)=1\}.
\end{equation}

\subsection{Relations between real quaternions and rotations}\label{sec.A2}
Here we recall the relation between $\mathbb{H}$ and $\so(3,\mathbb{R})$. Let $\xi\in \un(\qal)$ be a unit quaternion. The map $\qal\ni \eta\mapsto \xi\eta\xi^{-1}\in \qal$ is an \emph{isometric linear transformation} of $\qal$, which leaves the centre $\mathbb{R}$ of $\qal$ pointwise fixed and, therefore, also leaves the orthogonal subspace $\mathbb{R}^3$ invariant. Hence, the restriction of this map to $\mathbb{R}^3$ is an element of $O(3,\mathbb{R})$, that we denote by $\qrot(\xi)$:

\begin{equation}
\qrot(\xi)\mathbf{x}\coloneqq \xi \mathbf{x}\xi^{-1},\quad \mathbf{x}\in\mathbb{R}^3.
\end{equation}
Furthermore, $\qrot(\xi\nu)=\qrot(\xi)\qrot(\nu)$, i.e., $\qrot\colon \un(\qal)\rightarrow O(3,\mathbb{R})$ is a homomorphism. Let us derive the explicit form of $\qrot(\xi)$. For $\qal\ni\xi=q_0+\I q_1 +\J q_2 + \K q_3$, and $\mathbb{R}^3\ni\mathbf{x}=\I x+ \J y+ \K z$, we have:
\begin{align}
\xi \mathbf{x}\xi^{-1}=&\I\big(x(q_0^2+q_1^2-q_2^2-q_3^2)+2y(q_1q_2-q_3q_0)+2z(q_2q_0+q_3q_1)\big),\nonumber\\
&\J\big(2x(q_1q_2+q_0q_3)+y(q_0^2-q_1^2+q_2^2-q_3^2)+2z(q_2q_3-q_1q_0)\big),\nonumber\\
&\K\big(2x(q_1q_3-q_2q_0)+2y(q_1q_0+q_2q_3)+z(q_0^2-q_1^2-q_2^2+q_3^2)\big),
\end{align}
from which we deduce that $\qrot(\xi)$ is given by
\begin{equation}\label{eq:kappaxi}
\qrot(\xi)=
\begin{pmatrix}
q_0^2+q_1^2-q_2^2-q_3^2 & 2(q_1q_2-q_3q_0) & 2(q_2q_0+q_3q_1)\\
2(q_1q_2+q_0q_3) & q_0^2-q_1^2+q_2^2-q_3^2 & 2(q_2q_3-q_1q_0)\\
2(q_1q_3-q_2q_0) & 2(q_1q_0+q_2q_3) & q_0^2-q_1^2-q_2^2+q_3^2
\end{pmatrix}.
\end{equation}
A direct calculation shows that $\det(\qrot(\xi))=1$, i.e., $\qrot(\xi)\in \so(3,\mathbb{R})$.
\begin{remark}
The fact that $\qrot(\xi)\in \so(3,\mathbb{R})$ also follows by observing that $\un(\qal)$ is \emph{connected} and $\qrot\colon \un(\qal)\rightarrow \so(3,\mathbb{R})$ is continuous~\cite{folland2016course}.
\end{remark}
In the light of the discussion above, every unit quaternion $\xi=q_0+\I q_1 +\J q_2 + \K q_3\in \un(\qal)$ is associated with a rotation $R$ in $\so(3,\mathbb{R})$. In particular, $\kappa$ is the group homomorphism $\un(\mathbb{H})\rightarrow \so(3,\mathbb{R})$ in the short exact sequence~\eqref{realexac}~\cite{voight2021, Lam}. This then yields the group isomorphism 
\begin{equation}\label{eq:so3real2quat}
    \so(3,\mathbb{R})\cong \un(\mathbb{H})/\mathbb{F}_2.
\end{equation}
With a similar argument, based on the isometries $\mathbb{H}\ni\eta\mapsto \xi\eta\varrho^{-1}\in \mathbb{H}$ for $(\xi,\varrho)\in \un(\mathbb{H})\times \un(\mathbb{H})$, one can prove the following isomorphism~\cite{voight2021}:
\begin{equation}\label{eq:so4real2quat}
\so(4,\mathbb{R})\cong (\un(\qal)\times \un(\qal))/\mathbb{F}_2.
\end{equation}
\eqref{eq:so3real2quat} and~\eqref{eq:so4real2quat} become homeomorphism, considering the standard topology for the involved spaces, providing double coverings for $\so(3,\mathbb{R})$ and $\so(4,\mathbb{R})$.

\section{Alternative proof of Proposition~\ref{prop:omeoSO3}
}\label{alternhomeo}
In Section~\ref{sec:relquatrotp}, we showed that the group isomorphism $\psi\colon \mathbb{H}_p^\times/\mathbb{Q}_p^\times \rightarrow \so(3,\mathbb{Q}_p)$ given in Theorem~\ref{theo:SO3quat} is a homeomorphism. The proof of Proposition~\ref{prop:omeoSO3} relies on measure-theoretical results; here we provide an alternative proof which shows more explicitly the relation between $p$-adic rotations and quaternions, depending on their reduced norm.

As already argued in the proof of Proposition~\ref{prop:omeoSO3}, the groups $\pqal^\times$ and $\so(3,\rat_p)$ are locally compact, once supplied with their $p$-adic topology. The map $\kappa_p$ is continuous, as $\kappa_p(\xi)$ is a rational function on the parameters $q_0,q_1,q_2,q_3$ of $\xi=q_0+\I q_1+\J q_2+\K q_3$, with denominator $\rn(\xi)\neq0$ for every $\xi\neq0$. Therefore, $\kappa_p$ redefined on the quotient of $\mathbb{H}_p^\times$ modulo $\ker(\kappa_p)$ is continuous, i.e. $\psi$ is continuous. We are left to prove that also the inverse map $\psi^{-1}$ is continuous, or equivalently that $\psi$ is a closed function (it maps closed subsets of $\mathbb{H}_p^\times/\mathbb{Q}_p^\times$ to closed subsets of $\so(3,\mathbb{Q}_p)$). To ease this, we want to deal with compact spaces, rather than just locally compact ones. 

Notice that $\rn:\mathbb{H}_p^\times\rightarrow \mathbb{Q}_p^\times$ is a homomorphism (it is multiplicative), as well as the quotient map $q:\mathbb{Q}_p^\times\rightarrow \mathbb{Q}_p^\times/(\mathbb{Q}_p^\times)^2$, therefore $q\circ \rn:\mathbb{H}_p^\times\rightarrow \mathbb{Q}_p^\times/(\mathbb{Q}_p^\times)^2$ is a homomorphism. Its kernel is $\ker(q\circ \rn)=\left\{\xi\in\mathbb{H}_p^\times\mid \rn(\xi)\in(\mathbb{Q}_p^\times)^2\right\}$, and $\mathbb{Q}_p^\times$ is a normal subgroup of $\ker(q\circ \rn)$ and $\mathbb{H}_p^\times$ (being its centre). It follows, by the fundamental homomorphism theorem, that there exists a unique group
homomorphism $\varphi:\mathbb{H}_p^\times/\mathbb{Q}_p^\times \rightarrow \mathbb{Q}_p^\times/(\mathbb{Q}_p^\times)^2$ such that $q\circ \rn=\varphi\circ\pi$: this map is $\varphi(\xi\mathbb{Q}_p^\times)=\rn(\xi)\mod (\mathbb{Q}_p^\times)^2$. In fact, given two distinct representatives of the same class, i.e. $\nu\neq\xi$ such that $\nu\mathbb{Q}_p^\times=\xi\mathbb{Q}_p^\times$, we have $\nu\in\xi\mathbb{Q}_p^\times$ and hence
$\varphi(\nu\mathbb{Q}_p^\times)=\rn(\nu)\mod (\mathbb{Q}_p^\times)^2=\rn(\xi)\mod (\mathbb{Q}_p^\times)^2 = \varphi(\xi\mathbb{Q}_p^\times)$.

\begin{center}
\begin{tikzcd}
& \mathbb{H}_p^\times \arrow[r, "\kappa_p", two heads] \arrow[d, "\pi"', two heads] \arrow[ld, "\rn"'] & {\so(3,\mathbb{Q}_p)} \\
\mathbb{Q}_p^\times \arrow[rd, "q"'] & \mathbb{H}_p^\times/\mathbb{Q}_p^\times \arrow[d, "\varphi", dashed, two heads] \arrow[ru, "\psi"']     &    \\
& \mathbb{Q}_p^\times/(\mathbb{Q}_p^\times)^2   &                 
\end{tikzcd}
\end{center}
Actually, by the isomorphism theorem, we have $\mathbb{H}_p^\times/\ker(q\circ\rn)\simeq \mathbb{Q}_p^\times/(\mathbb{Q}_p^\times)^2$.
For $\ker(q\circ\rn)$ is stable under multiplication by scalars in $\mathbb{Q}_p^\times$, we have the induced quotient $\ker(q\circ\rn)/\mathbb{Q}_p^\times$, such that $(\mathbb{H}_p^\times/\mathbb{Q}_p^\times) / (\ker(q\circ\rn)/\mathbb{Q}_p^\times) \simeq  \mathbb{H}_p^\times/\ker(q\circ\rn) \simeq \mathbb{Q}_p^\times/(\mathbb{Q}_p^\times)^2$. We observe that $\xi$ and $\lambda\xi$ have same image w.r.t.\ $\varphi\circ \pi$ (or equivalently $q\circ\rn$), for every $\xi\in\mathbb{H}_p^\times,\,\lambda\in\mathbb{Q}_p^\times$. Therefore, to have a surjective map onto $\mathbb{Q}_p^\times/(\mathbb{Q}_p^\times)^2$, once chosen a representative for each equivalence class of this quotient group, it is enough to consider the restriction of $\mathbb{H}_p^\times$ to the set of quaternions having reduced norm exactly equal to those representatives. A similar argument applies to $\mathbb{H}_p^\times / \mathbb{Q}_p^\times$. We recall~\cite{serre2012course, Cassels} that 
\begin{equation}
\mathbb{Q}_p^\times/(\mathbb{Q}_p^\times)^2\cong\left\{\begin{aligned}
 & \{1,u,p,up\}\cong
 \mathbb{Z}/2\mathbb{Z}\times \mathbb{Z}/2\mathbb{Z},\quad p>2,\\
 & \{\pm1,\pm2,\pm5,\pm10\}\cong \mathbb{Z}/2\mathbb{Z}\times \mathbb{Z}/2\mathbb{Z}\times \mathbb{Z}/2\mathbb{Z},\quad p=2,\end{aligned}\right.
\end{equation}
where $u\in\mathbb{U}_p$ is a non-squared invertible $p$-adic integer. We thus define $S(\epsilon)\coloneqq\{\xi\in\mathbb{H}_p^\times\mid \rn(\xi)=\epsilon\}$, by varying $\epsilon$ in the set of $p$-adic integer representatives of the equivalence classes of $\mathbb{Q}_p^\times/(\mathbb{Q}_p^\times)^2$, i.e. for $\epsilon=1,u,p,up$ when $p>2$ and for $\epsilon=\pm1,\pm2,\pm5,\pm10$ when $p=2$. Now we have the following diagram where, by abuse of notation, we also denote by $\pi$ and $\varphi$ the homonyms maps redefined on $\bigcup_\epsilon S(\epsilon)$ and $\left(\bigcup_\epsilon S(\epsilon)\right)/\mathbb{Q}_p^\times$ respectively, and where the injective maps are simply the (closed, continuous) canonical embeddings of $\bigcup_\epsilon S(\epsilon)$ and $\left(\bigcup_\epsilon S(\epsilon)\right)/\mathbb{Q}_p^\times$ in $\mathbb{H}_p^\times$ and $\mathbb{H}_p^\times/\mathbb{Q}_p^\times$ respectively.
\begin{center}
    \begin{tikzcd}
\bigcup_\epsilon S(\epsilon) \arrow[r, hook] \arrow[d, "\pi"', two heads]  & \mathbb{H}_p^\times \arrow[r, "\kappa_p", two heads] \arrow[d, "\pi"', two heads] & {\so(3,\mathbb{Q}_p)} \\
\left(\bigcup_\epsilon S(\epsilon)\right)/\mathbb{Q}_p^\times \arrow[r, hook] \arrow[rd, "\varphi"'] & \mathbb{H}_p^\times/\mathbb{Q}_p^\times \arrow[ru, "\psi"'] \arrow[d, "\varphi"]  &  \\
& \mathbb{Q}_p^\times/(\mathbb{Q}_p^\times)^2  &                     \end{tikzcd}
\end{center}

The sets $S(\epsilon)$ are pairwise disjoint. Moreover, it can be shown that $S(\epsilon)\subset \left(p^{-1}\mathbb{Z}_p\right)^4$ is compact, in a similar fashion to the proof of Eq.~\eqref{eq:inclsSO4SO3SO2SLint} (see Theorem 5 in~\cite{our1st}). As a consequence, $\bigcup_\epsilon S(\epsilon)$ is a compact subset of $\mathbb{H}_p^\times$, being the finite union of compact sets. 
Now that we can consider just the compact subspace $\bigcup_\epsilon S(\epsilon)$ of $\mathbb{H}_p^\times$, the proof of the fact that $\psi$ is closed is straightforward. Consider a closed subset $C$ in $\left(\bigcup_\epsilon S(\epsilon)\right)/\mathbb{Q}_p^\times\subset \mathbb{H}_p^\times/\mathbb{Q}_p^\times$. Its preimage $\pi^{-1}(C)\subset\bigcup_\epsilon S(\epsilon)\subset \mathbb{H}_p^\times$ is closed, since $\pi$ is continuous. Closed subsets of the compact $\bigcup_\epsilon S(\epsilon)$ are compact, in particular $\pi^{-1}(C)$ is compact. The map $\kappa_p^\prime\coloneqq \kappa_p|_{\bigcup_\epsilon S(\epsilon)}:\bigcup_\epsilon S(\epsilon)\rightarrow \so(3,\mathbb{Q}_p)$ is continuous, as a restriction of the continuous map $\kappa_p$. The continuous image $\kappa_p^\prime\left(\pi^{-1}(C)\right)$ of the compact set $\pi^{-1}(C)$ is compact. In turn, $\kappa_p^\prime\left(\pi^{-1}(C)\right)$ is closed, being a compact subset of the compact Hausdorff group $\so(3,\mathbb{Q}_p)$. This proves that $\kappa_p^\prime$ is a closed map. Finally, this implies that $\psi$ is closed, and hence it is a homeomorphism.



\begin{thebibliography}{99}


\bibitem{brekke1993}
L. Brekke and P.G.O. Freund, ``$p$-Adic numbers in physics'', \textit{Phys.\ Rep.} {\bf 233}, pp.\ 1--66 (1993).


\bibitem{volovich1987}
I.V. Volovich, ``$p$-Adic string'',  \textit{Class.\ Quant.\ Grav.} {\bf 4}, pp.\ L83--L87 (1987).


\bibitem{volovich87}
I.V. Volovich, ``$p$-adic space-time and string theory'', \textit{Theor.\ Math.\ Phys.} {\bf 71}, pp.\ 574--576 (1987).


\bibitem{khrennikov1991pp}
A.Y. Khrennikov, ``Real non-Archimedean structure of spacetime'', \textit{Theor. Math. Phys.} {\bf 86}, pp.\ 121--130 (1991).


\bibitem{arefeva1991}
I.Ya. Aref'eva and P.H. Frampton,  ``Beyond Planck energy to non-Archimedean geometry'', \textit{Mod.\ Phys.\ Lett.} {\bf A6}, pp.\ 313--316 (1991).


\bibitem{vladimirov1994}
V.S. Vladimirov, I.V. Volovich and E.I. Zelenov, \textit{$p$-adic Analysis and Mathematical Physics}, Series on Soviet and East European Mathematics {\bf 1}, World Scientific, 
1994.


\bibitem{vladimirov1989}
V.S. Vladimirov and I.V. Volovich, ``$p$-adic quantum mechanics'', \textit{Comm.\ Math.\ Phys.} {\bf 123}, pp.\ 659--676  (1989).


 \bibitem{ruelle1989}
  P. Ruelle, E. Thiran, D. Verstegen and J. Weyers, ``Quantum mechanics on $p$-adic fields'', \textit{J.\ Math.\ Phys.} {\bf 30}, pp.\ 2854--2874  (1989).


\bibitem{meurice1989}
Y. Meurice, ``Quantum mechanics with $p$-adic numbers'', \textit{Int.\ J.\ Mod.\ Phys.\ A} {\bf 4}, pp.\ 5133--5147 (1989).


\bibitem{vladimirov89p}
V.S. Vladimirov and I.V. Volovich, ``A vacuum state in $p$-adic quantum mechanics'', \textit{Phys.\ Lett.\ B} {\bf 217}, pp.\ 411--414 (1989).


\bibitem{albeverio97}
S. Albeverio, R. Cianci and A. Khrennikov, ``On the spectrum of the $p$-adic position operator'', \textit{J.\ Phys.\ A: Math.\ Gen.} {\bf 30}, pp.\ 881--889 (1997).


\bibitem{vladimirov90}
V.S. Vladimirov, I.V. Volovich and E.I. Zelenov, ``Spectral theory in $p$-adic quantum mehcanics and representation theory'', \textit{Soviet.\ Math.\ Dokl.} {\bf 41}, pp.\ 40--44 (1990).


\bibitem{zelenov89}
E.I. Zelenov, ``$p$-Adic quantum mechanics for $p=2$'', \textit{Theoret.\ and Math.\ Phys.} {\bf 80}, pp.\ 848--856 (1989).


\bibitem{khrennikov2013}
A.Y. Khrennikov, \textit{$p$-adic Valued Distributions in Mathematical Physics}, Mathematics and Its Applications {\bf 309}, Kluwer Academic Publishers: Alphen aan den Rijn, 
1994.


\bibitem{albeverio1996}
 S. Albeverio, A.Y. Khrennikov, ``$p$-adic Hilbert space representation of quantum systems with an infinite number
of degrees of freedom'', \textit{Int.\ J.\ Mod.\ Phys.\ B} {\bf 10}, pp.\ 1665--1673 (1996).


\bibitem{khrennikov1991p}
A.Y. Khrennikov, ``$p$-adic quantum mechanics with $p$-adic valued functions'', \textit{J.\ Math.\ Phys.} {\bf 32}, pp.\ 932--937 (1991).


\bibitem{volovich2010}
 I.V. Volovich, ``Number theory as the ultimate physical theory'' \textit{$p$-Adic Numbers Ultr.\ Anal.\ Appl.} {\bf 2}, pp.\ 77--87 (2010).


\bibitem{khrennikov1990}
A.Y. Khrennikov, ``Mathematical methods of non-Archimedean physics'', \textit{Russ.\ Math.\ Surv.} {\bf 45}, pp.\ 87--125 (1990).


\bibitem{dragovich2017}
B. Dragovich, A.Y. Khrennikov, S.V. Kozyrev, I.V. Volovich and E.I. Zelenov, ``$p$-Adic mathematical physics: the first 30 years'', \textit{$P$-Adic Num.\ Ultrametr.\ Anal.\ Appl.} {\bf 9}, pp.\ 87--121 (2017).


\bibitem{parisi1988}
G. Parisi, ``On $p$-adic functional integrals'', \textit{Mod.\ Phys.\ Lett.\ A} {\bf 3}, pp.\ 639--643 (1988).


\bibitem{Khrennikov1990p}
A.Y. Khrennikov, ``Representation of second quantization over non-Archimedean number fields'', \textit{Dokl.\ Akad.\ Nauk.\ SSSR} {\bf 314}, pp.\ 1380--1384 (1990).


\bibitem{albeverio2009}
S. Albeverio, R. Cianci and A.Y. Khrennikov, ``$p$-adic valued quantization'', \textit{$p$-Adic Numbers Ultr.\ Anal.\ Appl.} {\bf 1}, pp.\ 91--104 (2009).


\bibitem{khrennikov1993p}
A.Y. Khrennikov, ``Statistical interpretation of $p$-adic valued quantum field theory'', \textit{Dokl.\ Akad.\ Nauk.\ SSSR} {\bf 328}, pp.\ 46--49 (1993).


\bibitem{albeverio1997p}
S. Albeverio, A.Y. Khrennikov and R. Cianci, ``A representation of quantum field Hamiltonian in a $p$-adic Hilbert space'', \textit{Teor.\ Mat.\ Fiz.} {\bf 112}, pp.\ 355--374 (1997).


\bibitem{mezard1984}
M. M\'ezard, G. Parisi, N. Sourlas, G. Toulouse, and M. Virasoro, ``Nature of the spin-glass phase'', \textit{Phys.\ Rev.\ Lett.} {\bf 52}, pp.\ 1156--1159 (1984).


\bibitem{rammal1986}
 R. Rammal, G.\ Toulouse and M.A. Virasoro, ``Ultrametricity for physicists'', \textit{Rev.\ Mod.\ Phys.} {\bf 58}, pp.\ 765--821 (1986).


\bibitem{PS00}
G. Parisi and N. Sourlas, ``$p$-adic numbers and replica symmetry breaking'', \textit{Eur.\ Phys.\ J.\ B Condens.\ Matter Complex Syst.} {\bf 14}, pp.\ 535--542 (2000).


\bibitem{anashin2009}
 V. Anashin, A. Khrennikov, \textit{Applied Algebraic Dynamics}, De Gruyter Expositions in Mathematics {\bf 49}, Walter de Gruyter \& Co., 
 2009.


\bibitem{zelenov2020}
E.I. Zelenov, ``Entropy Gain in $p$-Adic Quantum Channels'' \textit{Phys.\ Part.\ Nuclei} {\bf 51}, pp.\ 485--488 (2020).


\bibitem{aniello2022}
P. Aniello, S. Mancini and V. Parisi, ``Trace class operators and states in $p$-adic quantum mechanics'', \textit{J. Math Phys.} {\bf 64}, 053506 (2023).


\bibitem{our2}
I. Svampa, S. Mancini and A. Winter, ``An approach to $p$-adic qubits from irreducible representations of $\so(3)_p$'', \textit{J.\ Math.\ Phys.} \textbf{63}, 072202 (2022).


\bibitem{zelenov2022}
E.I. Zelenov, ``Coherent States of the $p$-adic Heisenberg Group, Heterodyne Measurements and Entropy Uncertainty Relation'' [Conference Presentation] New Trends in Mathematical Physics $2022$, Stekl. Math. Inst. Moscow. Available online: \url{https://www.mathnet.ru/php/presentation.phtml?presentid=36592&option_lang=eng}
(accessed on 7 November 2022).


\bibitem{aniello2023}
P. Aniello, S. Mancini and V. Parisi, ``A $p$-Adic Model of Quantum States and the $p$-Adic Qubit'', \textit{Entropy} \textbf{25}, 86 (2023). 


\bibitem{serre2012course}
J.-P. Serre, \textit{A course in arithmetic}, Graduate Texts in Mathematics \textbf{7}, Springer, 
2012.


\bibitem{our1st}
S. Di Martino, S. Mancini, M. Pigliapochi, I. Svampa, and A. Winter,
``Geometry of the $p$-adic special orthogonal group $\so(3)_p$'', \textit{Lobachevskii J.\ Math.} \textbf{44}(6), pp. 2135--2159 (2023).


\bibitem{folland2016course}
G.B. Folland, \textit{A course in abstract harmonic analysis}, Studies in advanced mathematics $\mathbf{29}$, CRC Taylor and Francis, 
2016.


\bibitem{AnielloSDP}
P. Aniello, G. Cassinelli, E. De Vito and A. Levrero,
``Square-integrability of induced representations of semidirect products'',
\textit{Rev. Math. Phys.} \textbf{10}, pp.~{301--313} (1998).


\bibitem{AnielloPR}
P. Aniello,
``Square integrable projective representations and square integrable representations modulo a relatively central subgroup'',
\textit{Int. J. Geom. Meth. Mod. Phys.} \textbf{3}, pp.~{233--267} (2006).


\bibitem{AnielloSP}
P. Aniello,
``Star products: a group-theoretical point of view'',
\textit{J. Phys. A: Math. Theor.} \textbf{42}, 475210 (2009).


\bibitem{AnielloSIR}
P. Aniello, ``Square integrable representations, an invaluable tool'', contained in \textit{Coherent States and Their Applications: a Contemporary Panorama} edited by J.-P.\ Antoine, F.\ Bagarello and J.-P.\ Gazeau, Springer Proceedings in Physics {\bf 205}, pp.~{17--40}, Springer, 
2018.


\bibitem{FollandHAPS}
G.B. Folland,
\textit{Harmonic Analysis in Phase Space}, Annals of Mathematics Studies {\bf 122}, Princeton University Press, 
1989.

\bibitem{reiter2000}
H. Reiter, J.D. Stegeman, \textit{Classical Harmonic Analysis and Locally Compact Groups}, London Mathematical Society Monographs {\bf 22}, 
Oxford Science Publications, 
2000. 

\bibitem{doran88}
J.M.G. Fell and R.S. Doran, \textit{Representations of $\ast$-Algebras, Locally Compact Groups, and Banach $\ast$-Algebraic Bundles}, vol. I, Academic Press, 
1988.

\bibitem{hewitt1979}
E. Hewitt and K.A. Ross, \textit{Abstract Harmonic Analysis I}, 
Grundlehren der mathematischen Wissenschaften \textbf{115}, Springer, 
1979.

\bibitem{folland99}
G.B. Folland, \textit{Real Analysis}, 
Pure and Applied Mathematics: A Wiley Series of Texts, Monographs, and Tracts, John Wiley, 
1999.
 
\bibitem{gaal}
S.A. Gaal, \textit{Linear Analysis and Representation Theory}, Die Grundlehren der mathematischen Wissenschaften in Einzeldarstellungen \textbf{198}, Springer, 
1973.

\bibitem{varadarajan}
V.S. Varadarajan, \textit{Geometry of Quantum Theory}, 
Springer, 
2007.

\bibitem{knapp05}
A.W. Knapp, \textit{Advanced Real Analysis}, Cornerstones, Birkh\"auser, 
2005.

\bibitem{schneider2011p}
P. Schneider, \textit{$p$-Adic Lie Groups}, Grundlehren der mathematischen Wissenschaften $\mathbf{344}$, Springer, 
2011.

\bibitem{glockner}
H. Gl\"ockner,  ``Lectures on Lie groups over Local Fields'' (see \url{arXiv:0804.2234v5}, 2016), contained in \textit{New Directions in Locally Compact Groups} edited by P.-E. Caprace and N. Monod, London Mathematical Society Lecture Notes Series \textbf{447}, pp.~{37--72}, Cambridge University Press, 
2018.

\bibitem{serre2009}
J.-P. Serre, \textit{Lie algebras and Lie groups}, 1964 lectures given at Harvard University, Springer, 
2009.

\bibitem{igusa2000}
J. Igusa, \textit{An introduction to the theory of local zeta functions}, Studies in Advanced Mathematics $\mathbf{14}$, American Mathematical Society International Press, 2000.

\bibitem{kechris95set}
A.S. Kechris, \textit{Classical Descriptive Set Theory}, Graduate Texts in Mathematics \textbf{156}, Springer, 
1995.

\bibitem{engelking89}
R. Engelking, \textit{General Topology},  Sigma Series in Pure Mathematics \textbf{6}, Heldermann Verlag, 
1989.

\bibitem{Cassels}
J.W.S. Cassels, \textit{Rational Quadratic Forms}, LMS Monographs \textbf{13}, Courier Dover Publications, 2008.

\bibitem{rooij78}
A.C.M. van Rooij, \textit{Non-Archimedean functional analysis}, Monographs and textbooks in pure and applied mathematics \textbf{51},  Marcel Dekker, 
1978.

\bibitem{knappLie}
A.W. Knapp, \textit{Lie Groups Beyond an Introduction}, 
Progress in Mathematics \textbf{140}, Birkh\"auser, 
2002.

\bibitem{helgason01}
S. Helgason, \textit{Differential Geometry, Lie Groups, and Symmetric Spaces}, GMS \textbf{34}, American Mathematical Society, 
2001.

\bibitem{voight2021}
J. Voight, \textit{Quaternion Algebras}, Graduate Texts in Mathematics \textbf{288}, Springer, 
2021.

\bibitem{kochubei2001pseudo}
A.N. Kochubei, \textit{Pseudo-differential equations and stochastics over non-Archimedean fields}, A Series of Monographs and Textbooks, CRC Taylor and Francis, 
2001.

\bibitem{Lam}
T.Y. Lam, \textit{Introduction to Quadratic Forms over Fields}, Graduate Studies in Mathematics \textbf{67}, American Mathematical Society, 
2005.

\bibitem{jacobson2012basic}
N. Jacobson, \textit{Basic algebra I}, 
Courier Dover Publications, 
2012.

\bibitem{smirnov2011linear}
V.I. Smirnov, \textit{Linear algebra and group theory}, 
Courier Dover Publications, 
2011.
\end{thebibliography}
\end{document}